\documentclass[journal]{IEEEtran}
\usepackage[utf8]{inputenc}
\usepackage{amsmath}
\usepackage{amsthm}
\usepackage{amsfonts}
\usepackage{mathtools}
\usepackage{tikz}
\usetikzlibrary{decorations.pathreplacing}

\IEEEoverridecommandlockouts

\newcommand\blankfootnote[1]{%
  \let\svthefootnote\thefootnote%
  \let\thefootnote\relax\footnotetext{#1}%
  \let\thefootnote\svthefootnote%
}

\newcommand{\lemmaconst}{\delta}
\newcommand{\channelpmf}{q}
\newcommand{\codebookpmf}{p}
\newcommand{\generalpmf}{r}
\newcommand{\generalrvOne}{A}
\newcommand{\generalrvOneValue}{a}
\newcommand{\generalrvOneAlph}{\mathcal{A}}
\newcommand{\generalrvTwo}{B}
\newcommand{\generalrvTwoValue}{b}
\newcommand{\generalrvTwoAlph}{\mathcal{B}}
\newcommand{\generalrvThree}{C}
\newcommand{\generalrvThreeValue}{c}
\newcommand{\generalrvThreeAlph}{\mathcal{C}}
\newcommand{\generalpdistOne}{P}
\newcommand{\generalpmfOne}{p}
\newcommand{\generalpdistTwo}{Q}
\newcommand{\generalpmfTwo}{q}
\newcommand{\generalSecondMoment}{\sigma}
\newcommand{\generalThirdMoment}{\rho}
\newcommand{\generalcdf}{F}
\newcommand{\generalintegrand}{x}
\newcommand{\generalSummationIndex}{k}
\newcommand{\generalLimitIndex}{k}

\newcommand{\generalSummationBound}{n}
\newcommand{\generalpmfset}{\mathcal{P}}
\newcommand{\generaldimension}{n}
\newcommand{\generaldimensionIndex}{k}
\newcommand{\generalfunction}{f}
\newcommand{\generalpmeasure}{P}
\newcommand{\coverNumber}{\chi}
\newcommand{\codebookRate}{R}
\newcommand{\codebookRateOne}{R_1}
\newcommand{\codebookRateTwo}{R_2}
\newcommand{\channelInOne}{X}
\newcommand{\channelInOneAlph}{\mathcal{X}}
\newcommand{\channelInOneAlphElement}{x}
\newcommand{\channelInTwo}{Y}
\newcommand{\channelInTwoAlph}{\mathcal{Y}}
\newcommand{\channelInTwoAlphElement}{y}
\newcommand{\channelOut}{Z}
\newcommand{\channelOutAlph}{\mathcal{Z}}
\newcommand{\channelOutAlphElement}{z}
\newcommand{\channel}{\mathcal{W}}
\newcommand{\channelWiretapper}{{\channel_\mathrm{tap}}}
\newcommand{\channelLegit}{{\channel_\mathrm{legit}}}
\newcommand{\channelOutAlphWiretapper}{\channelOutAlph_\mathrm{tap}}
\newcommand{\channelOutAlphElementWiretapper}{\channelOutAlphElement}
\newcommand{\channelOutAlphLegit}{\channelOutAlph_\mathrm{legit}}

\newcommand{\channelOutWiretapper}{\channelOut_\mathrm{tap}}
\newcommand{\channelOutLegit}{\channelOut_\mathrm{legit}}
\newcommand{\alphSubset}{A}
\newcommand{\codebook}{\mathcal{C}}
\newcommand{\codebookOne}{\mathcal{C}_1}
\newcommand{\codebookTwo}{\mathcal{C}_2}
\newcommand{\codebookOneWord}[1]{C_1(#1)}
\newcommand{\codebookTwoWord}[1]{C_2(#1)}
\newcommand{\codebookWord}[1]{C(#1)}
\newcommand{\codebookSet}{\mathbb{C}}
\newcommand{\codewordIndex}{m}
\newcommand{\codebookBlocklength}{n}
\newcommand{\blockIndex}{k}
\newcommand{\txIndex}{k}
\newcommand{\mutualInformation}[2]{I(#1;#2)}
\newcommand{\mutualInformationConditional}[3]{I(#1;#2|#3)}
\newcommand{\entropy}[1]{H(#1)}
\newcommand{\entropyConditional}[2]{H(#1 | #2)}
\newcommand{\finalconstOne}{\gamma_1}
\newcommand{\finalconstTwo}{\gamma_2}
\newcommand{\totalvariationBigg}[1]{\Bigg\lVert #1 \Bigg\rVert_\mathrm{TV}}
\newcommand{\totalvariation}[1]{\lVert #1 \rVert_\mathrm{TV}}
\newcommand{\totalvariationlr}[1]{\left\lVert #1 \right\rVert_\mathrm{TV}}
\newcommand{\absolute}[1]{\left\lvert #1 \right\rvert}
\newcommand{\positive}[1]{\left[ #1 \right]^+}
\newcommand{\renyiParam}{\alpha}
\newcommand{\proofconstantOne}{{\beta}}

\newcommand{\informationDensity}[2]{i({#1};{#2})}
\newcommand{\informationDensityConditional}[3]{i({#1};{#2} | {#3})}
\newcommand{\renyidiv}[3]{D_{#1}\left({#2} || {#3}\right)}

\newcommand{\Expectation}{\mathbb{E}}
\newcommand{\Probability}{\mathbb{P}}
\newcommand{\indicator}[1]{1_{#1}}
\newcommand{\cardinality}[1]{\lvert #1 \rvert}
\newcommand{\typicalityParam}{\varepsilon}
\newcommand{\typicalSetIndex}[3]{\mathcal{T}_{#3,#1}^{#2}}
\newcommand{\typicalSet}[2]{\mathcal{T}_{#1}^{#2}}
\newcommand{\lemmaexpectation}{\mu}
\newcommand{\channelDispersion}[1]{V_{#1}}
\newcommand{\channelThirdMoment}[1]{\rho_{#1}}
\newcommand{\normalcdfComplement}{\mathcal{Q}}
\newcommand{\normalcdf}{\Phi}
\newcommand{\normalcdfComplementInverse}{\mathcal{Q}^{-1}}
\newcommand{\secondOrderParamC}{c}
\newcommand{\secondOrderParamD}{d}
\newcommand{\secondOrderAtypicalProbability}[1]{\mu_{#1}}
\newcommand{\totvarAtypicalOne}{P_{\mathrm{atyp}, 1}}
\newcommand{\totvarAtypicalTwo}{P_{\mathrm{atyp}, 2}}
\newcommand{\totvarTypical}[1]{P_{\mathrm{typ}}({#1})}
\newcommand{\totvarTypicalOne}[2]{P_{\mathrm{typ}, 1}({#1},{#2})}
\newcommand{\codebookDecoder}{d}
\newcommand{\errorprob}{\mathcal{E}}
\newcommand{\indexForTypicalSet}{k}
\newcommand{\messageRV}{M}
\newcommand{\messageAlphabet}{\mathcal{M}}
\newcommand{\nPartitions}{t}
\newcommand{\nPartitionElements}{\ell}
\newcommand{\indexPartitions}{k}
\newcommand{\messageAlphabetElement}{m}
\newcommand{\codebookRandRate}{{L}}
\newcommand{\codebookRandRateOne}{{L_1}}
\newcommand{\codebookRandRateTwo}{{L_2}}
\newcommand{\codebookRandRateOneLower}[1]{{\underline{L}_{1,{#1}}}}
\newcommand{\codebookRandRateOneUpper}[1]{{\overline{L}_{1,{#1}}}}
\newcommand{\codebookRandRateTwoLower}[1]{{\underline{L}_{2,{#1}}}}
\newcommand{\codebookRandRateTwoUpper}[1]{{\overline{L}_{2,{#1}}}}
\newcommand{\randomnessIndex}{\ell}

\newcommand{\partition}{{\Pi}}
\newcommand{\wiretapperDecoder}{{f}}
\newcommand{\wiretapperGuesser}{{g}}
\newcommand{\semanticTotvarTerm}{{P_1}}
\newcommand{\semanticErrorTermOne}{{P_2}}
\newcommand{\semanticErrorTermTwo}{{P_3}}
\newcommand{\timeSharingRV}{V}
\newcommand{\timeSharingAlph}{\mathcal{V}}
\newcommand{\timeSharingAlphElement}{v}
\newcommand{\capacityRegion}[2]{\mathcal{S}_{#1, #2}}
\newcommand{\convexityParam}{\lambda}
\newcommand{\reals}{\mathbb{R}}

\newtheorem{theorem}{Theorem}

\newtheorem{lemma}{Lemma}
\newtheorem{cor}{Corollary}

\newtheorem{remark}{Remark}

\newtheorem{example}{Example}

\title{The MAC Resolvability Region, Semantic Security and Its Operational Implications}
\author{
Matthias Frey, Igor Bjelaković and Sławomir Stańczak
\\
Technische Universität Berlin
}

\allowdisplaybreaks

\begin{document}

\maketitle

\blankfootnote{
The work was supported by the German Research Foundation (DFG) under grant 
STA864/7-1 and by the German Federal Ministry of Education and Research under 
grant 16KIS0605.

This paper was presented in part at the 4th Workshop on Physical-Layer Methods for Wireless Security held during the 2017 IEEE Conference on Communications and Network Security~\cite{ConferenceVersion}.
}

\begin{abstract}
  Building upon previous work on the relation between secrecy and
  channel resolvability, we revisit a secrecy proof for the
  multiple-access channel (MAC) from the perspective of
  resolvability. In particular, we refine and extend the proof to
  obtain novel results on the second-order achievable rates. Then we
  characterize the resolvability region of the memoryless MAC by
  establishing a conceptually simple converse proof for resolvability
  that relies on the uniform continuity of Shannon's entropy with
  respect to the variational distance. We discuss the operational
  meaning of the information theoretic concept of semantic security
  from the viewpoint of Ahlswede's General Theory of Information
  Transfer and show interesting implications on the resilience against
  a large class of attacks even if no assumption can be made about the
  distributions of the transmitted messages. Finally we provide
  details on how the resolvability results can be used to construct a
  wiretap code that achieves semantic security.
\end{abstract}

\section{Introduction}
\label{section:Introduction}

With an increasing number of users and things being connected to each
other, not only the overall amount of communication increases, but
also the amount of private and personal information being
transferred. This information needs to be protected against various
attacks. In fact, for some emerging applications such as e-health
services where sensitive medical data might be transmitted using a
Body Area Network, the problem of providing secrecy guarantees is a
key issue.

As discovered by Csiszár~\cite{CsiszarSecrecy} and later more
explicitly by Bloch and Laneman~\cite{BlochStrongSecrecy} and
investigated by Yassaee and Aref \cite{YassaeeMACWiretap} for the
multiple-access case, the concept of channel resolvability can be
applied to provide such guarantees; in essence, an effective channel
is created by an encoder in such a way that the output observed by a
wiretapper is almost independent of the transmitted messages, which
renders the output useless for the wiretapper. Moreover, the concept
can also be used in point-to-point settings as a means of exploiting
channel noise to convey randomness in such a way that the distribution
observed at the receiver can be accurately controlled by the
transmitter.

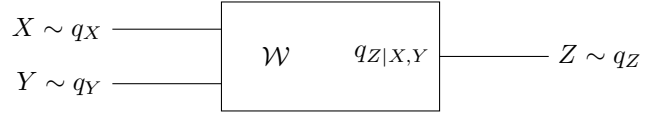
\begin{figure}
\begin{center}
\begin{tikzpicture}[scale=1.45]
\node[left]   (channelInOne)                at (   0,0.75) {$\channelInOne\sim\channelpmf_\channelInOne$};
\node[right]  (channelInOneConnect)         at (   1,0.75) {};
\node[left]   (channelInTwo)                at (   0,0.25) {$\channelInTwo\sim\channelpmf_\channelInTwo$};
\node[right]  (channelInTwoConnect)         at (   1,0.25) {};
\node[right]  (channelOut)                  at (   4, 0.5) {$\channelOut\sim\channelpmf_\channelOut$};
\node[left]   (channelOutConnect)           at (   3, 0.5) {$\channelpmf_{\channelOut | \channelInOne, \channelInTwo}$};
\node         (channel)                     at ( 1.5, 0.5) {$\channel$};

\draw (1,1) rectangle (3,0);
\draw (channelInOne) -- (channelInOneConnect);
\draw (channelInTwo) -- (channelInTwoConnect);
\draw (channelOutConnect) -- (channelOut);
\end{tikzpicture}
\end{center}
\caption{Idealized input-output model of the multiple-access channel.}
\label{figure:channel-idealized}
\end{figure}

\begin{figure}
\begin{center}
\begin{tikzpicture}[scale=1.45]
\node[draw]   (channelInOne)                at (   0,-.75) {$\codebookOne$};
\node[right]  (channelInOneConnect)         at ( 1,-.75) {};
\node[draw]   (channelInTwo)                at (   0,-1.25) {$\codebookTwo$};
\node[right]  (channelInTwoConnect)         at ( 1,-1.25) {};
\node[right]  (channelOut)                  at ( 3.3,-1) {$\channelOut^\codebookBlocklength\sim\codebookpmf_{\channelOut^\codebookBlocklength | \codebookOne, \codebookTwo}$};
\node[left]   (channelOutConnect)           at (   3,-1) {$\channelpmf_{\channelOut^\codebookBlocklength | \channelInOne^\codebookBlocklength, \channelInTwo^\codebookBlocklength}$};
\node         (channel)                     at ( 1.5,-1) {$\channel^\codebookBlocklength$};
\node[left]   (messageOne)                  at (-0.5,-.75) {$\messageRV_1$};
\node[left]   (messageTwo)                  at (-0.5,-1.25) {$\messageRV_2$};

\draw (1,-.5) rectangle (3,-1.5);
\draw (channelInOne) -- node[above]{$\channelInOne^\codebookBlocklength$} (channelInOneConnect);
\draw (channelInTwo) -- node[above]{$\channelInTwo^\codebookBlocklength$} (channelInTwoConnect);
\draw (channelOutConnect) -- (channelOut);
\draw (messageOne) -- (channelInOne);
\draw (messageTwo) -- (channelInTwo);
\end{tikzpicture}
\end{center}
\caption{Resolvability codebooks are used as an input to the MAC
  channel to create a channel output at the receiver that well
  approximtes the given ideal output distribution (i.e. the output
  looks as if  $\codebookBlocklength$ independent instances of the
  output in Figure~\ref{figure:channel-idealized} were observed).}
\label{figure:channel-resolved}
\end{figure}
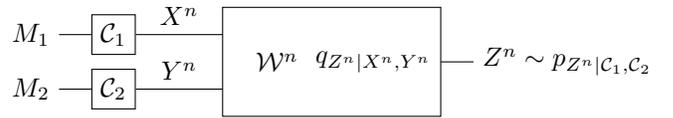

In this paper, we explore channel resolvability in a multiple-access
setting in which there is no communication between the transmitters,
yet they can control the distribution observed at the receiver in a
non-cooperative manner. More precisely, given a multiple-access
channel with an ideal or desired output distribution under some input
distribution (see Figure ~\ref{figure:channel-idealized}), the
resolvability problem for this multiple-access channel consists in
finding for each transmitter a finite sequence of codewords (called a
codebook) such that if each transmitter chooses a codeword at random
and independently of the other transmitter, then the MAC output well
approximates the given ideal output distribution. The coding setting
is illustrated in Figure~\ref{figure:channel-resolved} where each
codeword is a tuple of $\codebookBlocklength$ elements of the input
alphabet and $\codebookBlocklength$ is referred to as the block length.


The first part of the paper deals with the resolvability problem for
the memoryless MAC, while in the second part we show how the concept
of channel resolvability can be exploited to achieve strong
operational secrecy guarantees over a wiretap MAC.


\subsection{Literature}
\label{section:literature}

To the best of our knowledge, the concept of approximating a desired
output distribution over a communication channel using as little
randomness as possible at the transmitter was first introduced by
Wyner~\cite{WynerCommonInformation}, who used normalized
Kullback-Leibler divergence as a similarity measure between the actual
and the desired output distribution. 
Han and Verdú~\cite{HanApproximation} introduced the term
\emph{channel resolvability} for a similar concept; however, instead
of the Kullback-Leibler divergence, they used the variational distance
as a metric and showed, among others, the existence of a codebook that
achieves an arbitrarily small variational distance by studying the
expected variational distance of a random codebook.

A stronger result stating that the probability of drawing an
unsuitable random codebook is doubly exponentially small is due to
Cuff~\cite{CuffSoftCovering}. Related results were proposed before by
Csiszár~\cite{CsiszarSecrecy} and by
Devetak~\cite{DevetakPrivateCapacity} for the quantum setting, who
based his work on the non-commutative Chernoff
bound~\cite{AhlswedeIdentification}. Further secrecy results that are
either based on or related to the concept of channel resolvability can
be found in Hayashi~\cite{HayashiResolvability}, Bloch and Laneman
\cite{BlochStrongSecrecy}, Hou and Kramer~\cite{HouEffectiveSecrecy},
and Wiese and Boche~\cite{WieseWiretap}, who applied Devetak's
approach to a multiple-access setting. Cuff~\cite{CuffSoftCovering}
presented a result on the second-order rate; a related result was
proposed by Watanabe and
Hayashi~\cite{WantanabeSecondOrder}. Resolvability for MACs has been
explored by Steinberg~\cite{SteinbergResolvability} and later by
Oohama~\cite{OohamaConverse}. Explicit low-complexity codebooks for
the special case of symmetric MACs have been proposed by Chou, Bloch
and Kliewer~\cite{ChouLowComplexity}.

Converse results for channel resolvability have been proposed by Han
and Verdú~\cite{HanApproximation} and, for the multiple-access
channel, by Steinberg~\cite{SteinbergResolvability}. However, since
these studies consider channels with memory and allow for arbitrary
input and output distributions (in particular, they need not be
identical or independent across channel uses), these results are not
directly transferrable to the memoryless case considered in this
paper. Converse proofs for the memoryless single-user case were given
by Wyner~\cite{WynerCommonInformation} and later by
Hou~\cite{HouPhDThesis}; however, instead of the variational distance
considered in this study, they assumed the Kullback-Leibler divergence
as a similarity measure so that the results of \cite{HouPhDThesis}
cannot be applied directly either.

Security concepts from cryptography have been redefined as secrecy
concepts for the information theoretic setting by Bellare, Tessano and
Vardy in \cite{BellareSemantic} and \cite{BellareCryptographic}, where
the authors also discussed the interrelations between these concepts;
however, the proposed schemes involve some cryptographic elements,
such as shared secrets between transmitter and legitimate receiver or
shared (public) randomness. To the best of our knowledge, there are
currently no results explicitly linking the concept of semantic
security to the operational meaning in communication theory.

Wiretap codes that achieve semantic security over the Gaussian channel
have been proposed by Ling, Luzzi, Belfiore and
Stehlé~\cite{LingSemantically}. Thangaraj~\cite{ThangarajCoding}
showed how to construct semantically secure wiretap codes for the
single-user channel and already noted the link to channel
resolvability. Goldfeld, Cuff and
Permuter~\cite{GoldfeldSemantic}\cite{GoldfeldSemanticConference}
showed how channel resolvability techniques can be used to obtain
wiretap codebooks that achieve semantic security over the single-user
channel; the technique had been previously sketched by
Cuff~\cite{CuffStronger}.

Liang and Poor~\cite{LiangPoorMACConfidential}, Ekrem and
Ulukus~\cite{EkremUlukusSecrecyMAC}, as well as Liu and
Kang~\cite{LiuKangSecrecyCapacity} provided converse results for
related multiple-access wiretap channels, but these results do not
carry over in a straightforward fashion to the converse results for
the channel model we consider in this paper to study achievable
secrecy rates.

\section{Overview and Main Contributions}
\label{sec:maincontributions}

The main contributions of this study can be summarized as follows:
\begin{enumerate}
\item We characterize the resolvability region of the discrete
  memoryless multiple-access channel. Section \ref{sec:main} proves
  the direct theorem and a result on the second-order rate, while
  Section \ref{sec:converse} deals with the converse theorem.
\item In Section \ref{sec:secrecy} we show how resolvability can be
  exploited to achieve strong secrecy and semantic security, and
  discuss operational implications of semantic security.
\end{enumerate}

In the remainder of this section, we discuss our results in more
detail and explain their implications in the context of secrecy and
security. We start with a basic definition and refer the reader to
Section \ref{section:preliminaries} for further definitions.

Given a multiple-access channel
$\channel = (\channelInOneAlph, \channelInTwoAlph, \channelOutAlph,
\channelpmf_{\channelOut | \channelInOne, \channelInTwo})$ and an
output distribution $\channelpmf_\channelOut$, we call a rate pair
$(\codebookRateOne, \codebookRateTwo)$ \emph{achievable} if there are
sequences of codebooks
$(\codebook_{1,\codebookBlocklength})_{\codebookBlocklength \geq 1},
(\codebook_{2,\codebookBlocklength})_{\codebookBlocklength \geq 1}$ of
rates $\codebookRateOne$ and $\codebookRateTwo$, respectively, and
each of block length $\codebookBlocklength$ such that
\[
\lim\limits_{\codebookBlocklength \rightarrow \infty}
  \totalvariation{\codebookpmf_{\channelOut^\codebookBlocklength | \codebook_{1,\codebookBlocklength}, \codebook_{2,\codebookBlocklength}} - \channelpmf_{\channelOut^\codebookBlocklength}}
=
0\,.
\]
The \emph{resolvability region}
$\capacityRegion{\channel}{\channelpmf_\channelOut}$ is defined as the
closure of the set of all achievable rate pairs.

The main resolvability result in this work is the following characterization of $\capacityRegion{\channel}{\channelpmf_\channelOut}$.
\begin{theorem}
\label{theorem:resolvability-region}
Let $\channel = (\channelInOneAlph, \channelInTwoAlph, \channelOutAlph, \channelpmf_{\channelOut | \channelInOne, \channelInTwo})$ be a channel and $\channelpmf_\channelOut$ an output distribution. Let
\begin{alignat*}{2}
\label{capacity-theorem-input-condition}
\capacityRegion{\channel}{\channelpmf_\channelOut}'
:=
\smash{\bigcup\limits_{\codebookpmf_\timeSharingRV, \channelpmf_{\channelInOne | \timeSharingRV}, \channelpmf_{\channelInTwo | \timeSharingRV}}}
  \{(\codebookRateOne, \codebookRateTwo):
    &\mutualInformationConditional{\channelInOne, \channelInTwo}{\channelOut}{\timeSharingRV} &\leq& \codebookRateOne + \codebookRateTwo
    \\
    \wedge~
    &\mutualInformationConditional{\channelInOne}{\channelOut}{\timeSharingRV} &\leq& \codebookRateOne
    \\
    \wedge~
    &\mutualInformationConditional{\channelInTwo}{\channelOut}{\timeSharingRV} &\leq& \codebookRateTwo
  \},
\end{alignat*}
where $\codebookpmf_\timeSharingRV$ ranges over all probability mass functions on the alphabet $\timeSharingAlph := \{1, \dots, \cardinality{\channelOutAlph} + 3\}$ and $\channelpmf_{\channelInOne | \timeSharingRV}$, $\channelpmf_{\channelInTwo | \timeSharingRV}$ range over all conditional probability distributions such that, for every $\timeSharingAlphElement \in \timeSharingAlph$ and $\channelOutAlphElement \in \channelOutAlph$,
\[
\channelpmf_\channelOut(\channelOutAlphElement)
=
\sum\limits_{\channelInOneAlphElement \in \channelInOneAlph}
\sum\limits_{\channelInTwoAlphElement \in \channelInTwoAlph}
  \channelpmf_{\channelInOne | \timeSharingRV}(\channelInOneAlphElement | \timeSharingAlphElement)
  \channelpmf_{\channelInTwo | \timeSharingRV}(\channelInTwoAlphElement | \timeSharingAlphElement)
  \channelpmf_{\channelOut | \channelInOne, \channelInTwo}(\channelOutAlphElement | \channelInOneAlphElement, \channelInTwoAlphElement),
\]
and the mutual information is computed with respect to the induced joint probability mass function. Then
\[
\capacityRegion{\channel}{\channelpmf_\channelOut}' = \capacityRegion{\channel}{\channelpmf_\channelOut}.
\]
\end{theorem}
\begin{figure}
\begin{center}
\begin{tikzpicture}[scale=2]
\draw [<->,thick] (0,2.5) node (yaxis) [above] {$\codebookRateTwo$}
    |- (2.5,0) node (xaxis) [right] {$\codebookRateOne$};
\coordinate       (i1lesser)      at (0.8,0);
\coordinate       (i1greater)     at (1.8,0);
\draw[thick] (i1lesser)  -- +(0,-2pt) node[below] {$\mutualInformationConditional{\channelInOne}{\channelOut}{\timeSharingRV}$};
\draw[thick] (i1greater) -- +(0,-2pt) node[below] {$\mutualInformationConditional{\channelInOne}{\channelOut}{\timeSharingRV,\channelInTwo}$};
\coordinate       (i2lesser)      at (0,0.8);
\coordinate       (i2greater)     at (0,1.8);
\draw[thick] (i2lesser)   -- +(-2pt,0) node[left] {$\mutualInformationConditional{\channelInTwo}{\channelOut}{\timeSharingRV}$};
\draw[thick] (i2greater) -- +(-2pt,0) node[left] {$\mutualInformationConditional{\channelInTwo}{\channelOut}{\timeSharingRV,\channelInOne}$};
\coordinate        (corner1)      at (1.8,0.8)    {};
\coordinate        (corner2)      at (0.8,1.8)    {};
\coordinate        (end1)         at (2.5,0.8)    {};
\coordinate        (end2)         at (0.8,2.5)    {};
\draw[dashed]      (i1greater) -- (corner1) -- (end1);
\draw[dashed]      (i2greater) -- (corner2) -- (end2);
\draw[dashed]      (corner1) -- (corner2);
\node at (1,1) {
  \begin{tabular}{c}
    capacity \\
    region
  \end{tabular}
};
\node at (1.7,1.5) {
  \begin{tabular}{c}
    resolvability \\
    region
  \end{tabular}
};
\end{tikzpicture}
\end{center}
\caption{The MAC capacity region and the resolvability region as characterized in Theorem~\ref{theorem:resolvability-region}.}
\label{figure:resolvability-region}
\end{figure}
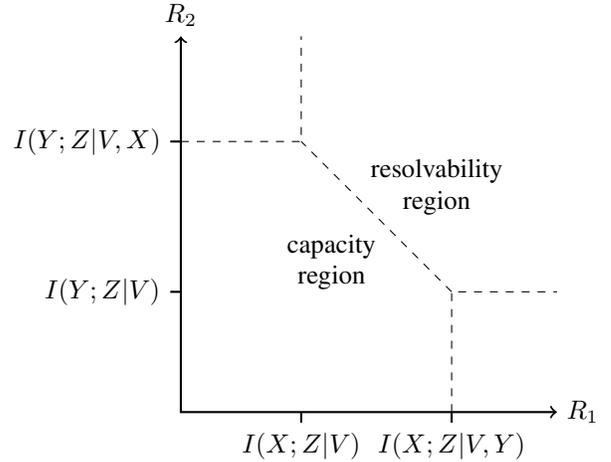

Figure~\ref{figure:resolvability-region} illustrates the resolvability
region as characterized in this theorem and the well-known capacity
region of the MAC.

Theorem~\ref{converse-theorem} in Section~\ref{sec:converse} yields
the converse part
$\capacityRegion{\channel}{\channelpmf_\channelOut}' \supseteq
\capacityRegion{\channel}{\channelpmf_\channelOut}$ of the above
theorem. The proof is remarkably simple and mainly relies on the fact
that entropy is continuous with respect to the variational distance.

The direct part
$\capacityRegion{\channel}{\channelpmf_\channelOut}' \subseteq
\capacityRegion{\channel}{\channelpmf_\channelOut}$ follows from a
stronger statement proven in Section~\ref{sec:main}. In particular,
the statement shows that not only does a sequence of codebook pairs
achieving the rates exist, but if they are constructed randomly, the
probability of drawing a bad pair of codebooks vanishes doubly
exponentially with increasing block length. Precisely, we show the
following. 

\begin{theorem}
\label{theorem:soft-covering-two-transmitters-convex}
Let
$\channel = (\channelInOneAlph, \channelInTwoAlph, \channelOutAlph,
\channelpmf_{\channelOut | \channelInOne, \channelInTwo})$ be a given
channel; let $\channelpmf_\channelInOne$ and $\channelpmf_\channelInTwo$
be input distributions,
\begin{align}
\label{theorem:soft-covering-two-transmitters-convex-rateone}
\codebookRateOne &> \mutualInformation{\channelInOne}{\channelOut}
\\
\label{theorem:soft-covering-two-transmitters-convex-ratetwo}
\codebookRateTwo &> \mutualInformation{\channelInTwo}{\channelOut}
\\
\label{theorem:soft-covering-two-transmitters-convex-sumrates}
\codebookRateOne + \codebookRateTwo &> \mutualInformation{\channelInOne, \channelInTwo}{\channelOut}.
\end{align}
Then there exist $\finalconstOne, \finalconstTwo > 0$ such that for
sufficiently large block length $\codebookBlocklength$, the codebook
distributions of block length $\codebookBlocklength$ and rates
$\codebookRateOne$ and $\codebookRateTwo$ satisfy
\begin{multline*}
\Probability_{\codebookOne, \codebookTwo} \left(
  \totalvariation{
    \codebookpmf_{\channelOut^\codebookBlocklength | \codebookOne, \codebookTwo} - \channelpmf_{\channelOut^\codebookBlocklength}
  }
  >
  \exp(-\finalconstOne\codebookBlocklength)
\right)
\\
\leq
\exp\left(-\exp\left(\finalconstTwo\codebookBlocklength\right)\right).
\end{multline*}
\end{theorem}

To prove this theorem, we revisit the proof in~\cite{WieseWiretap},
thereby focusing on channel resolvability. We use a slightly different
technique than that in~\cite{CuffSoftCovering}, which we extend to the
multiple-access case to provide an explicit statement and a more
intuitive proof for a result only implicitly contained
in~\cite{WieseWiretap}. In an excursus, we show how a slight
refinement of the same technique similarly as
in~\cite{CuffSoftCovering} yields also a partial second-order result.

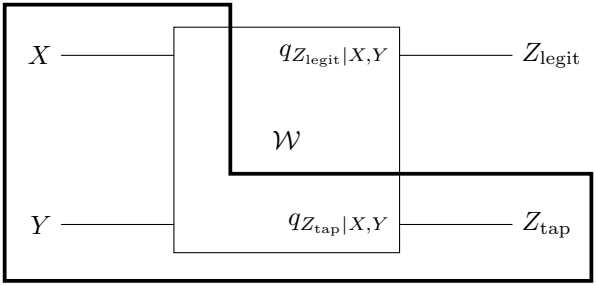
\begin{figure}
\begin{center}
\begin{tikzpicture}[scale=1.5]
\node[left]   (channelInOne)                at (   0,1.75) {$\channelInOne$};
\node[right]  (channelInOneConnect)         at (   1,1.75) {};
\node[left]   (channelInTwo)                at (   0,0.25) {$\channelInTwo$};
\node[right]  (channelInTwoConnect)         at (   1,0.25) {};
\node[right]  (channelOutLegit)             at (   4,1.75) {$\channelOutLegit$};
\node[left]   (channelOutLegitConnect)      at (   3,1.75) {$\channelpmf_{\channelOutLegit | \channelInOne, \channelInTwo}$};
\node[right]  (channelOutWiretapper)        at (   4,0.25) {$\channelOutWiretapper$};
\node[left]   (channelOutWiretapperConnect) at (   3,0.25) {$\channelpmf_{\channelOutWiretapper | \channelInOne, \channelInTwo}$};
\node         (channel)                     at (   2,   1) {$\channel$};

\draw (1,2) rectangle (3,0);
\draw (channelInOne) -- (channelInOneConnect);
\draw (channelInTwo) -- (channelInTwoConnect);
\draw (channelOutLegitConnect) -- (channelOutLegit);
\draw (channelOutWiretapperConnect) -- (channelOutWiretapper);


\draw[line width=.5mm] (-.5,-.25) -- (4.7,-.25) -- (4.7,.7) -- (1.5,.7) -- (1.5,2.2) -- (-.5,2.2) -- cycle;
\end{tikzpicture}
\end{center}
\caption{Channel resolvability problem as part of the wiretap coding problem.}
\label{figure:wiretap-decompose}
\end{figure}

In Section~\ref{sec:secrecy}, we discuss different notions of secrecy
and the interconnection between them. In particular, in addition to
the classical concept of strong secrecy, we review the notions of
distinguishing security and semantic security, both of which originate
from the field of cryptography. Distinguishing security allows us to
decompose the problem for wiretap coding into a problem of coding for
the MAC and a channel resolvability problem as shown in
Figure~\ref{figure:wiretap-decompose}. The concept of semantic
security, on the other hand, which is implied by distinguishing
security, has some very strong operational implications, some of which
we explore from the point of view of Ahlswede's General Theory of
Information Transfer~\cite{AhlswedeGeneral}.

\section{Notation, Definitions and Prerequisites}
\label{section:preliminaries}
The operations $\log$ and $\exp$ use Euler's number as a basis, and all information quantities are given in nats. $\positive{\cdot}$ denotes the maximum of its argument and $0$.

A \emph{channel}
$\channel = (\channelInOneAlph, \channelInTwoAlph, \channelOutAlph, \channelpmf_{\channelOut | \channelInOne, \channelInTwo})$
is given by finite input alphabets $\channelInOneAlph$ and $\channelInTwoAlph$, a finite output alphabet $\channelOutAlph$ and a collection of probability mass functions $\channelpmf_{\channelOut | \channelInOne, \channelInTwo}$ on $\channelOutAlph$ for each pair $(\channelInOneAlphElement,\channelInTwoAlphElement) \in \channelInOneAlph \times \channelInTwoAlph$. The random variables $\channelInOne$, $\channelInTwo$ and $\channelOut$ represent the two channel inputs and the channel output, respectively. \emph{Input distributions} for the channel are probability mass functions on $\channelInOneAlph$ and $\channelInTwoAlph$ denoted by $\channelpmf_{\channelInOne}$ and $\channelpmf_{\channelInTwo}$, respectively. We define an \emph{induced joint distribution} $\channelpmf_{\channelInOne, \channelInTwo, \channelOut}$ on $\channelInOneAlph \times \channelInTwoAlph \times \channelOutAlph$ by
$\channelpmf_{\channelInOne, \channelInTwo, \channelOut}(\channelInOneAlphElement,\channelInTwoAlphElement,\channelOutAlphElement) :=  \channelpmf_{\channelInOne}(\channelInOneAlphElement) \channelpmf_{\channelInTwo}(\channelInTwoAlphElement) \channelpmf_{\channelOut | \channelInOne, \channelInTwo}(\channelOutAlphElement | \channelInOneAlphElement,\channelInTwoAlphElement)$
and the \emph{output distribution}
$
\channelpmf_\channelOut(\channelOutAlphElement)
:=
\sum_{\channelInOneAlphElement \in \channelInOneAlph}
\sum_{\channelInTwoAlphElement \in \channelInTwoAlph}
  \channelpmf_{\channelInOne, \channelInTwo, \channelOut}(\channelInOneAlphElement, \channelInTwoAlphElement, \channelOutAlphElement)
$
is the marginal distribution of $\channelOut$.

By a pair of \emph{codebooks} of block length $\codebookBlocklength \geq 1$ and rates $\codebookRateOne \geq 0$ and $\codebookRateTwo \geq 0$, we mean finite sequences
$\codebookOne = (\codebookOneWord{\codewordIndex})_{\codewordIndex = 1}^{\exp(\codebookBlocklength\codebookRateOne)}$
and
$\codebookTwo = (\codebookTwoWord{\codewordIndex})_{\codewordIndex = 1}^{\exp(\codebookBlocklength\codebookRateTwo)}$,
where the \emph{codewords} $\codebookOneWord{\codewordIndex} \in \channelInOneAlph^\codebookBlocklength$ and $\codebookTwoWord{\codewordIndex} \in \channelInTwoAlph^\codebookBlocklength$ are finite sequences of elements of the input alphabets. We define a probability distribution $\Probability_{\codebookOne, \codebookTwo}$ on these codebooks as i.i.d. drawings in each component of each codeword according to $\channelpmf_\channelInOne$ and $\channelpmf_\channelInTwo$, respectively. Accordingly, we define the \emph{output distribution induced by $\codebookOne$ and $\codebookTwo$} on $\channelOutAlph^\codebookBlocklength$ by
\begin{multline*}
\codebookpmf_{\channelOut^\codebookBlocklength | \codebookOne, \codebookTwo}(\channelOutAlphElement^\codebookBlocklength) :=
  \exp(-\codebookBlocklength(\codebookRateOne+\codebookRateTwo))
  \\ \cdot
  \sum\limits_{\codewordIndex_1=1}^{\exp(\codebookBlocklength\codebookRateOne)}
  \sum\limits_{\codewordIndex_2=1}^{\exp(\codebookBlocklength\codebookRateTwo)}
      \channelpmf_{\channelOut^\codebookBlocklength | \channelInOne^\codebookBlocklength, \channelInTwo^\codebookBlocklength}(\channelOutAlphElement^\codebookBlocklength | \codebookOneWord{\codewordIndex_1}, \codebookTwoWord{\codewordIndex_2}).
\end{multline*}

Given probability distributions $\generalpdistOne$ and $\generalpdistTwo$ on a finite set $\generalrvOneAlph$ with mass functions $\generalpmfOne$ and $\generalpmfTwo$, respectively, and positive $\renyiParam \neq 1$, the \emph{Rényi divergence of order $\renyiParam$ of $\generalpdistOne$ from $\generalpdistTwo$} is defined as
\[
\renyidiv{\renyiParam}{\generalpdistOne}{\generalpdistTwo}
:=
\frac{1}{\renyiParam-1}
\log
\sum\limits_{\generalrvOneValue \in \generalrvOneAlph}
  \generalpmfOne(\generalrvOneValue)^\renyiParam
  \generalpmfTwo(\generalrvOneValue)^{1-\renyiParam}.
\]
Furthermore, we define the \emph{variational distance} between $\generalpdistOne$ and $\generalpdistTwo$ (or between their mass functions) as
\[
\totalvariation{\generalpmfOne - \generalpmfTwo}
:=
\frac{1}{2} \sum\limits_{\generalrvOneValue \in \generalrvOneAlph} \absolute{\generalpmfOne(\generalrvOneValue) - \generalpmfTwo(\generalrvOneValue)}
=
\sum\limits_{\generalrvOneValue \in \generalrvOneAlph} \positive{\generalpmfOne(\generalrvOneValue) - \generalpmfTwo(\generalrvOneValue)}.
\]

Given random variables $\generalrvOne$, $\generalrvTwo$ and $\generalrvThree$ distributed according to $\generalpmf_{\generalrvOne, \generalrvTwo, \generalrvThree}$, we define the \emph{(conditional) information density} as
\[
\informationDensity{\generalrvOneValue}{\generalrvTwoValue} := \log \frac{\generalpmf_{\generalrvTwo | \generalrvOne}(\generalrvTwoValue | \generalrvOneValue)}{\generalpmf_{\generalrvTwo}(\generalrvTwoValue)}
,~~
\informationDensityConditional{\generalrvOneValue}{\generalrvTwoValue}{\generalrvThreeValue} := \log \frac{\generalpmf_{\generalrvTwo | \generalrvOne, \generalrvThree}(\generalrvTwoValue | \generalrvOneValue, \generalrvThreeValue)}{\generalpmf_{\generalrvTwo | \generalrvThree}(\generalrvTwoValue | \generalrvThreeValue)}.
\]
The (conditional) mutual information is the expected value of the (conditional) information density. We use $\entropy{\generalrvOne}$ and $\entropyConditional{\generalrvOne}{\generalrvTwo}$ to denote Shannon's (conditional) entropy.

The following inequality was introduced in~\cite{Berry} and~\cite{Esseen}; we use a refinement here which follows e.g. from~\cite{BeekBerryEsseen}.
\begin{theorem}[Berry-Esseen Inequality]
\label{theorem:berry-esseen}
Given a sequence $(\generalrvOne_\generalSummationIndex)_{\generalSummationIndex=1}^{\generalSummationBound}$ of i.i.d. copies of a random variable $\generalrvOne$ on the reals with $\Expectation \generalrvOne = 0$, $\Expectation \generalrvOne^2 = \generalSecondMoment^2 < \infty$ and $\Expectation \absolute{\generalrvOne}^3 = \generalThirdMoment < \infty$, define $\bar{\generalrvOne} := (\generalrvOne_1 + \dots + \generalrvOne_\generalSummationBound)/\generalSummationBound$. Then the cumulative distribution functions $\generalcdf(\generalrvOneValue) := \Probability(\bar{\generalrvOne}\sqrt{\generalSummationBound}/\generalSecondMoment \leq \generalrvOneValue)$ of $\bar{\generalrvOne}\sqrt{\generalSummationBound}/\generalSecondMoment$ and $\normalcdf(\generalrvOneValue) := \int_{-\infty}^\generalrvOneValue 1/(2\pi) \exp(-\generalintegrand^2/2)  d \generalintegrand$ of the standard normal distribution satisfy
\[
\forall \generalrvOneValue \in \reals:~
\absolute{\generalcdf(\generalrvOneValue) - \normalcdf(\generalrvOneValue)}
\leq
\frac{\generalThirdMoment}
     {\generalSecondMoment^3 \sqrt{\generalSummationBound}}.
\]
\end{theorem}

We further use variations of the concentration bounds introduced in~\cite{HoeffdingInequalities}.
\begin{theorem}[Chernoff-Hoeffding Bound]
\label{theorem:hoeffding}
Suppose $\generalrvOne = \sum_{\generalSummationIndex=1}^{\generalSummationBound} \generalrvOne_\generalSummationIndex$, where the random variables in the sequence $(\generalrvOne_\generalSummationIndex)_{\generalSummationIndex=1}^\generalSummationBound$ are independently distributed with values in $[0,1]$ and $\Expectation \generalrvOne \leq \lemmaexpectation$. Then for $0 < \lemmaconst < 1$,
\[
\Probability(\generalrvOne > \lemmaexpectation(1+\lemmaconst)) \leq \exp\left(-\frac{\lemmaconst^2}{3} \lemmaexpectation \right).
\]
\end{theorem}
This version can e.g. be found in~\cite[Ex. 1.1]{ConcentrationTextbook}. We also use an extension of the Chernoff-Hoeffding bound for dependent variables due to Janson~\cite[Theorem 2.1]{JansonLargeDeviations}, of which we state only a specialized instance that is used in this paper.
\begin{theorem}[Janson~\cite{JansonLargeDeviations}]
\label{theorem:janson}
Suppose $\generalrvOne = \sum_{\generalSummationIndex=1}^{\generalSummationBound} \generalrvOne_\generalSummationIndex$, where the random variables in the sequence $(\generalrvOne_\generalSummationIndex)_{\generalSummationIndex=1}^\generalSummationBound$ take values in $[0,1]$ and can be partitioned into $\coverNumber \geq 1$ sets such that the random variables in each set are independently distributed. Then, for $\lemmaconst > 0$,
\[
\Probability(\generalrvOne \geq \Expectation \generalrvOne + \lemmaconst)
\leq
\exp\left(
  -2 \frac{\lemmaconst^2}
          {\coverNumber \cdot \generalSummationBound}
\right).
\]
\end{theorem}
We also use~\cite[Lemma 2.7]{CsiszarInformation} which provides bounds for the entropy difference of random variables that are close to each other in terms of their variational distance. Note that the definition of the variational distance in~\cite{CsiszarInformation} differs from the one used in this paper by a factor of $1/2$, and we state the lemma conforming with our definitions.
\begin{lemma}
\label{totvar-entropy-lemma}
Let $\generalrvOne$ and $\generalrvTwo$ be random variables on an alphabet $\generalrvOneAlph$, distributed according to probability mass functions $\generalpmf_\generalrvOne$ and $\generalpmf_\generalrvTwo$, respectively. If $\totalvariation{\generalpmf_\generalrvOne - \generalpmf_\generalrvTwo} = \lemmaconst \leq 1/4$, then
\[
\absolute{
  \entropy{\generalrvOne}
  -
  \entropy{\generalrvTwo}
}
\leq
-
\frac{1}{2}
\lemmaconst
\log \frac{\lemmaconst}{2\cardinality{\generalrvOneAlph}}.
\]
\end{lemma}
\begin{remark}
Since the function $-\lemmaconst/2 \cdot \log (\lemmaconst/2\cardinality{\generalrvOneAlph})$ is nondecreasing (in particular) for $0 \leq \lemmaconst \leq 1/4$, the assumption can be weakened to $\totalvariation{\generalpmf_\generalrvOne - \generalpmf_\generalrvTwo} \leq \lemmaconst \leq 1/4$.
\end{remark}
Furthermore, for the proof of the converse theorem, we need a lemma that is based on the Fenchel-Eggleston-Carathéodory theorem and appeared first in~\cite{AhlswedeSource} and of which we state the version from~\cite[Appendix C]{ElGamalNetworkIT}.
\begin{lemma}
\label{convex-cover-lemma}
Let $\generalrvOneAlph$ be a finite set and $\generalrvTwoAlph$ an arbitrary set. Suppose that $\generalpmfset$ is a connected compact subset of probability mass functions on $\generalrvOneAlph$ and consider a family $(\generalpmf_\generalrvTwoValue)_{\generalrvTwoValue \in \generalrvTwoAlph}$ of elements of $\generalpmfset$. Suppose further that $(\generalfunction_\generaldimensionIndex)_{\generaldimensionIndex=1}^\generaldimension$ are continuous functions mapping from $\generalpmfset$ to the reals and $\generalrvTwo$ is a random variable on $\generalrvTwoAlph$ distributed according to some probability measure $\generalpmeasure$. Then there exist a random variable $\generalrvTwo'$ on some alphabet $\generalrvTwoAlph'$ of cardinality at most $\generaldimension$ with probability mass function $\generalpmf$ and a family $(\generalpmf'_{\generalrvTwoValue'})_{\generalrvTwoValue' \in \generalrvTwoAlph'}$ of elements of $\generalpmfset$ such that for each $\generaldimensionIndex \in \{1, \dots, \generaldimension\}$,
\begin{align}
\label{convex-cover-lemma-function-property}
\int\limits_\generalrvTwoAlph
  \generalfunction_\generaldimensionIndex(\generalpmf_\generalrvTwoValue)
\generalpmeasure(d\generalrvTwoValue)
=
\sum\limits_{\generalrvTwoValue' \in \generalrvTwoAlph'}
  \generalfunction_\generaldimensionIndex(\generalpmf'_{\generalrvTwoValue'})
  \generalpmf(\generalrvTwoValue').
\end{align}
\end{lemma}

\section{Proof of the Direct Theorem}
\label{sec:main}
In this section, we prove Theorem~\ref{theorem:soft-covering-two-transmitters-convex}. Together with the observation that the resolvability region $\capacityRegion{\channel}{\channelpmf_\channelOut}$ is convex, by a standard time sharing argument, this yields the direct part of Theorem~\ref{theorem:resolvability-region} as an immediate consequence.

First, we show a slightly weaker version of Theorem~\ref{theorem:soft-covering-two-transmitters-convex}.
\begin{theorem}
\label{theorem:soft-covering-two-transmitters}
Suppose
$\channel = (\channelInOneAlph, \channelInTwoAlph, \channelOutAlph, \channelpmf_{\channelOut | \channelInOne, \channelInTwo})$
is a channel, $\channelpmf_\channelInOne$ and $\channelpmf_\channelInTwo$ are input distributions,
\begin{align*}
\codebookRateOne &> \mutualInformationConditional{\channelInOne}{\channelOut}{\channelInTwo}
\\
\codebookRateTwo &> \mutualInformation{\channelInTwo}{\channelOut}.
\end{align*}
Then there exist $\finalconstOne, \finalconstTwo > 0$ such that for large enough block length $\codebookBlocklength$, the codebook distributions of block length $\codebookBlocklength$ and rates $\codebookRateOne$ and $\codebookRateTwo$ satisfy
\begin{multline}
\label{theorem:soft-covering-two-transmitters-probability-statement}
\Probability_{\codebookOne, \codebookTwo} \left(
  \totalvariation{
    \codebookpmf_{\channelOut^\codebookBlocklength | \codebookOne, \codebookTwo} - \channelpmf_{\channelOut^\codebookBlocklength}
  }
  >
  \exp(-\finalconstOne\codebookBlocklength)
\right)
\\
\leq
\exp\left(-\exp\left(\finalconstTwo\codebookBlocklength\right)\right).
\end{multline}
\end{theorem}
The full Theorem~\ref{theorem:soft-covering-two-transmitters-convex} can then be proven with a time sharing argument which we detail at the end of this section. In the proof of Theorem~\ref{theorem:soft-covering-two-transmitters}, we consider two types of typical sets:
\begin{align*}
  \typicalSetIndex{\typicalityParam}{\codebookBlocklength}{1}
  &:=
  \{
    (\channelInOneAlphElement^\codebookBlocklength, \channelInTwoAlphElement^\codebookBlocklength,\channelOutAlphElement^\codebookBlocklength)
    :
    \informationDensityConditional{\channelInOneAlphElement^\codebookBlocklength}{\channelOutAlphElement^\codebookBlocklength}{\channelInTwoAlphElement^\codebookBlocklength}
    \leq
    \codebookBlocklength(\mutualInformationConditional{\channelInOne}{\channelOut}{\channelInTwo}+\typicalityParam)
  \}
\\
   \typicalSetIndex{\typicalityParam}{\codebookBlocklength}{2}
   &:=
   \{
     (\channelInTwoAlphElement^\codebookBlocklength,\channelOutAlphElement^\codebookBlocklength)
     :
     \informationDensity{\channelInTwoAlphElement^\codebookBlocklength}{\channelOutAlphElement^\codebookBlocklength}
    \leq
    \codebookBlocklength(\mutualInformation{\channelInTwo}{\channelOut}+\typicalityParam)
  \}.
\end{align*}
We split the variational distance in atypical and typical parts as follows, where $\totvarAtypicalOne$, $\totvarAtypicalTwo$ and $\totvarTypical{\channelOutAlphElement^\codebookBlocklength}$ are defined by~(\ref{def:soft-covering-atypical-term-one}), (\ref{def:soft-covering-atypical-term-two}) and (\ref{def:soft-covering-typical-term}) shown on the next page.

\begin{figure*}
\normalsize
\begin{align}
\label{def:soft-covering-atypical-term-one}
\totvarAtypicalOne
&:=
\sum\limits_{\channelOutAlphElement^\codebookBlocklength \in \channelOutAlph^\codebookBlocklength}
\exp(-\codebookBlocklength(\codebookRateOne+\codebookRateTwo))
\sum\limits_{\codewordIndex_1=1}^{\exp(\codebookBlocklength\codebookRateOne)}
\sum\limits_{\codewordIndex_2=1}^{\exp(\codebookBlocklength\codebookRateTwo)}
    \channelpmf_{\channelOut^\codebookBlocklength | \channelInOne^\codebookBlocklength, \channelInTwo^\codebookBlocklength}(\channelOutAlphElement^\codebookBlocklength | \codebookOneWord{\codewordIndex_1}, \codebookTwoWord{\codewordIndex_2})
    \indicator{(\codebookOneWord{\codewordIndex_1}, \codebookTwoWord{\codewordIndex_2}, \channelOutAlphElement^\codebookBlocklength) \notin \typicalSetIndex{\typicalityParam}{\codebookBlocklength}{1}}
\\
\label{def:soft-covering-atypical-term-two}
\totvarAtypicalTwo
&:=
\sum\limits_{\channelOutAlphElement^\codebookBlocklength \in \channelOutAlph^\codebookBlocklength}
\exp(-\codebookBlocklength(\codebookRateOne+\codebookRateTwo))
\sum\limits_{\codewordIndex_1=1}^{\exp(\codebookBlocklength\codebookRateOne)}
\sum\limits_{\codewordIndex_2=1}^{\exp(\codebookBlocklength\codebookRateTwo)}
    \channelpmf_{\channelOut^\codebookBlocklength | \channelInOne^\codebookBlocklength, \channelInTwo^\codebookBlocklength}(\channelOutAlphElement^\codebookBlocklength | \codebookOneWord{\codewordIndex_1}, \codebookTwoWord{\codewordIndex_2})
    \indicator{(\codebookTwoWord{\codewordIndex_2}, \channelOutAlphElement^\codebookBlocklength) \notin \typicalSetIndex{\typicalityParam}{\codebookBlocklength}{2}}
\\
\label{def:soft-covering-typical-term}
\totvarTypical{\channelOutAlphElement^\codebookBlocklength}
&:=
    \sum\limits_{\codewordIndex_1=1}^{\exp(\codebookBlocklength\codebookRateOne)}
    \sum\limits_{\codewordIndex_2=1}^{\exp(\codebookBlocklength\codebookRateTwo)}
        \exp(-\codebookBlocklength(\codebookRateOne+\codebookRateTwo))
        \frac{\channelpmf_{\channelOut^\codebookBlocklength | \channelInOne^\codebookBlocklength, \channelInTwo^\codebookBlocklength}(\channelOutAlphElement^\codebookBlocklength | \codebookOneWord{\codewordIndex_1}, \codebookTwoWord{\codewordIndex_2})}
             {\channelpmf_{\channelOut^\codebookBlocklength}(\channelOutAlphElement^\codebookBlocklength)}
        \indicator{(\codebookTwoWord{\codewordIndex_2}, \channelOutAlphElement^\codebookBlocklength) \in \typicalSetIndex{\typicalityParam}{\codebookBlocklength}{2}}
        \indicator{(\codebookOneWord{\codewordIndex_1}, \codebookTwoWord{\codewordIndex_2}, \channelOutAlphElement^\codebookBlocklength) \in \typicalSetIndex{\typicalityParam}{\codebookBlocklength}{1}}
\end{align}
\hrulefill
\end{figure*}

\begin{align}
\notag
&\totalvariation{ \codebookpmf_{\channelOut^\codebookBlocklength | \codebookOne, \codebookTwo} - \channelpmf_{\channelOut^\codebookBlocklength}}
\\
\notag
=
&\sum\limits_{\channelOutAlphElement^\codebookBlocklength \in \channelOutAlph^\codebookBlocklength}
  \channelpmf_{\channelOut^\codebookBlocklength}(\channelOutAlphElement^\codebookBlocklength)
  \positive{\frac{\codebookpmf_{\channelOut^\codebookBlocklength | \codebookOne, \codebookTwo}(\channelOutAlphElement^\codebookBlocklength)}
                 {\channelpmf_{\channelOut^\codebookBlocklength}(\channelOutAlphElement^\codebookBlocklength)}
  - 1
  }
\\
\label{proof:soft-covering-two-transmitters-typical-split}
\leq
&\totvarAtypicalOne + \totvarAtypicalTwo
+
\sum\limits_{\channelOutAlphElement^\codebookBlocklength \in \channelOutAlph^\codebookBlocklength}
\channelpmf_{\channelOut^\codebookBlocklength}(\channelOutAlphElement^\codebookBlocklength)
\positive{\totvarTypical{\channelOutAlphElement^\codebookBlocklength} - 1}.
\end{align}

\begin{remark}
The denominator of the fraction is almost surely not equal to $0$ as long as the numerator is not equal to $0$. We implicitly let the summation range only over the support of the denominator, as we do in all further summations.
\end{remark}

So the theorem can be proven by considering typical and atypical terms separately.
But first, we prove two lemmas to help us to bound the typical and the atypical terms.

\begin{lemma}[Bound for typical terms]
\label{lemma:soft-covering-two-transmitters-typical}
Given a block length $\codebookBlocklength$, $\typicalityParam > 0$, $0 < \lemmaconst < 1$, random variables $\generalrvOne$, $\generalrvTwo$ and $\generalrvThree$ on finite alphabets $\generalrvOneAlph$, $\generalrvTwoAlph$ and $\generalrvThreeAlph$ respectively with joint probability mass function $\generalpmf_{\generalrvOne, \generalrvTwo, \generalrvThree}$, a rate $\codebookRate$ and a codebook
$\codebook = (\codebookWord{\codewordIndex})_{\codewordIndex=1}^{\exp(\codebookBlocklength\codebookRate)}$ with each component of each codeword drawn i.i.d. according to $\generalpmf_\generalrvOne$, for any $\generalrvTwoValue^\codebookBlocklength \in \generalrvTwoAlph^\codebookBlocklength$ and $\generalrvThreeValue^\codebookBlocklength \in \generalrvThreeAlph^\codebookBlocklength$, we have
\begin{multline*}
\check{\Probability} :=
\Probability_{\codebook}\left(
  \sum\limits_{\codewordIndex=1}^{\exp(\codebookBlocklength\codebookRate)}
  \exp(-\codebookBlocklength\codebookRate)
  \frac{\generalpmf_{\generalrvThree^\codebookBlocklength | \generalrvOne^\codebookBlocklength, \generalrvTwo^\codebookBlocklength}(\generalrvThreeValue^\codebookBlocklength | \codebookWord{\codewordIndex}, \generalrvTwoValue^\codebookBlocklength)}
       {\generalpmf_{\generalrvThree^\codebookBlocklength | \generalrvTwo^\codebookBlocklength}(\generalrvThreeValue^\codebookBlocklength | \generalrvTwoValue^\codebookBlocklength)}
  \right.
  \\
  \left.
  \vphantom{\sum\limits_{\codewordIndex=1}^{\exp(\codebookBlocklength\codebookRate)}}
  \cdot
  \indicator{(\codebookWord{\codewordIndex}, \generalrvTwoValue^\codebookBlocklength, \generalrvThreeValue^\codebookBlocklength) \in \typicalSet{\typicalityParam}{\codebookBlocklength}}
  >
  1 + \lemmaconst
\right) \\
\leq
\exp\left(
  -\frac{\lemmaconst^2}{3} \exp(-\codebookBlocklength (\mutualInformationConditional{\generalrvOne}{\generalrvThree}{\generalrvTwo} + \typicalityParam - \codebookRate))
\right),
\end{multline*}
where the typical set is defined as
\begin{align}
\label{lemma:soft-covering-two-transmitters-typical-def}
\typicalSet{\typicalityParam}{\codebookBlocklength}
:=
\{
  (\generalrvOneValue^\codebookBlocklength, \generalrvTwoValue^\codebookBlocklength, \generalrvThreeValue^\codebookBlocklength)
  :
  \informationDensityConditional{\generalrvOneValue^\codebookBlocklength}{\generalrvThreeValue^\codebookBlocklength}{\generalrvTwoValue^\codebookBlocklength}
  \leq
  \codebookBlocklength(\mutualInformationConditional{\generalrvOne}{\generalrvThree}{\generalrvTwo}+\typicalityParam)
\}.
\end{align}

\end{lemma}

\begin{proof}
We have
\begin{multline*}
\check{\Probability} =
\Probability_{\codebook}\left(
  \sum\limits_{\codewordIndex=1}^{\exp(\codebookBlocklength\codebookRate)}
  \exp(-\codebookBlocklength (\mutualInformationConditional{\generalrvOne}{\generalrvThree}{\generalrvTwo} + \typicalityParam))
  \right.
  \\
  \cdot
  \frac{\generalpmf_{\generalrvThree^\codebookBlocklength | \generalrvOne^\codebookBlocklength, \generalrvTwo^\codebookBlocklength}(\generalrvThreeValue^\codebookBlocklength | \codebookWord{\codewordIndex}, \generalrvTwoValue^\codebookBlocklength)}
       {\generalpmf_{\generalrvThree^\codebookBlocklength | \generalrvTwo^\codebookBlocklength}(\generalrvThreeValue^\codebookBlocklength | \generalrvTwoValue^\codebookBlocklength)}
  \cdot
  \indicator{(\codebookWord{\codewordIndex}, \generalrvTwoValue^\codebookBlocklength, \generalrvThreeValue^\codebookBlocklength) \in \typicalSet{\typicalityParam}{\codebookBlocklength}}
  \\
  >
  \left. \vphantom{\sum\limits_{\codewordIndex=1}^{\exp(\codebookBlocklength\codebookRate)}}
  \exp(-\codebookBlocklength (\mutualInformationConditional{\generalrvOne}{\generalrvThree}{\generalrvTwo} + \typicalityParam - \codebookRate))
  (1 + \lemmaconst)
  \right).
\end{multline*}
By the definition of $\typicalSet{\typicalityParam}{\codebookBlocklength}$ in~(\ref{lemma:soft-covering-two-transmitters-typical-def}), the summands are at most $1$, and furthermore, the expectation of the sum can be bounded as
\begin{align*}
&
\begin{aligned}
  \Expectation_{\codebook}\left(
    \sum\limits_{\codewordIndex=1}^{\exp(\codebookBlocklength\codebookRate)}
    \right.
    &\exp(-\codebookBlocklength (\mutualInformationConditional{\generalrvOne}{\generalrvThree}{\generalrvTwo} + \typicalityParam))
    \\
    &\cdot
    \frac{\generalpmf_{\generalrvThree^\codebookBlocklength | \generalrvOne^\codebookBlocklength, \generalrvTwo^\codebookBlocklength}(\generalrvThreeValue^\codebookBlocklength | \codebookWord{\codewordIndex}, \generalrvTwoValue^\codebookBlocklength)}
        {\generalpmf_{\generalrvThree^\codebookBlocklength | \generalrvTwo^\codebookBlocklength}(\generalrvThreeValue^\codebookBlocklength | \generalrvTwoValue^\codebookBlocklength)}
    \indicator{(\codebookWord{\codewordIndex}, \generalrvTwoValue^\codebookBlocklength, \generalrvThreeValue^\codebookBlocklength) \in \typicalSet{\typicalityParam}{\codebookBlocklength}}
  \left.
  \vphantom{\sum\limits_{\codewordIndex=1}^{\exp(\codebookBlocklength\codebookRate)}}
  \right)
\end{aligned}
\\
&
\begin{aligned}
  \leq
  \sum\limits_{\codewordIndex=1}^{\exp(\codebookBlocklength\codebookRate)}
  &\exp(-\codebookBlocklength (\mutualInformationConditional{\generalrvOne}{\generalrvThree}{\generalrvTwo} + \typicalityParam))
  \\
  &\cdot
  \Expectation_{\codebook}\left(
    \frac{\generalpmf_{\generalrvThree^\codebookBlocklength | \generalrvOne^\codebookBlocklength, \generalrvTwo^\codebookBlocklength}(\generalrvThreeValue^\codebookBlocklength | \codebookWord{\codewordIndex}, \generalrvTwoValue^\codebookBlocklength)}
        {\generalpmf_{\generalrvThree^\codebookBlocklength | \generalrvTwo^\codebookBlocklength}(\generalrvThreeValue^\codebookBlocklength | \generalrvTwoValue^\codebookBlocklength)}
  \right)
\end{aligned}
\\
&=
\exp(-\codebookBlocklength (\mutualInformationConditional{\generalrvOne}{\generalrvThree}{\generalrvTwo} + \typicalityParam - \codebookRate)).
\end{align*}
Now applying Theorem~\ref{theorem:hoeffding} to the above shows the desired probability statement and completes the proof.
\end{proof}

\begin{lemma}[Bound for atypical terms]
\label{lemma:soft-covering-two-transmitters-atypical}
Given a channel
$\channel = (\channelInOneAlph, \channelInTwoAlph, \channelOutAlph, \channelpmf_{\channelOut | \channelInOne, \channelInTwo})$,
input distributions $\channelpmf_\channelInOne$ and $\channelpmf_\channelInTwo$, some set $\alphSubset \subseteq \channelInOneAlph^\codebookBlocklength \times \channelInTwoAlph^\codebookBlocklength \times \channelOutAlph^\codebookBlocklength$, $\lemmaconst > 0$, $\lemmaexpectation \geq \Probability((\channelInOne^\codebookBlocklength, \channelInTwo^\codebookBlocklength, \channelOut^\codebookBlocklength) \in \alphSubset)$ as well as rates $\codebookRateOne$ and $\codebookRateTwo$ and codebooks distributed according to $\Probability_{\codebookOne, \codebookTwo}$ defined in Section~\ref{section:preliminaries}, we have
\begin{multline*}
\hat{\Probability} :=
\Probability_{\codebookOne,\codebookTwo}\left(
  \sum\limits_{\channelOutAlphElement^\codebookBlocklength \in \channelOutAlph^\codebookBlocklength}
  \exp(-\codebookBlocklength(\codebookRateOne+\codebookRateTwo))
  \right.
  \\
  \sum\limits_{\codewordIndex_1=1}^{\exp(\codebookBlocklength\codebookRateOne)}
  \sum\limits_{\codewordIndex_2=1}^{\exp(\codebookBlocklength\codebookRateTwo)}
      \channelpmf_{\channelOut^\codebookBlocklength | \channelInOne^\codebookBlocklength, \channelInTwo^\codebookBlocklength}(\channelOutAlphElement^\codebookBlocklength | \codebookOneWord{\codewordIndex_1}, \codebookTwoWord{\codewordIndex_2})
      \\
      \indicator{(\codebookOneWord{\codewordIndex_1}, \codebookTwoWord{\codewordIndex_2}, \channelOutAlphElement^\codebookBlocklength) \in \alphSubset}
  \left. \vphantom{\sum\limits_{\channelOutAlphElement^\codebookBlocklength \in \channelOutAlph^\codebookBlocklength}} >
  \lemmaexpectation(1+\lemmaconst)
\right) \\
\leq
\exp(-2 \lemmaconst^2 \lemmaexpectation^2 \exp(\codebookBlocklength\min(\codebookRateOne,\codebookRateTwo))).
\end{multline*}
\end{lemma}

\begin{proof}
We have
\begin{align*}
&
\begin{aligned}
\hat{\Probability} =
&\Probability_{\codebookOne,\codebookTwo}\left(
  \sum\limits_{\codewordIndex_1=1}^{\exp(\codebookBlocklength\codebookRateOne)}
  \sum\limits_{\codewordIndex_2=1}^{\exp(\codebookBlocklength\codebookRateTwo)}
  \sum\limits_{\channelOutAlphElement^\codebookBlocklength \in \channelOutAlph^\codebookBlocklength}
  \right.
  \\
      &~
      \channelpmf_{\channelOut^\codebookBlocklength | \channelInOne^\codebookBlocklength, \channelInTwo^\codebookBlocklength}(\channelOutAlphElement^\codebookBlocklength | \codebookOneWord{\codewordIndex_1}, \codebookTwoWord{\codewordIndex_2})
      \indicator{(\codebookOneWord{\codewordIndex_1}, \codebookTwoWord{\codewordIndex_2}, \channelOutAlphElement^\codebookBlocklength) \in \alphSubset}
  \\
  &\left. \vphantom{\sum\limits_{\channelOutAlphElement^\codebookBlocklength \in \channelOutAlph^\codebookBlocklength}} >
  \exp(\codebookBlocklength(\codebookRateOne+\codebookRateTwo))
  (
    \lemmaexpectation
    +
    \lemmaexpectation
    \lemmaconst
  )
\right)
\end{aligned}
\\
&
\begin{aligned}
\leq
&\Probability_{\codebookOne,\codebookTwo}\left(
  \sum\limits_{\codewordIndex_1=1}^{\exp(\codebookBlocklength\codebookRateOne)}
  \sum\limits_{\codewordIndex_2=1}^{\exp(\codebookBlocklength\codebookRateTwo)}
  \sum\limits_{\channelOutAlphElement^\codebookBlocklength \in \channelOutAlph^\codebookBlocklength}
  \right.
  \\
      &~
      \channelpmf_{\channelOut^\codebookBlocklength | \channelInOne^\codebookBlocklength, \channelInTwo^\codebookBlocklength}(\channelOutAlphElement^\codebookBlocklength | \codebookOneWord{\codewordIndex_1}, \codebookTwoWord{\codewordIndex_2})
      \indicator{(\codebookOneWord{\codewordIndex_1}, \codebookTwoWord{\codewordIndex_2}, \channelOutAlphElement^\codebookBlocklength) \in \alphSubset}
  \\
  &\left. \vphantom{\sum\limits_{\channelOutAlphElement^\codebookBlocklength \in \channelOutAlph^\codebookBlocklength}} >
  \exp\Big(\codebookBlocklength(\codebookRateOne+\codebookRateTwo)\Big)
  \Big(
    \Probability((\channelInOne^\codebookBlocklength, \channelInTwo^\codebookBlocklength, \channelOut^\codebookBlocklength) \in \alphSubset)
    +
    \lemmaexpectation
    \lemmaconst
  \Big)
\right)
\end{aligned}
\\
&\leq
\exp\left(
  -2\frac{\exp(2\codebookBlocklength(\codebookRateOne+\codebookRateTwo))\lemmaexpectation^2\lemmaconst^2}
         {\exp(\codebookBlocklength\max(\codebookRateOne,\codebookRateTwo)) \exp(\codebookBlocklength(\codebookRateOne + \codebookRateTwo))}
\right)
\\
&=
\exp(-2 \lemmaconst^2 \lemmaexpectation^2 \exp(\codebookBlocklength\min(\codebookRateOne,\codebookRateTwo))),
\end{align*}
where the inequality follows from Theorem~\ref{theorem:janson} by observing that the innermost sum is confined to $[0,1]$, the two outer summations together have $\exp(\codebookBlocklength(\codebookRateOne+\codebookRateTwo)$ summands which can be partitioned into $\exp(\codebookBlocklength(\max(\codebookRateOne,\codebookRateTwo))$ sets with $\exp(\codebookBlocklength\min(\codebookRateOne,\codebookRateTwo))$ independently distributed elements each, and the overall expectation of the term is $\exp(\codebookBlocklength(\codebookRateOne+\codebookRateTwo)\Probability((\channelInOne^\codebookBlocklength, \channelInTwo^\codebookBlocklength, \channelOut^\codebookBlocklength) \in \alphSubset)$.
\end{proof}

We now have all the technical ingerdients needed to finish the proof of Theorem~\ref{theorem:soft-covering-two-transmitters}.

\begin{proof}[Proof of Theorem~\ref{theorem:soft-covering-two-transmitters}]
In order to bound $\totvarAtypicalOne$, we observe that for any $\renyiParam > 1$, we can bound
\begin{align}
\label{proof:soft-covering-two-transmitters-probability-bound-start}
&\phantom{{}={}}
\Probability_{\channelInOne^\codebookBlocklength, \channelInTwo^\codebookBlocklength, \channelOut^\codebookBlocklength}((\channelInOne^\codebookBlocklength, \channelInTwo^\codebookBlocklength, \channelOut^\codebookBlocklength) \notin \typicalSetIndex{\typicalityParam}{\codebookBlocklength}{1})
\\
&\begin{aligned}
=
\Probability_{\channelInOne^\codebookBlocklength, \channelInTwo^\codebookBlocklength, \channelOut^\codebookBlocklength}\left(
  \vphantom{
    \frac{\channelpmf_{\channelOut^\codebookBlocklength | \channelInOne^\codebookBlocklength, \channelInTwo^\codebookBlocklength}(\channelOut^\codebookBlocklength | \channelInOne^\codebookBlocklength, \channelInTwo^\codebookBlocklength)}
        {\channelpmf_{\channelOut^\codebookBlocklength | \channelInTwo^\codebookBlocklength}(\channelOut^\codebookBlocklength | \channelInTwo^\codebookBlocklength)}
  }
  \right.
  &\frac{\channelpmf_{\channelOut^\codebookBlocklength | \channelInOne^\codebookBlocklength, \channelInTwo^\codebookBlocklength}(\channelOut^\codebookBlocklength | \channelInOne^\codebookBlocklength, \channelInTwo^\codebookBlocklength)}
       {\channelpmf_{\channelOut^\codebookBlocklength | \channelInTwo^\codebookBlocklength}(\channelOut^\codebookBlocklength | \channelInTwo^\codebookBlocklength)}
\\
  &>
  \left.
  \vphantom{
    \frac{\channelpmf_{\channelOut^\codebookBlocklength | \channelInOne^\codebookBlocklength, \channelInTwo^\codebookBlocklength}(\channelOut^\codebookBlocklength | \channelInOne^\codebookBlocklength, \channelInTwo^\codebookBlocklength)}
        {\channelpmf_{\channelOut^\codebookBlocklength | \channelInTwo^\codebookBlocklength}(\channelOut^\codebookBlocklength | \channelInTwo^\codebookBlocklength)}
  }
  \exp(\codebookBlocklength(\mutualInformationConditional{\channelInOne}{\channelOut}{\channelInTwo}+\typicalityParam))
\right)
\end{aligned}
\\
&\begin{aligned}
=
\Probability_{\channelInOne^\codebookBlocklength, \channelInTwo^\codebookBlocklength, \channelOut^\codebookBlocklength}
\left(
\vphantom{
  \frac{\channelpmf_{\channelOut^\codebookBlocklength | \channelInOne^\codebookBlocklength, \channelInTwo^\codebookBlocklength}(\channelOut^\codebookBlocklength | \channelInOne^\codebookBlocklength, \channelInTwo^\codebookBlocklength)}
          {\channelpmf_{\channelOut^\codebookBlocklength | \channelInTwo^\codebookBlocklength}(\channelOut^\codebookBlocklength | \channelInTwo^\codebookBlocklength)}
}
\right.
  &\left(
    \frac{\channelpmf_{\channelOut^\codebookBlocklength | \channelInOne^\codebookBlocklength, \channelInTwo^\codebookBlocklength}(\channelOut^\codebookBlocklength | \channelInOne^\codebookBlocklength, \channelInTwo^\codebookBlocklength)}
        {\channelpmf_{\channelOut^\codebookBlocklength | \channelInTwo^\codebookBlocklength}(\channelOut^\codebookBlocklength | \channelInTwo^\codebookBlocklength)}
  \right)^{\renyiParam-1}
  \\
  &>
  \exp(\codebookBlocklength(\renyiParam-1)(\mutualInformationConditional{\channelInOne}{\channelOut}{\channelInTwo}+\typicalityParam))
\left.
\vphantom{
  \frac{\channelpmf_{\channelOut^\codebookBlocklength | \channelInOne^\codebookBlocklength, \channelInTwo^\codebookBlocklength}(\channelOut^\codebookBlocklength | \channelInOne^\codebookBlocklength, \channelInTwo^\codebookBlocklength)}
       {\channelpmf_{\channelOut^\codebookBlocklength | \channelInTwo^\codebookBlocklength}(\channelOut^\codebookBlocklength | \channelInTwo^\codebookBlocklength)}
}
\right)
\end{aligned}
\\
&\begin{aligned}
\leq
&\Expectation_{\channelInOne^\codebookBlocklength, \channelInTwo^\codebookBlocklength, \channelOut^\codebookBlocklength}\left(
  \left(
    \frac{\channelpmf_{\channelOut^\codebookBlocklength | \channelInOne^\codebookBlocklength, \channelInTwo^\codebookBlocklength}(\channelOut^\codebookBlocklength | \channelInOne^\codebookBlocklength, \channelInTwo^\codebookBlocklength)}
        {\channelpmf_{\channelOut^\codebookBlocklength | \channelInTwo^\codebookBlocklength}(\channelOut^\codebookBlocklength | \channelInTwo^\codebookBlocklength)}
  \right)^{\renyiParam-1}
\right)
\\
&\cdot \exp(-\codebookBlocklength(\renyiParam-1)(\mutualInformationConditional{\channelInOne}{\channelOut}{\channelInTwo}+\typicalityParam))
\end{aligned}
\\
&\begin{aligned}
=
\exp\big(
  &\codebookBlocklength(\renyiParam-1)
  \cdot (
    \renyidiv{\renyiParam}{\Probability_{\channelInOne, \channelInTwo, \channelOut}}{\Probability_{\channelInOne | \channelInTwo}\Probability_{\channelOut | \channelInTwo}\Probability_\channelInTwo}
\\
    &-
    \mutualInformationConditional{\channelInOne}{\channelOut}{\channelInTwo}-\typicalityParam 
  )
\big)
\end{aligned}
\\
&\leq
\exp(-\codebookBlocklength\proofconstantOne),
\label{proof:soft-covering-two-transmitters-probability-bound-end}
\end{align}
where (\ref{proof:soft-covering-two-transmitters-probability-bound-end}) holds as long as $\proofconstantOne < (\renyiParam-1)(\mutualInformationConditional{\channelInOne}{\channelOut}{\channelInTwo}+\typicalityParam-\renyidiv{\renyiParam}{\Probability_{\channelInOne, \channelInTwo, \channelOut}}{\Probability_{\channelInOne | \channelInTwo}\Probability_{\channelOut | \channelInTwo}\Probability_\channelInTwo})$. We can achieve this for sufficiently small $\proofconstantOne > 0$ as long as $\renyiParam>1$ and $\mutualInformationConditional{\channelInOne}{\channelOut}{\channelInTwo}+\typicalityParam-\renyidiv{\renyiParam}{\Probability_{\channelInOne, \channelInTwo, \channelOut}}{\Probability_{\channelInOne | \channelInTwo}\Probability_{\channelOut | \channelInTwo}\Probability_\channelInTwo} > 0$. In order to choose an $\renyiParam > 1$ such that the latter requirement holds, note that since our alphabets are finite, the Rényi divergence is also finite and thus it is continuous and approaches the Kullback-Leibler divergence for $\renyiParam$ tending to $1$~\cite{RenyiDiv}, which is in this case equal to the mutual information term.

We apply Lemma~\ref{lemma:soft-covering-two-transmitters-atypical} with $\alphSubset = (\channelInOneAlph^\codebookBlocklength \times \channelInTwoAlph^\codebookBlocklength \times \channelOutAlph^\codebookBlocklength) \setminus \typicalSetIndex{\typicalityParam}{\codebookBlocklength}{1}$ and $\lemmaconst = 1$ to obtain
\begin{multline}
\label{proof:soft-covering-two-transmitters-atypical-bound-1}
\Probability_{\codebookOne, \codebookTwo}\left(
  \totvarAtypicalOne
  >
  2\exp(-\codebookBlocklength\proofconstantOne)
\right)
\\
\leq
\exp(
  -2\exp(
    \codebookBlocklength(
      \min(\codebookRateOne,\codebookRateTwo) - 2\proofconstantOne
    )
  )
).
\end{multline}
Proceeding along similar lines of reasoning including another application of Lemma~\ref{lemma:soft-covering-two-transmitters-atypical} with $\alphSubset = \channelInOneAlph^\codebookBlocklength \times ((\channelInTwoAlph^\codebookBlocklength \times \channelOutAlph^\codebookBlocklength) \setminus \typicalSetIndex{\typicalityParam}{\codebookBlocklength}{2})$ and $\lemmaconst=1$, we show that if $\proofconstantOne>0$ is small enough,
\begin{multline}
\label{proof:soft-covering-two-transmitters-atypical-bound-2}
\Probability_{\codebookOne, \codebookTwo}\left(
  \totvarAtypicalTwo
  >
  2\exp(-\codebookBlocklength\proofconstantOne)
\right)
\\
\leq
\exp(
  -2\exp(
    \codebookBlocklength(
      \min(\codebookRateOne,\codebookRateTwo) - 2\proofconstantOne
    )
  )
).
\end{multline}
As for the typical term, we first observe that for any fixed $\channelInTwoAlphElement^\codebookBlocklength$ and $\channelOutAlphElement^\codebookBlocklength$, we can apply Lemma~\ref{lemma:soft-covering-two-transmitters-typical} with $\generalrvOne=\channelInOne$, $\generalrvTwo=\channelInTwo$, $\generalrvThree=\channelOut$ and $\lemmaconst=\exp(-\codebookBlocklength\proofconstantOne)$ to obtain
\begin{multline}
\label{proof:soft-covering-two-transmitters-typical-bound}
\Probability_{\codebookOne}\left(
  \totvarTypicalOne{\channelInTwoAlphElement^\codebookBlocklength}{\channelOutAlphElement^\codebookBlocklength}
  >
  1 + \exp(-\codebookBlocklength\proofconstantOne)
\right)
\\
\leq
\exp\left(
  -\frac{1}{3} \exp(-\codebookBlocklength (\mutualInformationConditional{\channelInOne}{\channelOut}{\channelInTwo} + \typicalityParam + 2\proofconstantOne - \codebookRateOne))
\right),
\end{multline}
where we used
\begin{multline}
\label{def:soft-covering-typical-term-one}
\totvarTypicalOne{\channelInTwoAlphElement^\codebookBlocklength}{\channelOutAlphElement^\codebookBlocklength}
:=
\sum\limits_{\codewordIndex_1=1}^{\exp(\codebookBlocklength\codebookRateOne)}
    \exp(-\codebookBlocklength(\codebookRateOne))
    \\
    \cdot \frac{\channelpmf_{\channelOut^\codebookBlocklength | \channelInOne^\codebookBlocklength, \channelInTwo^\codebookBlocklength}(\channelOutAlphElement^\codebookBlocklength | \codebookOneWord{\codewordIndex_1}, \channelInTwoAlphElement^\codebookBlocklength)}
          {\channelpmf_{\channelOut^\codebookBlocklength | \channelInTwo^\codebookBlocklength}(\channelOutAlphElement^\codebookBlocklength | \channelInTwoAlphElement^\codebookBlocklength)}
    \indicator{(\codebookOneWord{\codewordIndex_1}, \channelInTwoAlphElement^\codebookBlocklength, \channelOutAlphElement^\codebookBlocklength) \in \typicalSetIndex{\typicalityParam}{\codebookBlocklength}{1}}.
\end{multline}
We define a set of codebooks
\begin{align}
\label{def:soft-covering-good-codebooks}
\codebookSet_{\channelOutAlphElement^\codebookBlocklength}
:=
\bigcap\limits_{\channelInTwoAlphElement^\codebookBlocklength \in \channelInTwoAlph^\codebookBlocklength}
  \left\{
    \codebookOne:
    \totvarTypicalOne{\channelInTwoAlphElement^\codebookBlocklength}{\channelOutAlphElement^\codebookBlocklength}
    \leq
    1 + \exp(-\codebookBlocklength\proofconstantOne)
  \right\}
\end{align}
and bound for arbitrary but fixed $\channelOutAlphElement^\codebookBlocklength$
\begin{align*}
\tilde{\Probability}
:=
\Probability_{\codebookOne, \codebookTwo}\left(
  \totvarTypical{\channelOutAlphElement^\codebookBlocklength}
  >
  1 + 3\exp(-\codebookBlocklength\proofconstantOne)
  ~|~
  \codebookOne \in \codebookSet_{\channelOutAlphElement^\codebookBlocklength}
\right)
\end{align*}
in~(\ref{proof:soft-covering-two-transmitters-totalprob1}) to~(\ref{proof:soft-covering-two-transmitters-lemmaapplication2}) in the appendix.

We can now put everything together as shown in (\ref{proof:soft-covering-two-transmitters-union-bound-start}) to (\ref{proof:soft-covering-two-transmitters-union-bound-substitutions}) in the appendix.

What remains is to choose $\finalconstOne$ and $\finalconstTwo$ such that (\ref{theorem:soft-covering-two-transmitters-probability-statement}) holds. First, we have to choose $\typicalityParam$ and $\proofconstantOne$ small enough such that
\begin{align*}
\hat{\lemmaconst} := \min\big(
  &\min(\codebookRateOne,\codebookRateTwo)-2\proofconstantOne,
  \\
  &\codebookRateOne - 2\proofconstantOne - \typicalityParam - \mutualInformationConditional{\channelInOne}{\channelOut}{\channelInTwo},
  \\
  &\codebookRateTwo - 2\proofconstantOne - \typicalityParam - \mutualInformation{\channelInTwo}{\channelOut}
\big)
\end{align*}
is positive. Since there have so far been no constraints on $\proofconstantOne$ and $\typicalityParam$ except that they are positive and sufficiently small, such a choice is possible provided $\codebookRateOne > \mutualInformationConditional{\channelInOne}{\channelOut}{\channelInTwo}$ and $\codebookRateTwo > \mutualInformation{\channelInTwo}{\channelOut}$. The theorem then follows for large enough $\codebookBlocklength$ by choosing positive $\finalconstTwo < \hat{\lemmaconst}$ and $\finalconstOne < \proofconstantOne$.
\end{proof}

Obviously, we can reverse the roles of $\channelInOne$ and $\channelInTwo$ in Theorem~\ref{theorem:soft-covering-two-transmitters}. In order to prove the full Theorem~\ref{theorem:soft-covering-two-transmitters-convex}, however, we need time sharing between the two corner points, as we detail in the following proof.

\begin{proof}[Proof of Theorem~\ref{theorem:soft-covering-two-transmitters-convex}]
We are given $\codebookRateOne$ and $\codebookRateTwo$ satisfying (\ref{theorem:soft-covering-two-transmitters-convex-rateone}), (\ref{theorem:soft-covering-two-transmitters-convex-ratetwo}) and (\ref{theorem:soft-covering-two-transmitters-convex-sumrates}). In this proof, we show how to find a time sharing parameter $\hat{\convexityParam}$ with the following properties:
\begin{align*}
\codebookRateOne
&>
\hat{\convexityParam} \mutualInformation{\channelInOne}{\channelOut}
+
(1-\hat{\convexityParam}) \mutualInformationConditional{\channelInOne}{\channelOut}{\channelInTwo}
\\
\codebookRateTwo
&>
\hat{\convexityParam} \mutualInformationConditional{\channelInTwo}{\channelOut}{\channelInOne}
+
(1-\hat{\convexityParam}) \mutualInformation{\channelInTwo}{\channelOut}
\end{align*}
We can then conclude the proof by applying Theorem~\ref{theorem:soft-covering-two-transmitters} twice: First we apply it with a block length of $(1-\hat{\convexityParam})\codebookBlocklength$ and then we apply it with reversed roles of $\channelInOne$ and $\channelInTwo$ and a block length of $\hat{\convexityParam}\codebookBlocklength$.

To this end, we define functions mapping from $[0,1]$ to the reals
\begin{align*}
\codebookRateOne(\convexityParam)
&:=
\convexityParam \mutualInformation{\channelInOne}{\channelOut}
+
(1-\convexityParam) \mutualInformationConditional{\channelInOne}{\channelOut}{\channelInTwo}
\\
\codebookRateTwo(\convexityParam)
&:=
\convexityParam \mutualInformationConditional{\channelInTwo}{\channelOut}{\channelInOne}
+
(1-\convexityParam) \mutualInformation{\channelInTwo}{\channelOut}.
\end{align*}
Note that $\codebookRateOne(\convexityParam)$ and $\codebookRateTwo(\convexityParam)$ are continuous in $\convexityParam$, $\codebookRateOne(0) = \mutualInformationConditional{\channelInOne}{\channelOut}{\channelInTwo}$, $\codebookRateOne(1) = \mutualInformation{\channelInOne}{\channelOut}$, $\codebookRateTwo(0) = \mutualInformation{\channelInTwo}{\channelOut}$ and $\codebookRateTwo(1) = \mutualInformationConditional{\channelInTwo}{\channelOut}{\channelInOne}$. If $\codebookRateOne > \mutualInformationConditional{\channelInOne}{\channelOut}{\channelInTwo}$ or $\codebookRateTwo > \mutualInformationConditional{\channelInTwo}{\channelOut}{\channelInOne}$, there is nothing to prove since Theorem~\ref{theorem:soft-covering-two-transmitters} can be applied directly. So we assume $\codebookRateOne \leq \mutualInformationConditional{\channelInOne}{\channelOut}{\channelInTwo}$ and $\codebookRateTwo \leq \mutualInformationConditional{\channelInTwo}{\channelOut}{\channelInOne}$, which means that $\codebookRateOne(\convexityParam)$ is strictly decreasing and $\codebookRateTwo(\convexityParam)$ is strictly increasing in $\convexityParam$ and we can find $\convexityParam_1, \convexityParam_2 \in [0,1]$ such that $\codebookRateOne(\convexityParam_1) = \codebookRateOne$ and $\codebookRateTwo(\convexityParam_2) = \codebookRateTwo$. Next, we note that $\convexityParam_1 < \convexityParam_2$ because under the assumption $\convexityParam_1 \geq \convexityParam_2$, we can obtain the contradiction
\begin{align*}
\codebookRateOne + \codebookRateTwo \leq \codebookRateOne(\convexityParam_2) + \codebookRateTwo(\convexityParam_2) = \mutualInformation{\channelInOne, \channelInTwo}{\channelOut}.
\end{align*}
So we can find $\hat{\convexityParam}$ with $\convexityParam_1 < \hat{\convexityParam} < \convexityParam_2$. We then have $\codebookRateOne > \codebookRateOne(\hat{\convexityParam})$ and $\codebookRateTwo > \codebookRateTwo(\hat{\convexityParam})$, as required.
\end{proof}

\subsection{Second Order Rates}
Although the probability statement in Theorem~\ref{theorem:soft-covering-two-transmitters-convex} is quite strong in the sense that it guarantees a doubly exponentially low probability of drawing a bad pair of codebooks, a weak spot of the theorem as stated is that the constants $\finalconstOne$ and $\finalconstTwo$ are guaranteed to be positive, but can be arbitrarily small depending on the channel, the input distribution and the codebook rates and that the theorem is only valid for sufficiently large $\codebookBlocklength$. So in particular, if we are given a fixed $\codebookBlocklength$ (even a very large one), we cannot draw any conclusions for the error probabilities of codebooks of block length $\codebookBlocklength$. However, a closer examination of the proof shows that stronger statements are possible in this case. In this subsection, we give details how the proof of Theorem~\ref{theorem:soft-covering-two-transmitters-convex} can be modified so that it yields the following second-order result.
\begin{theorem}
\label{theorem:soft-covering-two-transmitters-second-order}
Given a channel
$\channel = (\channelInOneAlph, \channelInTwoAlph, \channelOutAlph, \channelpmf_{\channelOut | \channelInOne, \channelInTwo})$,
input distributions $\channelpmf_\channelInOne$ and $\channelpmf_\channelInTwo$, $\typicalityParam \in (0,1)$, let the central second and absolute third moment of $\informationDensityConditional{\channelInOne}{\channelOut}{\channelInTwo}$ be $\channelDispersion{1}$ and $\channelThirdMoment{1}$, respectively; analogously, we use $\channelDispersion{2}$ and $\channelThirdMoment{2}$ to denote the central second and absolute third moment of $\informationDensity{\channelInTwo}{\channelOut}$. Suppose the rates $\codebookRateOne, \codebookRateTwo$ depend on $\codebookBlocklength$ in the following way:
\begin{alignat}{3}
\label{theorem:soft-covering-two-transmitters-second-order-rate-one}
\codebookRateOne
&=
\mutualInformationConditional{\channelInOne}{\channelOut}{\channelInTwo}&
&+
\sqrt{\frac{\channelDispersion{1}}{\codebookBlocklength}} \normalcdfComplementInverse(\typicalityParam)&
&+
\secondOrderParamC\frac{\log \codebookBlocklength}
                       {\codebookBlocklength} \\
\label{theorem:soft-covering-two-transmitters-second-order-rate-two}
\codebookRateTwo
&=
\mutualInformation{\channelInTwo}{\channelOut}&
&+
\sqrt{\frac{\channelDispersion{2}}{\codebookBlocklength}} \normalcdfComplementInverse(\typicalityParam)&
&+
\secondOrderParamC\frac{\log \codebookBlocklength}
                       {\codebookBlocklength},
\end{alignat}
where $\normalcdfComplement := 1 - \normalcdf$ with $\normalcdf$ as defined in the statement of Theorem~\ref{theorem:berry-esseen}, and $\secondOrderParamC>1$. Then, for any $\secondOrderParamD \in (0, \secondOrderParamC-1)$, we have
\begin{multline*}
\begin{aligned}
  \Probability_{\codebookOne, \codebookTwo}
  \left( \vphantom{\frac{1}{\sqrt{\codebookBlocklength}}}
  \right.
    &\totalvariation{\codebookpmf_{\channelOut^\codebookBlocklength | \codebookOne, \codebookTwo} - \channelpmf_{\channelOut^\codebookBlocklength}}
    >
    \\ &\left.
    (\secondOrderAtypicalProbability{1} + \secondOrderAtypicalProbability{2})
    \left(1+\frac{1}{\sqrt{\codebookBlocklength}}\right)
    +
    \frac{3}{\sqrt{\codebookBlocklength}}
  \right)
\end{aligned}
\\
\leq
2\exp\left(
  -\frac{2\min(\secondOrderAtypicalProbability{1}^2,\secondOrderAtypicalProbability{2}^2)}
        {\codebookBlocklength}
  \exp(\codebookBlocklength \min(\codebookRateOne,\codebookRateTwo))
\right) \\
+
2\exp\left(
  \codebookBlocklength(\log \cardinality{\channelOutAlph} + \log \cardinality{\channelInTwoAlph})
  -\frac{1}{3}
  \codebookBlocklength^{\secondOrderParamC - \secondOrderParamD - 1}
\right),
\end{multline*}
where for both $\txIndex=1$ and $\txIndex=2$,
\begin{align*}
\secondOrderAtypicalProbability{\txIndex}
:=
\normalcdfComplement\left(
  \normalcdfComplementInverse(\typicalityParam)
  +
  \frac{\secondOrderParamD \log \codebookBlocklength}
       {\sqrt{\codebookBlocklength\channelDispersion{\txIndex}}}
\right)
+
\frac{\channelThirdMoment{\txIndex}}
     {\channelDispersion{\txIndex}^{\frac{3}{2}} \sqrt{\codebookBlocklength}}
\end{align*}
tends to $\typicalityParam$ for $\codebookBlocklength \rightarrow \infty$.
\end{theorem}
An immediate consequence of this theorem is the following corollary that follows by applying the theorem with the roles of $\channelInOne$ and $\channelInTwo$ reversed.
\begin{cor}
\label{cor:soft-covering-two-transmitters-second-order}
Theorem~\ref{theorem:soft-covering-two-transmitters-second-order} holds with (\ref{theorem:soft-covering-two-transmitters-second-order-rate-one}) and (\ref{theorem:soft-covering-two-transmitters-second-order-rate-two}) replaced by
\begin{align*}
\codebookRateOne
&=
\mutualInformation{\channelInOne}{\channelOut}
+
\sqrt{\frac{\channelDispersion{1}}{\codebookBlocklength}} \normalcdfComplementInverse(\typicalityParam)
+
\secondOrderParamC\frac{\log \codebookBlocklength}
                       {\codebookBlocklength} \\
\codebookRateTwo
&=
\mutualInformationConditional{\channelInTwo}{\channelOut}{\channelInOne}
+
\sqrt{\frac{\channelDispersion{2}}{\codebookBlocklength}} \normalcdfComplementInverse(\typicalityParam)
+
\secondOrderParamC\frac{\log \codebookBlocklength}
                       {\codebookBlocklength}
\end{align*}
and $\channelDispersion{1}$, $\channelThirdMoment{1}$, $\channelDispersion{2}$ and $\channelThirdMoment{2}$ redefined to be the second and third moments of $\informationDensity{\channelInOne}{\channelOut}$ and $\informationDensityConditional{\channelInTwo}{\channelOut}{\channelInOne}$, respectively.
\end{cor}
\begin{remark}
The question of how the achievable second-order rates behave near the line connecting these two corner points should be a subject of further research.
\end{remark}

\begin{proof}[Proof of Theorem~\ref{theorem:soft-covering-two-transmitters-second-order}]
We consider the typical sets $\typicalSetIndex{\typicalityParam_1}{\codebookBlocklength}{1}$ and $\typicalSetIndex{\typicalityParam_2}{\codebookBlocklength}{2}$, where for $\txIndex=1,2$, we choose $\typicalityParam_\txIndex >0$ to be
\begin{align}
\label{proof:soft-covering-two-transmitters-second-order-typicalityparam}
\typicalityParam_\txIndex
:=
\sqrt{\frac{\channelDispersion{\txIndex}}
           {\codebookBlocklength}
}
\normalcdfComplementInverse(\typicalityParam)
+
\secondOrderParamD
\frac{\log \codebookBlocklength}{\codebookBlocklength}.
\end{align}
The definitions~(\ref{def:soft-covering-atypical-term-one}), (\ref{def:soft-covering-atypical-term-two}) and (\ref{def:soft-covering-typical-term}) change accordingly.

In order to bound $\totvarAtypicalOne$, we use Theorem~\ref{theorem:berry-esseen} to obtain
\begin{align*}
&\phantom{{}={}}
\Probability_{\channelInOne^\codebookBlocklength, \channelInTwo^\codebookBlocklength, \channelOut^\codebookBlocklength}((\channelInOne^\codebookBlocklength, \channelInTwo^\codebookBlocklength, \channelOut^\codebookBlocklength) \notin \typicalSetIndex{\typicalityParam_1}{\codebookBlocklength}{1})
\\
&
=
\Probability_{\channelInOne^\codebookBlocklength, \channelInTwo^\codebookBlocklength, \channelOut^\codebookBlocklength}\left(
  \frac{1}{\codebookBlocklength}
  \sum\limits_{\blockIndex=1}^\codebookBlocklength
  \left(
    \informationDensityConditional{\channelInOne_\blockIndex}{\channelOut_\blockIndex}{\channelInTwo_\blockIndex}
    -
    \mutualInformationConditional{\channelInOne}{\channelOut}{\channelInTwo}
  \right)
  >
  \typicalityParam_1
\right)
\\
&\leq
\normalcdfComplement\left(
  \typicalityParam_1
  \sqrt{\frac{\codebookBlocklength}
             {\channelDispersion{1}}
  }
\right)
+
\frac{\channelThirdMoment{1}}
     {\channelDispersion{1}^{\frac{3}{2}} \sqrt{\codebookBlocklength}}
=
\secondOrderAtypicalProbability{1}.
\end{align*}
An application of Lemma~\ref{lemma:soft-covering-two-transmitters-atypical} with $\lemmaconst = 1/\sqrt{\codebookBlocklength}$ yields
\begin{multline}
\label{proof:soft-covering-two-transmitters-second-order-atypical-bound-1}
\Probability_{\codebookOne, \codebookTwo} \left(
  \totvarAtypicalOne
  >
  \secondOrderAtypicalProbability{1}\left(
    1+\frac{1}{\sqrt{\codebookBlocklength}}
  \right)
\right)
\\
\leq
\exp\left(
  -\frac{2\secondOrderAtypicalProbability{1}^2}
        {\codebookBlocklength}
  \exp(\codebookBlocklength \min(\codebookRateOne,\codebookRateTwo))
\right).
\end{multline}
Reasoning along similar lines shows
\begin{align*}
\Probability_{\channelInTwo^\codebookBlocklength, \channelOut^\codebookBlocklength}((\channelInTwo^\codebookBlocklength, \channelOut^\codebookBlocklength) \notin \typicalSetIndex{\typicalityParam_2}{\codebookBlocklength}{2})
\leq
\secondOrderAtypicalProbability{2}
\end{align*}
so that a further application of Lemma~\ref{lemma:soft-covering-two-transmitters-atypical} yields
\begin{multline}
\label{proof:soft-covering-two-transmitters-second-order-atypical-bound-2}
\Probability_{\codebookOne, \codebookTwo}\left(
  \totvarAtypicalTwo
  >
  \secondOrderAtypicalProbability{2}\left(
    1+\frac{1}{\sqrt{\codebookBlocklength}}
  \right)
\right)
\\
\leq
\exp\left(
  -\frac{2\secondOrderAtypicalProbability{2}^2}
        {\codebookBlocklength}
  \exp(\codebookBlocklength \min(\codebookRateOne,\codebookRateTwo))
\right).
\end{multline}
For the typical term, we use the definitions~(\ref{def:soft-covering-typical-term-one}) and~(\ref{def:soft-covering-good-codebooks}) with the typical set $\typicalSetIndex{\typicalityParam_1}{\codebookBlocklength}{1}$, and observe that for any fixed $\channelInTwoAlphElement^\codebookBlocklength$ and $\channelOutAlphElement^\codebookBlocklength$, we can apply Lemma~\ref{lemma:soft-covering-two-transmitters-typical} with $\generalrvOne=\channelInOne$, $\generalrvTwo=\channelInTwo$, $\generalrvThree=\channelOut$ and $\lemmaconst=1/\sqrt{\codebookBlocklength}$ to obtain
\begin{multline}
\label{proof:soft-covering-two-transmitters-second-order-typical-bound}
\Probability_{\codebookOne}\left(
  \totvarTypicalOne{\channelInTwoAlphElement^\codebookBlocklength}{\channelOutAlphElement^\codebookBlocklength}
  >
  1 + \frac{1}{\sqrt{\codebookBlocklength}})
\right) \\
\leq
\exp\left(
  -\frac{1}{3\codebookBlocklength} \exp(-\codebookBlocklength (\mutualInformationConditional{\channelInOne}{\channelOut}{\channelInTwo} + \typicalityParam_1 - \codebookRateOne))
\right).
\end{multline}

Now proceeding in a similar manner as in~(\ref{proof:soft-covering-two-transmitters-totalprob1}) to~(\ref{proof:soft-covering-two-transmitters-lemmaapplication2}) shows
\begin{multline*}
\Probability_{\codebookOne, \codebookTwo}\left(
  \totvarTypical{\channelOutAlphElement^\codebookBlocklength}
  >
  1 + \frac{3}{\sqrt{\codebookBlocklength}}
  ~|~
  \codebookOne \in \codebookSet_{\channelOutAlphElement^\codebookBlocklength}
\right)
\\
\leq
\exp\left(
  -\frac{1}{3\codebookBlocklength} \exp(-\codebookBlocklength (\mutualInformation{\channelInTwo}{\channelOut} + \typicalityParam_2 - \codebookRateTwo))
\right),
\end{multline*}
where there is no assumption on $\codebookBlocklength$ because $1/\sqrt{\codebookBlocklength} \leq 1$ for all $\codebookBlocklength \geq 1$.

The theorem then follows from (\ref{proof:soft-covering-two-transmitters-second-order-union-bound-start}) to (\ref{proof:soft-covering-two-transmitters-second-order-union-bound-end}) in the appendix.
\end{proof}

\section{Converse Theorem}
\label{sec:converse}

In this section, we establish the following theorem which implies the converse part of Theorem~\ref{theorem:resolvability-region}.
\begin{theorem}
\label{converse-theorem}
Let $\channel = (\channelInOneAlph, \channelInTwoAlph, \channelOutAlph, \channelpmf_{\channelOut | \channelInOne, \channelInTwo})$ be a channel, $\channelpmf_\channelInOne, \channelpmf_\channelInTwo$ input distributions (with induced output distribution $\channelpmf_\channelOut$) and $(\codebook_{1,\codebookBlocklength})_{\codebookBlocklength \geq 1}, (\codebook_{2,\codebookBlocklength})_{\codebookBlocklength \geq 1}$ sequences of codebooks of block lengths $\codebookBlocklength$ and rates $\codebookRateOne$ and $\codebookRateTwo$, respectively, such that
\[
\lim\limits_{\codebookBlocklength \rightarrow \infty}
\totalvariation{\codebookpmf_{\channelOut^\codebookBlocklength | \codebook_{1,\codebookBlocklength}, \codebook_{2,\codebookBlocklength}} - \channelpmf_{\channelOut^\codebookBlocklength}}
=
0.
\]
Then there are random variables $\channelInOne$ on $\channelInOneAlph$, $\channelInTwo$ on $\channelInTwoAlph$, $\channelOut$ on $\channelOutAlph$ and $\timeSharingRV$ on alphabet $\timeSharingAlph=\{1, \dots, \cardinality{\channelOutAlph} + 3\}$ with a joint probability mass function $\codebookpmf_\timeSharingRV \codebookpmf_{\channelInOne | \timeSharingRV} \codebookpmf_{\channelInTwo | \timeSharingRV} \channelpmf_{\channelOut | \channelInOne, \channelInTwo}$ such that the marginal distribution of $\channelOut$ is $\channelpmf_\channelOut$ and
\begin{align}
\label{converse-theorem-sumrate}
\mutualInformationConditional{\channelInOne, \channelInTwo}{\channelOut}{\timeSharingRV}
&\leq
\codebookRateOne + \codebookRateTwo
\\
\label{converse-theorem-rateone}
\mutualInformationConditional{\channelInOne}{\channelOut}{\timeSharingRV}
&\leq
\codebookRateOne
\\
\label{converse-theorem-ratetwo}
\mutualInformationConditional{\channelInTwo}{\channelOut}{\timeSharingRV}
&\leq
\codebookRateTwo.
\end{align}
\end{theorem}

We first prove the following lemma, which states conclusions we can draw from the existence of a codebook pair for a fixed block length $\codebookBlocklength$. Theorem~\ref{converse-theorem} then follows by considering the limit $\codebookBlocklength \rightarrow \infty$, as we detail at the end of this section.

\begin{lemma}
\label{converse-lemma}
Let $\channel = (\channelInOneAlph, \channelInTwoAlph, \channelOutAlph, \channelpmf_{\channelOut | \channelInOne, \channelInTwo})$ be a channel, $\channelpmf_\channelInOne, \channelpmf_\channelInTwo$ input distributions (with induced output distribution $\channelpmf_\channelOut$) and $\codebookOne, \codebookTwo$ a pair of codebooks of rates $\codebookRateOne, \codebookRateTwo$ and block length $\codebookBlocklength$ such that
\begin{align} 
\label{converse-lemma-tvassumption}
\totalvariation{\codebookpmf_{\channelOut^\codebookBlocklength | \codebookOne, \codebookTwo} - \channelpmf_{\channelOut^\codebookBlocklength}} \leq \lemmaconst \leq \frac{1}{4}.
\end{align}
Then there are random variables $\hat{\channelInOne}$ on $\channelInOneAlph$, $\hat{\channelInTwo}$ on $\channelInTwoAlph$, $\hat{\channelOut}$ on $\channelOutAlph$ and $\hat{\timeSharingRV}$ on $\timeSharingAlph :=  \{1,\dots, \codebookBlocklength\}$ with a joint probability mass function $\codebookpmf_{\hat{\timeSharingRV}} \codebookpmf_{\hat{\channelInOne} | \hat{\timeSharingRV}} \codebookpmf_{\hat{\channelInTwo} | \hat{\timeSharingRV}} \channelpmf_{\hat{\channelOut} | \hat{\channelInOne}, \hat{\channelInTwo}}$ such that $\channelpmf_{\hat{\channelOut} | \hat{\channelInOne}, \hat{\channelInTwo}} = \channelpmf_{\channelOut | \channelInOne, \channelInTwo}$ and
\begin{align}
\label{converse-lemma-convergence}
\totalvariationlr{\codebookpmf_{\hat{\channelOut}} - \channelpmf_{\channelOut}}
&\leq
\lemmaconst
\\
\label{converse-lemma-sumrate}
\mutualInformationConditional{\hat{\channelInOne}, \hat{\channelInTwo}}{\hat{\channelOut}}{\hat{\timeSharingRV}}
+
\lemmaconst
\log
\frac{\lemmaconst}{2\cardinality{\channelOutAlph}}
&\leq
\codebookRateOne + \codebookRateTwo
\\
\label{converse-lemma-rateone}
\mutualInformationConditional{\hat{\channelInOne}}{\hat{\channelOut}}{\hat{\timeSharingRV}}
+
\lemmaconst
\log
\frac{\lemmaconst}{2\cardinality{\channelOutAlph}}
&\leq
\codebookRateOne
\\
\label{converse-lemma-ratetwo}
\mutualInformationConditional{\hat{\channelInTwo}}{\hat{\channelOut}}{\hat{\timeSharingRV}}
+
\lemmaconst
\log
\frac{\lemmaconst}{2\cardinality{\channelOutAlph}}
&\leq
\codebookRateTwo.
\end{align}
\end{lemma}
\begin{proof}[Proof of Lemma~\ref{converse-lemma}]
The codebooks induce random variables $\tilde{\channelInOne}^\codebookBlocklength$, $\tilde{\channelInTwo}^\codebookBlocklength$ and $\tilde{\channelOut}^\codebookBlocklength$ with joint probability mass function
\begin{multline*}
\codebookpmf_{\tilde{\channelInOne}^\codebookBlocklength, \tilde{\channelInTwo}^\codebookBlocklength, \tilde{\channelOut}^\codebookBlocklength}(\channelInOneAlphElement^\codebookBlocklength, \channelInTwoAlphElement^\codebookBlocklength, \channelOutAlphElement^\codebookBlocklength)
=
\codebookpmf_{\channelInOne^\codebookBlocklength, \channelInTwo^\codebookBlocklength, \channelOut^\codebookBlocklength | \codebookOne, \codebookTwo}(\channelInOneAlphElement^\codebookBlocklength, \channelInTwoAlphElement^\codebookBlocklength, \channelOutAlphElement^\codebookBlocklength)
\\
=
\codebookpmf_{\channelInOne^\codebookBlocklength | \codebookOne}(\channelInOneAlphElement^\codebookBlocklength)
\codebookpmf_{\channelInTwo^\codebookBlocklength | \codebookTwo}(\channelInTwoAlphElement^\codebookBlocklength)
\channelpmf_{\channelOut^\codebookBlocklength | \channelInOne^\codebookBlocklength, \channelInTwo^\codebookBlocklength}(\channelOutAlphElement^\codebookBlocklength | \channelInOneAlphElement^\codebookBlocklength, \channelInTwoAlphElement^\codebookBlocklength),
\end{multline*}
where
\begin{align*}
\codebookpmf_{\channelInOne^\codebookBlocklength | \codebookOne}(\channelInOneAlphElement^\codebookBlocklength)
&=
\sum\limits_{\codewordIndex=1}^{\exp(\codebookBlocklength\codebookRateOne)}
  \exp(-\codebookBlocklength\codebookRateOne)
  \indicator{\codebookOneWord{\codewordIndex} = \channelInOneAlphElement^\codebookBlocklength}
\\
\codebookpmf_{\channelInTwo^\codebookBlocklength | \codebookTwo}(\channelInTwoAlphElement^\codebookBlocklength)
&=
\sum\limits_{\codewordIndex=1}^{\exp(\codebookBlocklength\codebookRateTwo)}
  \exp(-\codebookBlocklength\codebookRateTwo)
  \indicator{\codebookTwoWord{\codewordIndex} = \channelInTwoAlphElement^\codebookBlocklength}.
\end{align*}
Note that for notational convenience, we will throughout this proof use the definition of $\codebookpmf_{\tilde{\channelInOne}^\codebookBlocklength, \tilde{\channelInTwo}^\codebookBlocklength, \tilde{\channelOut}^\codebookBlocklength}$ and the definitions implied through the marginals of this probability mass function even if the same probability mass functions have already been introduced using a different notation before; e.g. we will write $\codebookpmf_{\tilde{\channelOut}^\codebookBlocklength}$ instead of $\codebookpmf_{\channelOut^\codebookBlocklength | \codebookOne, \codebookTwo}$.

Additionally, we define a random variable $\hat{\timeSharingRV}$ that is uniformly distributed on $\{1,\dots,\codebookBlocklength\}$ and independent of the other random variables. We furthermore define $\hat{\channelInOne} := \tilde{\channelInOne}_\timeSharingRV$, $\hat{\channelInTwo} := \tilde{\channelInTwo}_\timeSharingRV$ and $\hat{\channelOut} := \tilde{\channelOut}_\timeSharingRV$, noting that the conditional distribution of $\hat{\channelOut}$ given $\hat{\channelInOne}$ and $\hat{\channelInTwo}$ is $\channelpmf_{\channelOut | \channelInOne, \channelInTwo}$.

We observe that for any $\blockIndex \in \{1, \dots, \codebookBlocklength\}$, an application of the triangle inequality yields
\begin{align}
\totalvariation{\codebookpmf_{\tilde{\channelOut}_\blockIndex} - \channelpmf_{\channelOut}}
&=
\frac{1}{2}
\sum_{\channelOutAlphElement \in \channelOutAlph}
  \absolute{
    \codebookpmf_{\tilde{\channelOut}_\blockIndex}(\channelOutAlphElement)
    -
    \channelpmf_{\channelOut}(\channelOutAlphElement)
  }
\nonumber
\\
&=
\frac{1}{2}
\sum_{\channelOutAlphElement \in \channelOutAlph}
  \absolute{
    \sum_{\substack{\channelOutAlphElement^\codebookBlocklength \in \channelOutAlph^\codebookBlocklength \\
                    \channelOutAlphElement_\blockIndex = \channelOutAlphElement}}
    (
      \codebookpmf_{\tilde{\channelOut}^\codebookBlocklength}(\channelOutAlphElement^\codebookBlocklength)
      -
      \channelpmf_{\channelOut^\codebookBlocklength}(\channelOutAlphElement^\codebookBlocklength)
    )
  }
\nonumber
\\
&\leq
\frac{1}{2}
\sum_{\channelOutAlphElement^\codebookBlocklength \in \channelOutAlph^\codebookBlocklength}
  \absolute{
      \codebookpmf_{\tilde{\channelOut}^\codebookBlocklength}(\channelOutAlphElement^\codebookBlocklength)
      -
      \channelpmf_{\channelOut^\codebookBlocklength}(\channelOutAlphElement^\codebookBlocklength)
  }
\nonumber
\\
&=
\totalvariation{\codebookpmf_{\tilde{\channelOut}^\codebookBlocklength} - \channelpmf_{\channelOut^\codebookBlocklength}}
\leq
\lemmaconst.
\label{converse-lemma-totvar-single}
\end{align}
Now, we can bound
\begin{align}
\label{converse-lemma-tvbound-def}
\totalvariationlr{\codebookpmf_{\hat{\channelOut}} - \channelpmf_{\channelOut}}
&=
\totalvariationlr{
  \frac{1}{\codebookBlocklength}
  \sum\limits_{\blockIndex=1}^\codebookBlocklength
  \codebookpmf_{\tilde{\channelOut}_\blockIndex}
  -
  \channelpmf_{\channelOut}
}
\\
\label{converse-lemma-tvbound-triangle}
&\leq
\frac{1}{\codebookBlocklength}
\sum\limits_{\blockIndex=1}^\codebookBlocklength
  \totalvariationlr{
    \codebookpmf_{\tilde{\channelOut}_\blockIndex}
    -
    \channelpmf_{\channelOut}
  }
\\
\label{converse-lemma-tvbound-observation}
&\leq
\lemmaconst,
\end{align}
where (\ref{converse-lemma-tvbound-def}) is by definition of $\hat{\timeSharingRV}$ and $\hat{\channelOut}$, (\ref{converse-lemma-tvbound-triangle}) follows by the triangle inequality and (\ref{converse-lemma-tvbound-observation}) by applying (\ref{converse-lemma-totvar-single}).

Next, we bound
\begin{align}
\label{converse-mutinf-bound}
\codebookBlocklength
\mutualInformationConditional{\hat{\channelInOne}, \hat{\channelInTwo}}{\hat{\channelOut}}{\hat{\timeSharingRV}}
-
\mutualInformation{\tilde{\channelInOne}^\codebookBlocklength, \tilde{\channelInTwo}^\codebookBlocklength}{\tilde{\channelOut}^\codebookBlocklength}
\leq
-
\codebookBlocklength
\lemmaconst
\log
\frac{\lemmaconst}{2\cardinality{\channelOutAlph}}
\end{align}
as shown in (\ref{converse-mutinf-bound-start}) through (\ref{converse-mutinf-bound-end}) in the appendix.

Taking into consideration that
\begin{align*}
\codebookBlocklength
(\codebookRateOne + \codebookRateTwo)
\geq
\entropy{\tilde{\channelInOne}^\codebookBlocklength, \tilde{\channelInTwo}^\codebookBlocklength}
\geq
\mutualInformation{\tilde{\channelInOne}^\codebookBlocklength, \tilde{\channelInTwo}^\codebookBlocklength}{\tilde{\channelOut}^\codebookBlocklength}
\end{align*}
and that $(\codebookBlocklength+1)/\codebookBlocklength \leq 2$, this proves~(\ref{converse-lemma-sumrate}).

Furthermore, we note that
\begin{multline}
\codebookBlocklength
\mutualInformationConditional{\hat{\channelInOne}}{\hat{\channelOut}}{\hat{\timeSharingRV}}
-
\mutualInformation{\tilde{\channelInOne}^\codebookBlocklength}{\tilde{\channelOut}^\codebookBlocklength}
\\
\begin{aligned}
&
\begin{aligned}
=
&
\big(
  \codebookBlocklength
  \mutualInformationConditional{\hat{\channelInOne}, \hat{\channelInTwo}}{\hat{\channelOut}}{\hat{\timeSharingRV}}
  -
  \mutualInformation{\tilde{\channelInOne}^\codebookBlocklength, \tilde{\channelInTwo}^\codebookBlocklength}{\tilde{\channelOut}^\codebookBlocklength}
\big)
\\
&-
\big(
  \codebookBlocklength
  \mutualInformationConditional{\hat{\channelInTwo}}{\hat{\channelOut}}{\hat{\channelInOne}, \hat{\timeSharingRV}}
  -
  \mutualInformationConditional{\tilde{\channelInTwo}^\codebookBlocklength}{\tilde{\channelOut}^\codebookBlocklength}{\tilde{\channelInOne}^\codebookBlocklength}
\big)
\end{aligned}
\\
&\leq
-
\codebookBlocklength
\lemmaconst
\log
\frac{\lemmaconst}{2\cardinality{\channelOutAlph}},
\end{aligned}
\label{converse-lemma-rateone-bound-prereq}
\end{multline}
where the inequality step follows by~(\ref{converse-mutinf-bound}) and the observation that
$
  \codebookBlocklength
  \mutualInformationConditional{\hat{\channelInTwo}}{\hat{\channelOut}}{\hat{\channelInOne}, \hat{\timeSharingRV}}
  -
  \mutualInformationConditional{\tilde{\channelInTwo}^\codebookBlocklength}{\tilde{\channelOut}^\codebookBlocklength}{\tilde{\channelInOne}^\codebookBlocklength}
$
is nonnegative, because
\begin{align}
&\hphantom{{}={}}
\mutualInformationConditional{\tilde{\channelInTwo}^\codebookBlocklength}{\tilde{\channelOut}^\codebookBlocklength}{\tilde{\channelInOne}^\codebookBlocklength}
\nonumber
\\
&=
\sum_{\blockIndex=1}^\codebookBlocklength
  \Big(
    \entropyConditional{\tilde{\channelOut}_\blockIndex}{\tilde{\channelInOne}^\codebookBlocklength, \tilde{\channelOut}^{\blockIndex-1}}
    -
    \entropyConditional{\tilde{\channelOut}_\blockIndex}{\tilde{\channelInOne}^\codebookBlocklength, \tilde{\channelInTwo}^\codebookBlocklength, \tilde{\channelOut}^{\blockIndex-1}}
  \Big)
\label{converse-lemma-rateone-ibound-chain-rule}
\\
&=
\sum_{\blockIndex=1}^\codebookBlocklength
  \Big(
    \entropyConditional{\tilde{\channelOut}_\blockIndex}{\tilde{\channelInOne}^\codebookBlocklength, \tilde{\channelOut}^{\blockIndex-1}}
    -
    \entropyConditional{\tilde{\channelOut}_\blockIndex}{\tilde{\channelInOne}_\blockIndex, \tilde{\channelInTwo}_\blockIndex}
  \Big)
\label{converse-lemma-rateone-ibound-conditional-independence}
\\
&\leq
\sum_{\blockIndex=1}^\codebookBlocklength
  \Big(
    \entropyConditional{\tilde{\channelOut}_\blockIndex}{\tilde{\channelInOne}_\blockIndex}
    -
    \entropyConditional{\tilde{\channelOut}_\blockIndex}{\tilde{\channelInOne}_\blockIndex, \tilde{\channelInTwo}_\blockIndex}
  \Big)
\label{converse-lemma-rateone-ibound-conditioning-entropy}
\\
&=
\codebookBlocklength
\mutualInformationConditional{\hat{\channelInTwo}}{\hat{\channelOut}}{\hat{\channelInOne}, \hat{\timeSharingRV}},
\nonumber
\end{align}
where (\ref{converse-lemma-rateone-ibound-chain-rule}) follows by the entropy chain rule, (\ref{converse-lemma-rateone-ibound-conditional-independence}) follows from the fact that $\tilde{\channelOut}_\blockIndex$ is conditionally independent of the other random variables given $\tilde{\channelInOne}_\blockIndex$ and $\tilde{\channelInTwo}_\blockIndex$ and (\ref{converse-lemma-rateone-ibound-conditioning-entropy}) follows because conditioning does not increase entropy.

Analogous reasoning exchanging the roles of $\channelInOne$ and $\channelInTwo$ yields
\begin{align}
\label{converse-lemma-ratetwo-bound-prereq}
\codebookBlocklength
\mutualInformationConditional{\hat{\channelInTwo}}{\hat{\channelOut}}{\hat{\timeSharingRV}}
-
\mutualInformation{\tilde{\channelInOne}^\codebookBlocklength}{\tilde{\channelOut}^\codebookBlocklength}
\leq
-
\codebookBlocklength
\lemmaconst
\log
\frac{\lemmaconst}{2\cardinality{\channelOutAlph}}.
\end{align}

In a similar way as we obtained~(\ref{converse-lemma-sumrate}) from~(\ref{converse-mutinf-bound}), we can derive the bounds (\ref{converse-lemma-rateone}) and (\ref{converse-lemma-ratetwo}) from (\ref{converse-lemma-rateone-bound-prereq}) and (\ref{converse-lemma-ratetwo-bound-prereq}), respectively.
\end{proof}

In order to ensure convergence in the limit $\codebookBlocklength \rightarrow \infty$, we need a slight refinement of Lemma~\ref{converse-lemma}.

\begin{cor}
\label{converse-cor}
Lemma~\ref{converse-lemma} holds with $\timeSharingAlph$ replaced by some (fixed) $\timeSharingAlph'$ with $\cardinality{\timeSharingAlph'} \leq \cardinality{\channelOutAlph} + 3$.
\end{cor}
\begin{proof}
We apply Lemma~\ref{convex-cover-lemma} with $\generalrvOneAlph := \channelInOneAlph \times \channelInTwoAlph$, $\generalrvTwoAlph := \timeSharingAlph$, letting $\generalpmfset$ be the set of all probability distributions on $ \channelInOneAlph \times \channelInTwoAlph$ corresponding to independent random variables $\channelInOne$ on $\channelInOneAlph$ and $\channelInTwo$ on $\channelInTwoAlph$. Define
\[
\generalpmf_\timeSharingAlphElement(\channelInOneAlphElement, \channelInTwoAlphElement) := \codebookpmf_{\hat{\channelInOne} | \hat{\timeSharingRV}}(\channelInOneAlphElement | \timeSharingAlphElement) \codebookpmf_{\hat{\channelInTwo} | \hat{\timeSharingRV}}(\channelInTwoAlphElement | \timeSharingAlphElement)
\]
and
\begin{align*}
\generalfunction_1(\generalpmf_\timeSharingAlphElement)
&:=
\mutualInformationConditional{\hat{\channelInOne}}{\hat{\channelOut}}{\hat{\timeSharingRV} = \timeSharingAlphElement}
\\
\generalfunction_2(\generalpmf_\timeSharingAlphElement)
&:=
\mutualInformationConditional{\hat{\channelInTwo}}{\hat{\channelOut}}{\hat{\timeSharingRV} = \timeSharingAlphElement}
\\
\generalfunction_3(\generalpmf_\timeSharingAlphElement)
&:=
\mutualInformationConditional{\hat{\channelInOne},\hat{\channelInTwo}}{\hat{\channelOut}}{\hat{\timeSharingRV} = \timeSharingAlphElement}
\\
\generalfunction_{4, \channelOutAlphElement}(\generalpmf_\timeSharingAlphElement)
&:=
\sum\limits_{\channelInOneAlphElement \in \channelInOneAlph}
\sum\limits_{\channelInTwoAlphElement \in \channelInTwoAlph}
  \generalpmf_\timeSharingAlphElement(\channelInOneAlphElement, \channelInTwoAlphElement)
  \channelpmf_{\channelOut | \channelInOne, \channelInTwo}(\channelOutAlphElement | \channelInOneAlphElement, \channelInTwoAlphElement),
\end{align*}
where $\generalfunction_{4,\channelOutAlphElement}$ defines a set of $\cardinality{\channelOutAlph}$ functions. Note that the mutual information terms depend only on $\generalpmf_\timeSharingAlphElement$ and are continuous. By Lemma~\ref{convex-cover-lemma}, we have a random variable $\timeSharingRV'$ distributed according to $\generalpmf$ on some alphabet of cardinality at most $\cardinality{\channelOutAlph} + 3$ and probability mass functions $(\generalpmf'_{\timeSharingAlphElement'})_{\timeSharingAlphElement' \in \timeSharingAlph'}$ corresponding to random variables independently distributed on $\channelInOneAlph$ and $\channelInTwoAlph$. These probability mass functions induce random variables $\channelInOne'$ and $\channelInTwo'$ that depend on $\timeSharingRV'$, but are conditionally independent given $\timeSharingRV'$. We further define a random variable $\channelOut'$ on $\channelOutAlph$ depending on $\channelInOne'$ and $\channelInTwo'$ through $\channelpmf_{\channelOut | \channelInOne, \channelInTwo}$. We have, by (\ref{convex-cover-lemma-function-property}),
\begin{align*}
\mutualInformationConditional{\channelInOne', \channelInTwo'}{\channelOut'}{\timeSharingRV'}
&=
\sum\limits_{\timeSharingAlphElement' \in \timeSharingAlph'}
\generalpmf(\timeSharingAlphElement')
\mutualInformationConditional{\channelInOne', \channelInTwo'}{\channelOut'}{\timeSharingRV'=\timeSharingAlphElement'}
\\
&=
\sum\limits_{\blockIndex=1}^\codebookBlocklength
  \frac{1}{\codebookBlocklength}
  \mutualInformationConditional{\hat{\channelInOne}, \hat{\channelInTwo}}{\hat{\channelOut}}{\hat{\timeSharingRV} = \blockIndex}
\\
&=
\mutualInformationConditional{\hat{\channelInOne}, \hat{\channelInTwo}}{\hat{\channelOut}}{\hat{\timeSharingRV}},
\end{align*}
and, by similar arguments, $\mutualInformationConditional{\channelInOne'}{\channelOut'}{\timeSharingRV'} = \mutualInformationConditional{\hat{\channelInOne}}{\hat{\channelOut}}{\hat{\timeSharingRV}}$ and $\mutualInformationConditional{\channelInTwo'}{\channelOut'}{\timeSharingRV'} = \mutualInformationConditional{\hat{\channelInTwo}}{\hat{\channelOut}}{\hat{\timeSharingRV}}$, so properties (\ref{converse-lemma-sumrate}), (\ref{converse-lemma-rateone}) and (\ref{converse-lemma-ratetwo}) are preserved for the new random variables. Furthermore, also by (\ref{convex-cover-lemma-function-property}), we have for each $\channelOutAlphElement \in \channelOutAlph$,
\begin{align*}
\codebookpmf_{\channelOut'}(\channelOutAlphElement)
&=
\sum\limits_{\timeSharingAlphElement' \in \timeSharingAlph'}
  \generalpmf(\timeSharingAlphElement')
  \sum\limits_{\channelInOneAlphElement \in \channelInOneAlph}
  \sum\limits_{\channelInTwoAlphElement \in \channelInTwoAlph}
    \generalpmf'_{\timeSharingAlphElement'}(\channelInOneAlphElement, \channelInTwoAlphElement)
    \channelpmf_{\channelOut | \channelInOne, \channelInTwo}(\channelOutAlphElement | \channelInOneAlphElement, \channelInTwoAlphElement)
\\
&=
\sum\limits_{\blockIndex=1}^\codebookBlocklength
    \frac{1}{n}
  \sum\limits_{\channelInOneAlphElement \in \channelInOneAlph}
  \sum\limits_{\channelInTwoAlphElement \in \channelInTwoAlph}
    \generalpmf_\timeSharingAlphElement(\channelInOneAlphElement, \channelInTwoAlphElement)
    \channelpmf_{\channelOut | \channelInOne, \channelInTwo}(\channelOutAlphElement | \channelInOneAlphElement, \channelInTwoAlphElement)
\\
&=
\codebookpmf_{\hat{\channelOut}}(\channelOutAlphElement),
\end{align*}
implying that (\ref{converse-lemma-convergence}) is preserved and completing the proof of the corollary.
\end{proof}

We are now ready to put everything together and prove our main converse result.
\begin{proof}[Proof of Theorem~\ref{converse-theorem}]
We write
\[
\lemmaconst_\codebookBlocklength
:=
\totalvariation{\codebookpmf_{\channelOut^\codebookBlocklength | \codebook_{1,\codebookBlocklength}, \codebook_{2,\codebookBlocklength}} - \channelpmf_{\channelOut^\codebookBlocklength}},
\]
and apply Corollary~\ref{converse-cor} with $\codebookOne := \codebook_{1,\codebookBlocklength}$, $\codebookTwo := \codebook_{2,\codebookBlocklength}$ and $\lemmaconst := \lemmaconst_\codebookBlocklength$ to obtain random variables $\channelInOne_\codebookBlocklength$, $\channelInTwo_\codebookBlocklength$, $\channelOut_\codebookBlocklength$ and $\timeSharingRV_\codebookBlocklength$ satisfying conditions (\ref{converse-lemma-convergence}), (\ref{converse-lemma-sumrate}), (\ref{converse-lemma-rateone}) and (\ref{converse-lemma-ratetwo}). We assume without loss of generality that $\timeSharingRV_\codebookBlocklength$ are all on the same alphabet $\timeSharingAlph$ with cardinality $\cardinality{\channelOutAlph} + 3$ and observe that the space of probability distributions on a finite alphabet is compact. So we can pick an increasing sequence $(\codebookBlocklength_\generalLimitIndex)_{\generalLimitIndex \geq 1}$ such that $\codebookpmf_{\timeSharingRV_{\codebookBlocklength_\generalLimitIndex}}$ converges to a probability mass function $\codebookpmf_{\timeSharingRV}$, $\codebookpmf_{\channelInOne_{\codebookBlocklength_\generalLimitIndex} |\timeSharingRV_{\codebookBlocklength_\generalLimitIndex}=\timeSharingAlphElement}$ converges for all $\timeSharingAlphElement \in \timeSharingAlph$ to a probability mass function $\codebookpmf_{\channelInOne|\timeSharingRV=\timeSharingAlphElement}$ and $\codebookpmf_{\channelInTwo_{\codebookBlocklength_\generalLimitIndex}|\timeSharingRV_{\codebookBlocklength_\generalLimitIndex}=\timeSharingAlphElement}$ converges for all $\timeSharingAlphElement \in \timeSharingAlph$ to a probability mass function $\codebookpmf_{\channelInTwo|\timeSharingRV=\timeSharingAlphElement}$. The random variables $\channelInOne$, $\channelInTwo$, $\channelOut$ and $\timeSharingRV$ are then defined by the joint probability mass function $\codebookpmf_\timeSharingRV \codebookpmf_{\channelInOne | \timeSharingRV} \codebookpmf_{\channelInTwo | \timeSharingRV} \channelpmf_{\channelOut | \channelInOne, \channelInTwo}$. (\ref{converse-lemma-convergence}) implies that the marginal distribution of $\channelOut$ is $\channelpmf_\channelOut$, (\ref{converse-lemma-sumrate}) implies (\ref{converse-theorem-sumrate}), (\ref{converse-lemma-rateone}) implies (\ref{converse-theorem-rateone}) and (\ref{converse-lemma-ratetwo}) implies (\ref{converse-theorem-ratetwo}).
\end{proof}

\section{Channel Resolvability and Secrecy}
\label{sec:secrecy}

In this section, we summarize some existing notions of secrecy and a
few facts about their relations to each other, and also discuss their
operational meaning. We then show explicitly how our channel
resolvability results can be applied to derive a secrecy result for
the multiple-access wiretap channel that not only implies strong
secrecy, but has even stronger operational properties than the
classical notion of strong secrecy. We start by providing additional
definitions in Subsection~\ref{sec:secrecy-addiditonal-def}, then we
discuss different notions of secrecy and interconnections between them
in Subsection~\ref{sec:secrecy-concepts}. Then we show how our
resolvability results can be used to achieve secrecy over the
multiple-access channel. To this end, we first repeat a well-known
result on achievable rates for the multiple-access channel in
Subsection~\ref{sec:secrecy-mac-coding} which is the second crucial
ingredient besides resolvability for the secrecy theorem we prove in
Subsection~\ref{sec:secrecy-achieve}.

\subsection{Additional Definitions}
\label{sec:secrecy-addiditonal-def}
In contrast to the other sections, in this section we will also look at probability distributions on the alphabet of possible messages. We therefore use $\messageAlphabet_1 := \{1, \dots, \exp(\codebookBlocklength\codebookRateOne)\}$ and $\messageAlphabet_2 := \{1, \dots, \exp(\codebookBlocklength\codebookRateTwo)\}$ to denote the message alphabets; for convenience we denote by $\messageAlphabet := \messageAlphabet_1 \times \messageAlphabet_2$ the message alphabet of both transmitters. $\messageRV_1$, $\messageRV_2$ and $\messageRV = (\messageRV_1, \messageRV_2)$ are random variables on $\messageAlphabet_1$, $\messageAlphabet_2$ and $\messageAlphabet$, respectively, to denote the messages chosen by the transmitters to be sent and $\codebookpmf_\messageRV$ denotes a probability mass function on $\messageAlphabet$ describing the distribution of $\messageRV$.

By a \emph{wiretap channel} $\channel=(\channelLegit, \channelWiretapper)$, we mean a pair of channels, where $\channelLegit=(\channelInOneAlph, \channelInTwoAlph, \channelOutAlphLegit, \channelpmf_{\channelOutLegit | \channelInOne, \channelInTwo})$, which we call the \emph{legitimate receiver's channel} and $\channelWiretapper=(\channelInOneAlph, \channelInTwoAlph, \channelOutAlphWiretapper, \channelpmf_{\channelOutWiretapper | \channelInOne, \channelInTwo})$, which we call the \emph{wiretapper's channel}, share the same input alphabets.

A \emph{wiretap code} of block length $\codebookBlocklength$ for the multiple-access channel consists of two collections of probability mass functions $\codebookpmf_{\channelInOne^\codebookBlocklength | \messageRV_1}$ for each $\messageAlphabetElement_1 \in \messageAlphabet_1$ and $\codebookpmf_{\channelInTwo^\codebookBlocklength | \messageRV_2}$ for each $\messageAlphabetElement_2 \in \messageAlphabet_2$ (called the \emph{encoders}), as well as a \emph{decoder} $\codebookDecoder_\codebookBlocklength: \channelOutAlphLegit^\codebookBlocklength \rightarrow \messageAlphabet$. The encoders and the wiretap channel together induce for each $\messageAlphabetElement \in \messageAlphabet$ a \emph{wiretapper's output distribution} on $\channelOutAlphWiretapper^\codebookBlocklength$.

In this paper, we will consider a special class of wiretap encoders: By a pair of \emph{wiretap codebooks} of block length $\codebookBlocklength$, information rates $\codebookRateOne$, $\codebookRateTwo$ and randomness rates $\codebookRandRateOne$, $\codebookRandRateTwo$, we mean functions $\codebookOne: \{1, \dots, \exp(\codebookBlocklength\codebookRateOne)\} \times \{1, \dots, \exp(\codebookBlocklength\codebookRandRateOne)\} \rightarrow \channelInOneAlph^\codebookBlocklength$ and $\codebookTwo: \{1, \dots, \exp(\codebookBlocklength\codebookRateTwo)\} \times \{1, \dots, \exp(\codebookBlocklength\codebookRandRateTwo)\} \rightarrow \channelInTwoAlph^\codebookBlocklength$. Given input distributions $\channelpmf_\channelInOne$ and $\channelpmf_\channelInTwo$, we define a probability distribution $\Probability_{\codebookOne, \codebookTwo}$ on the set of all possible codebooks such that every component of every function value is independently distributed according to $\channelpmf_\channelInOne$ or $\channelpmf_\channelInTwo$, respectively. Such a pair of codebooks induces encoders in the sense defined above: For $\txIndex \in \{1,2\}$, transmitter $\txIndex$ picks (independently of the other transmitter and without cooperation) a message $\codewordIndex_\txIndex \in \{1, \dots, \exp(\codebookBlocklength\codebookRate_\txIndex)\}$ and draws uniformly at random $\randomnessIndex_\txIndex \in \{1, \dots, \exp(\codebookBlocklength\codebookRandRate_\txIndex)\}$. The transmitted codeword is then $\codebook_\txIndex(\codewordIndex_\txIndex, \randomnessIndex_\txIndex)$. The wiretapper's output distribution induced by these encoders is
\begin{multline*}
\codebookpmf_{\channelOutWiretapper^\codebookBlocklength | \codebookOne(\codewordIndex_1, \cdot), \codebookTwo(\codewordIndex_2, \cdot)}(\channelOutAlphElementWiretapper^\codebookBlocklength)
:=
\exp(-\codebookBlocklength(\codebookRandRateOne+\codebookRandRateTwo))
\\
\sum\limits_{\randomnessIndex_1=1}^{\exp(\codebookBlocklength\codebookRandRateOne)}
\sum\limits_{\randomnessIndex_2=1}^{\exp(\codebookBlocklength\codebookRandRateTwo)}
  \channelpmf_{\channelOutWiretapper^\codebookBlocklength | \channelInOne^\codebookBlocklength, \channelInTwo^\codebookBlocklength} (\channelOutAlphElementWiretapper^\codebookBlocklength | \codebookOne(\codewordIndex_1, \randomnessIndex_1), \codebookTwo(\codewordIndex_2, \randomnessIndex_2)).
\end{multline*}

\subsection{Secrecy Concepts and their Operational Meanings}
\label{sec:secrecy-concepts}
In this subsection, we give a short overview over some notions of information-theoretic secrecy. We adapt the notation in the sources we are citing in such a way that it conforms with the notation we have been using in this paper, which is specific to the multiple-access wiretap channel, but the concepts apply also to (and have originally been introduced for) wiretap channels with only one transmitter.

Before we give formal definitions of secrecy, we first take a look at what operational implications we want a good notion of secrecy to have. In~\cite{AhlswedeGeneral}, Ahlswede generalized the communication task in the following way: The transmitter picks a message $\messageRV \in \messageAlphabet$, encodes it and transmits the codeword over a channel. The receiver then aims not necessarily at decoding the exact message $\messageRV$, but rather has a partition $\{\messageAlphabet_1, \dots, \messageAlphabet_\nPartitionElements\}$ of $\messageAlphabet$, i.e. disjoint subsets of $\messageAlphabet$ with $\cup_{\indexPartitions = 1}^{\nPartitionElements} \messageAlphabet_\indexPartitions = \messageAlphabet$. Its task is then to reconstruct $\indexPartitions$ such that $\messageRV \in \messageAlphabet_\indexPartitions$. The transmitter does not know the exact partition $\partition$ that defines the receiver's task, but has only a subset $\{\partition_1, \dots, \partition_\nPartitions\}$ of the set of all possible partitions of $\messageAlphabet$ and has to encode the message in such a way that the receiver can decode with high probability no matter which partition in this set defines its task, or put differently, we assume there is one receiver for each partition and we want each receiver to be able to decode with high probability. Ahlswede gave a list of specializations of this very general concept that he deemed potentially useful for real-world communication systems. Two of these are:
\begin{enumerate}
 \item $\nPartitions = 1$ and the only partition $\partition_1$ is the partition of singletons $\{\{\messageAlphabetElement\}: \messageAlphabetElement \in \messageAlphabet\}$. Here the receiver's task amounts to Shannon's classical decoding problem.\label{ahlswede-classification-shannon}
 \item $\nPartitions = \cardinality{\messageAlphabet}$ and the set of partitions is
 \[\Big\{\big\{\{\messageAlphabetElement\}, \messageAlphabet \setminus \{\messageAlphabetElement\}\big\}: \messageAlphabetElement \in \messageAlphabet\Big\}.\]
 This corresponds to the problem of identification codes introduced by Ahlswede and Dueck~\cite{AhlswedeDueckIdentification}.\label{ahlswede-classification-identification}
\end{enumerate}
These decoding tasks can not only be understood as having to be performed by a legitimate receiver in cooperation with the transmitter, but can also be seen as possible attacks carried out by a wiretapper, where each possible attack is defined by a partition on $\messageAlphabet$. For instance, we could deem the wiretapper's attack to have been successful only if it is able to decode the exact message which corresponds to item \ref{ahlswede-classification-shannon}), or, by an alternative definition, we might choose a particular message $\messageAlphabetElement$ that the wiretapper has to identify corresponding to one of the partitions of item \ref{ahlswede-classification-identification}) and thus obtain a much weaker definition of what constitutes a successful attack. 

The best thing we might wish for in this context would be that the wiretapper be unable to perform attacks corresponding to \emph{any} possible partition of the message set. In fact, demanding this is equivalent to insisting that the wiretapper shall not be able to decode any function of the transmitted message (because a function induces a partition of its domain via the preimages of singletons in its range).
 
Shannon introduced the concept of information theoretic secrecy and required for his condition of \emph{perfect secrecy}~\cite{ShannonSecrecy} that $\mutualInformation{\messageRV}{\channelOutWiretapper^\codebookBlocklength} = 0$. In Shannon's cipher system, both legitimate receiver and wiretapper get the same channel output, but transmitter and legitimate receiver share a key, which is unknown to the wiretapper.

Later, Wyner introduced the wiretap channel, where the wiretapper's channel output is degraded with respect to the legitimate receiver's channel output, but wiretapper and legitimate receiver do not share a key. In order to be more amenable to analysis, he weakened Shannon's perfect secrecy criterion and introduced the concept of \emph{weak secrecy}~\cite{WynerWiretap}. A sequence of codes for increasing block length $\codebookBlocklength$ is said to achieve weak secrecy if $\mutualInformation{\messageRV}{\channelOutWiretapper^\codebookBlocklength}/\codebookBlocklength$ approaches $0$ as $\codebookBlocklength$ tends to infinity, where $\messageRV$ is assumed to be uniformly distributed over the message set. It turns out that a sequence of codes can achieve weak secrecy while the total amount of information leaked still tends to infinity with the block length approaching infinity. Therefore, the notion of~\emph{strong secrecy} was introduced by Maurer~\cite{MaurerStrongSecret} which demands that $\mutualInformation{\messageRV}{\channelOutWiretapper^\codebookBlocklength}$ shall approach $0$ as $\codebookBlocklength$ tends to infinity, where $\messageRV$ is still assumed to be distributed uniformly.

Strong secrecy has been shown to imply that full decoding is only possible with an error probability that approaches the probability of guessing correctly without any channel output at all and it has also been shown to provide some resilience against identification attacks~\cite{BjelakovicSecrecy}, however, some messages might remain identifiable and no method was proposed of ensuring that any particular given message $\messageAlphabetElement$ not be identifiable.

This leads us to the concept of \emph{semantic security}, which originates from the field of cryptography~\cite{GoldwasserProbabilistic}, but has also been redefined to yield an information-theoretic secrecy criterion~\cite{BellareSemantic}\cite{BellareCryptographic}.

Following the discussion in~\cite{BellareCryptographic} and adapting to our notation, we say that a sequence of codebooks achieves \emph{semantic security} if
\begin{multline}
\label{def:semantic-security}
\max\limits_{\partition, \Probability_\messageRV} \bigg(
  \max\limits_{\wiretapperDecoder: \channelOutAlphWiretapper^\codebookBlocklength \rightarrow \partition}
  \Probability_{\messageRV, \channelOutWiretapper^\codebookBlocklength}(\messageRV \in \wiretapperDecoder(\channelOutWiretapper^\codebookBlocklength))
  \\
  -
  \max\limits_{\wiretapperGuesser \in \partition}
  \Probability_{\messageRV}(\messageRV \in \wiretapperGuesser)
\bigg)
\end{multline}
tends to $0$ with $\codebookBlocklength \rightarrow \infty$, where $\partition$ ranges over all partitions of $\messageAlphabet$ and $\Probability_\messageRV$ over all possible probability distributions of $\messageRV \in \messageAlphabet$. The positive term can be seen as the probability that the wiretapper solves its decoding task successfully given its channel output and an optimal decoder, while the second term can be seen as subtracting the probability that the decoding task is solved without knowledge of the wiretapper's channel output (i.e. by optimal guessing). This notion has a wide range of operational implications in the sense that it is impossible for the wiretapper to carry out an attack corresponding to any of Ahlswede's general communication tasks, as we show in Examples~\ref{example:semantic-security-id} and~\ref{example:semantic-security-general}.
\begin{remark}
(\ref{def:semantic-security}) remains unchanged even if we allow stochastic wiretapper's decoders, i.e. if we let $\wiretapperDecoder$ range over functions that map from $\channelOutAlphWiretapper^\codebookBlocklength$ to probability mass functions on $\partition$.
\end{remark}

On the other hand, it is not easy to prove directly that a particular wiretap code achieves semantic security. We therefore define a notion of secrecy that is more amenable to analysis, namely, we say that a sequence of wiretap codebooks achieves \emph{distinguishing security}  if
\begin{multline}
\label{def:distinguishing-security}
\max\limits_{\codewordIndex_1, \codewordIndex_2, \tilde{\codewordIndex}_1, \tilde{\codewordIndex}_2}
  \totalvariation{
    \codebookpmf_{\channelOutWiretapper^\codebookBlocklength | \codebookOne(\codewordIndex_1, \cdot), \codebookTwo(\codewordIndex_2, \cdot)}
    \\
    -
    \codebookpmf_{\channelOutWiretapper^\codebookBlocklength | \codebookOne(\tilde{\codewordIndex}_1, \cdot), \codebookTwo(\tilde{\codewordIndex}_2, \cdot)}
  }
\end{multline}
approaches $0$ with $\codebookBlocklength$ tending to infinity.
\begin{remark}
Due to the triangle inequality, this notion of distinguishing security is equivalent to
\begin{align*}
\max\limits_{\codewordIndex_1, \codewordIndex_2}
  \totalvariation{
    \codebookpmf_{\channelOutWiretapper^\codebookBlocklength | \codebookOne(\codewordIndex_1, \cdot), \codebookTwo(\codewordIndex_2, \cdot)}
    -
    \codebookpmf
  }
\end{align*}
approaching $0$, where $\codebookpmf$ is an arbitrary probability mass function on $\channelOutAlphWiretapper^\codebookBlocklength$ that does not change with $\codewordIndex_1$ or $\codewordIndex_2$.
\end{remark}

In \cite{BellareCryptographic}, these and several other related concepts of information-theoretical secrecy and their interrelations are discussed. It turns out that distinguishing and semantic security are equivalent and strictly stronger than strong secrecy. On the other hand, if the definition of strong secrecy is strengthened to demand that $\max_{\codebookpmf_\messageRV} \mutualInformation{\messageRV}{\channelOutWiretapper^\codebookBlocklength}$ approaches $0$, where $\codebookpmf_\messageRV$ ranges over \emph{all} distributions of $\messageRV$ (not just uniform), then it becomes equivalent to the notions of distinguishing and semantic security. We will present here versions of the proofs adapted to our notation only for the implications that seem most important to us in this context, namely that distinguishing security, which we show in Theorem~\ref{theorem:distinguishing-security} can be achieved over the multiple-access channel, implies both strong secrecy (Lemma~\ref{lemma:distinguishing-implies-strong}) and semantic security (Lemma~\ref{lemma:distinguishing-implies-semantic}).

We now give two examples to show how the concept of semantic security implies resilience to the Ahlswede attacks discussed before. Example~\ref{example:semantic-security-id} is a special case of Example~\ref{example:semantic-security-general}, but we lay it out explicitly anyway because we believe that it illustrates the ideas involved in a particularly accessible way and that the type of attack it describes could be particularly relevant in practice.
\begin{example}[Identification attacks]
\label{example:semantic-security-id}
Suppose the wiretapper's objective is to determine from an observation of $\channelOutWiretapper^\codebookBlocklength$ whether the transmitted message was $\messageAlphabetElement$ or not, i.e. the attack we consider can be described by the partition $\partition = \{\{\messageAlphabetElement\}, \messageAlphabet \setminus \{\messageAlphabetElement\}\}$. To this end, we assume that it has found some decoding function $\wiretapperDecoder: \channelOutAlphWiretapper^\codebookBlocklength \rightarrow \partition$. Without assuming any probability distribution on $\messageAlphabet$, we can define two types of error
\begin{align*}
\errorprob_1
&:=
\Probability_{\channelOutWiretapper^\codebookBlocklength} \left(
  \wiretapperDecoder(\channelOutWiretapper^\codebookBlocklength) = \messageAlphabet \setminus \{\messageAlphabetElement\}
  \big|
  \messageRV = \messageAlphabetElement
\right)
\\
\errorprob_2^{\hat{\messageAlphabetElement}}
&:=
\Probability_{\channelOutWiretapper^\codebookBlocklength} \left(
  \wiretapperDecoder(\channelOutWiretapper^\codebookBlocklength) = \{\messageAlphabetElement\}
  \big|
  \messageRV = \hat{\messageAlphabetElement}
\right),
\end{align*}
where $\errorprob_2^{\hat{\messageAlphabetElement}}$ is defined for each $\hat{\messageAlphabetElement} \neq \messageAlphabetElement$. Note that if the wiretapper is allowed to guess stochastically, it can trivially achieve any values for $\errorprob_1$ and $\errorprob_2^{\hat{\messageAlphabetElement}}$ that satisfy $\errorprob_1 + \errorprob_2^{\hat{\messageAlphabetElement}} = 1$ (where $\errorprob_2^{\hat{\messageAlphabetElement}}$ can be chosen the same for all $\hat{\messageAlphabetElement}$), even without observing $\channelOutWiretapper^\codebookBlocklength$. On the other hand, we now show that if (\ref{def:semantic-security}) is bounded above by $\lemmaconst > 0$ and we fix an arbitrary $\hat{\messageAlphabetElement} \neq \messageAlphabetElement$, the wiretapper cannot achieve much lower error probabilities $\errorprob_1$ and $\errorprob_2^{\hat{\messageAlphabetElement}}$, even given the observation of $\channelOutAlphWiretapper^\codebookBlocklength$. We first note that if $\Probability_\messageRV$ is defined by the probability mass function
\[
\codebookpmf_\messageRV(\tilde{\messageAlphabetElement}) :=
\begin{cases}
\frac{1}{2}, & \tilde{\messageAlphabetElement} = \messageAlphabetElement\ \vee
               \tilde{\messageAlphabetElement} = \hat{\messageAlphabetElement}
\\
0,           & \text{otherwise,}
\end{cases}
\]
then $\Probability_{\messageRV}(\messageRV \in \wiretapperGuesser) = 1/2$ regardless of the choice of $\wiretapperGuesser$. We can thus bound
\begin{align}
\label{example:semantic-security-id-def}
\lemmaconst
&>
\Probability_{\messageRV, \channelOutWiretapper^\codebookBlocklength}(\messageRV \in \wiretapperDecoder(\channelOutWiretapper^\codebookBlocklength))
-
\max\limits_{\wiretapperGuesser \in \partition}
\Probability_{\messageRV}(\messageRV \in \wiretapperGuesser)
\\
\label{example:semantic-security-id-prob-substitution}
&
\begin{aligned}
=
&\frac{1}{2}
\Probability_{\channelOutWiretapper^\codebookBlocklength}(\messageAlphabetElement \in \wiretapperDecoder(\channelOutWiretapper^\codebookBlocklength) | \messageRV = \messageAlphabetElement)
\\
&+
\frac{1}{2}
\Probability_{\channelOutWiretapper^\codebookBlocklength}(\hat{\messageAlphabetElement} \in \wiretapperDecoder(\channelOutWiretapper^\codebookBlocklength) | \messageRV = \hat{\messageAlphabetElement})
-
\frac{1}{2}
\end{aligned}
\\
\label{example:semantic-security-id-errorprob-def}
&=
\frac{1}{2}(1-\errorprob_1 + 1-\errorprob_2^{\hat{\messageAlphabetElement}} - 1),
\end{align}
where (\ref{example:semantic-security-id-def}) is by assumption valid for all $\Probability_\messageRV$, (\ref{example:semantic-security-id-prob-substitution}) follows by plugging in $\Probability_\messageRV$ as defined above and (\ref{example:semantic-security-id-errorprob-def}) follows by plugging in the definitions of $\errorprob_1$ and $\errorprob_2^{\hat{\messageAlphabetElement}}$. We can express this equivalently as
\[
\errorprob_1 + \errorprob_2^{\hat{\messageAlphabetElement}} > 1 - 2\lemmaconst,
\]
which shows that as (\ref{def:semantic-security}) tends to $0$, the wiretapper's minimum error probabilities approach values that can be achieved by stochastic guessing without observing any channel output at all.
\end{example}

\begin{example}[General Ahlswede Attacks]
\label{example:semantic-security-general}
As a second, slightly more involved example, we show the operational meaning of semantic security for attacks defined by more general partitions. Let $\partition = \{\messageAlphabet_1, \dots, \messageAlphabet_\nPartitionElements\}$ be a partition of $\messageAlphabet$. Instead of assuming that the wiretapper can determine the correct partition element based on any transmitted message, we make a weaker assumption, namely that there is (at least) one message in each partition element that can be correctly identified. Formally, we say that the attacker has a decoding function $\wiretapperDecoder: \channelOutAlphWiretapper^\codebookBlocklength \rightarrow \partition$ and we fix messages $\messageAlphabetElement_1 \in \messageAlphabet_1, \dots, \messageAlphabetElement_\nPartitionElements \in \messageAlphabet_\nPartitionElements$. For all $\indexPartitions \in \{1, \dots, \nPartitionElements\}$, the error probabilities associated with partition element $\messageAlphabet_\indexPartitions$ and message $\messageAlphabetElement_\indexPartitions \in \messageAlphabet_\indexPartitions$ is then defined as
\[
\errorprob_\indexPartitions^{\messageAlphabetElement_\indexPartitions}
:=
\Probability_{\channelOutWiretapper^\codebookBlocklength} \left(
  \wiretapperDecoder(\channelOutWiretapper^\codebookBlocklength) \neq \messageAlphabet_\indexPartitions
  \big|
  \messageRV = \messageAlphabetElement_\indexPartitions
\right).
\]
We observe that if stochastic guessing is allowed, the wiretapper can achieve any combination of error probabilities that satisfies $\errorprob_1^{\messageAlphabetElement_1} + \dots + \errorprob_\nPartitionElements^{\messageAlphabetElement_\nPartitionElements} = \nPartitionElements - 1$. We consider $\Probability_\messageRV$ defined by the probability mass function
\[
\codebookpmf_\messageRV(\messageAlphabetElement) :=
\begin{cases}
\frac{1}{\nPartitionElements}, & \messageAlphabetElement = \messageAlphabetElement_\indexPartitions \text{ for some } \indexPartitions
\\
0,           & \text{otherwise,}
\end{cases}
\]
noting that $\Probability_{\messageRV}(\messageRV \in \wiretapperGuesser) = 1/\nPartitionElements$ regardless of the choice of $\wiretapperGuesser$. Assuming that (\ref{def:semantic-security}) is upper bounded by $\lemmaconst > 0$, we have
\begin{align}
\label{example:semantic-security-general-def}
\lemmaconst
&>
\Probability_{\messageRV, \channelOutWiretapper^\codebookBlocklength}(\messageRV \in \wiretapperDecoder(\channelOutWiretapper^\codebookBlocklength))
-
\max\limits_{\wiretapperGuesser \in \partition}
\Probability_{\messageRV}(\messageRV \in \wiretapperGuesser)
\\
\label{example:semantic-security-general-prob-substitution}
&=
\sum\limits_{\indexPartitions=1}^\nPartitionElements
\frac{1}{\nPartitionElements}
\Probability_{\channelOutWiretapper^\codebookBlocklength}(\wiretapperDecoder(\channelOutWiretapper^\codebookBlocklength) = \messageAlphabet_\indexPartitions | \messageRV = \messageAlphabetElement_\indexPartitions)
-
\frac{1}{\nPartitionElements}
\\
\label{example:semantic-security-general-errorprob-def}
&=
\sum\limits_{\indexPartitions=1}^\nPartitionElements
\frac{1}{\nPartitionElements}
(1-\errorprob_\indexPartitions^{\messageAlphabetElement_\indexPartitions})
-
\frac{1}{\nPartitionElements}
\end{align}
where (\ref{example:semantic-security-general-def}) is by assumption valid for all $\Probability_\messageRV$, (\ref{example:semantic-security-general-prob-substitution}) follows by plugging in $\Probability_\messageRV$ and (\ref{example:semantic-security-general-errorprob-def}) follows by plugging in the definitions of the $\errorprob_\indexPartitions^{\messageAlphabetElement_\indexPartitions}$. We can express this equivalently as
\[
\sum\limits_{\indexPartitions=1}^\nPartitionElements
\errorprob_\indexPartitions^{\messageAlphabetElement_\indexPartitions}
>
\nPartitionElements - 1 - \nPartitionElements\lemmaconst,
\]
which shows that also in this more general case, as (\ref{def:semantic-security}) tends to $0$, the wiretapper's error probabilities approach the values that can be achieved by stochastic guessing without observing any channel output at all, even if these error probabilities are only assumed to hold for one message in every partition element.
\end{example}

\begin{lemma}
\label{lemma:distinguishing-implies-strong}
If distinguishing security holds, i.e. (\ref{def:distinguishing-security}) approaches $0$, then so does $\max_{\codebookpmf_\messageRV} \mutualInformation{\messageRV}{\channelOutWiretapper^\codebookBlocklength}$. In particular, strong secrecy holds.
\end{lemma}
\begin{proof}
Fix an arbitrary distribution $\codebookpmf_\messageRV$ on $\messageAlphabet$ and observe that
\begin{align}
\mutualInformation{\messageRV}{\channelOutWiretapper^\codebookBlocklength}
&\leq
\entropy{\channelOutWiretapper^\codebookBlocklength} - \entropyConditional{\channelOutWiretapper^\codebookBlocklength}{\messageRV = \tilde{\messageAlphabetElement}}
\\
\label{eq:ds-implies-strong-mutinf}
&\leq
-\frac{1}{2}
\lemmaconst
\log \frac{\lemmaconst}
          {2\absolute{\channelOutAlphWiretapper}^\codebookBlocklength},
\end{align}

where $\tilde{\messageAlphabetElement}$ is chosen as the element in $\messageAlphabet$ that minimizes $\entropyConditional{\channelOutWiretapper^\codebookBlocklength}{\messageRV = \tilde{\messageAlphabetElement}}$ and the last inequality follows from Lemma~\ref{totvar-entropy-lemma}, as long as $\lemmaconst < 1/4$, where
\begin{align*}
\lemmaconst
&
=
\totalvariationBigg{
  \sum\limits_{(\messageAlphabetElement_1, \messageAlphabetElement_2) \in \messageAlphabet}
    \codebookpmf_\messageRV(\messageAlphabetElement_1, \messageAlphabetElement_2)
    \codebookpmf_{\channelOutWiretapper^\codebookBlocklength | \codebookOne(\messageAlphabetElement_1,\cdot), \codebookTwo(\messageAlphabetElement_2, \cdot)}
  \\
  &~~~~~~~-
  \codebookpmf_{\channelOutWiretapper^\codebookBlocklength | \codebookOne(\tilde{\messageAlphabetElement}_1,\cdot), \codebookTwo(\tilde{\messageAlphabetElement}_2, \cdot)}
}
\\
&
\begin{aligned}
\leq
&\sum\limits_{(\messageAlphabetElement_1, \messageAlphabetElement_2) \in \messageAlphabet}
\codebookpmf_\messageRV(\messageAlphabetElement_1, \messageAlphabetElement_2)
\\
&
\totalvariationlr{
  \codebookpmf_{\channelOutWiretapper^\codebookBlocklength | \codebookOne(\messageAlphabetElement_1,\cdot), \codebookTwo(\messageAlphabetElement_2, \cdot)}
  -
  \codebookpmf_{\channelOutWiretapper^\codebookBlocklength | \codebookOne(\tilde{\messageAlphabetElement}_1,\cdot), \codebookTwo(\tilde{\messageAlphabetElement}_2, \cdot)}
}.
\end{aligned}
\end{align*}

Note that the latter variational distance term is clearly bounded by (\ref{def:distinguishing-security}), and so (\ref{eq:ds-implies-strong-mutinf}) approaches $0$ if (\ref{def:distinguishing-security}) does.
\end{proof}

\begin{lemma}
\label{lemma:distinguishing-implies-semantic}
(\ref{def:semantic-security}) is upper bounded by (\ref{def:distinguishing-security}). In particular, distinguishing security implies semantic security.
\end{lemma}

\begin{proof}
Fix a partition $\partition$ of $\messageAlphabet$ and a probability mass function $\codebookpmf_\messageRV$ on $\messageAlphabet$ that realize the outer maximum in (\ref{def:semantic-security}) as well as a decoding function $\wiretapperDecoder: \channelOutWiretapper^\codebookBlocklength \rightarrow \partition$ and a guess $\wiretapperGuesser \in \partition$ that realize the inner maxima. We can then rewrite (\ref{def:semantic-security}) as $\semanticTotvarTerm + \semanticErrorTermOne - \semanticErrorTermTwo$ by steps as shown in (\ref{distinguishing-implies-semantic-begin}) through (\ref{distinguishing-implies-semantic-end}) in the appendix.

Observe that $\semanticErrorTermOne = \Probability_\messageRV(\messageRV \notin \wiretapperGuesser)$ and that, taking into account that $\wiretapperGuesser$ minimizes $\Probability_\messageRV(\messageRV \notin \wiretapperGuesser)$ over the elements of $\partition$,
\begin{align*}
&
\begin{aligned}
\semanticErrorTermTwo
=
&\sum\limits_{(\tilde{\messageAlphabetElement}_1, \tilde{\messageAlphabetElement}_2) \in \messageAlphabet}
  \codebookpmf_\messageRV(\tilde{\messageAlphabetElement}_1, \tilde{\messageAlphabetElement}_2)
  \sum\limits_{\channelOutAlphElementWiretapper^\codebookBlocklength \in \channelOutAlphWiretapper^\codebookBlocklength}
  \\
    &~~~~\codebookpmf_{\channelOutWiretapper^\codebookBlocklength | \codebookOne(\tilde{\messageAlphabetElement}_1, \cdot), \codebookTwo(\tilde{\messageAlphabetElement}_2, \cdot)}(\channelOutAlphElementWiretapper^\codebookBlocklength)
    \Probability_\messageRV(\messageRV \notin \wiretapperDecoder(\channelOutWiretapper^\codebookBlocklength))
\end{aligned}
\\
&
\begin{aligned}
\hphantom{\semanticErrorTermTwo}
\geq
&\sum\limits_{(\tilde{\messageAlphabetElement}_1, \tilde{\messageAlphabetElement}_2) \in \messageAlphabet}
  \codebookpmf_\messageRV(\tilde{\messageAlphabetElement}_1, \tilde{\messageAlphabetElement}_2)
  \sum\limits_{\channelOutAlphElementWiretapper^\codebookBlocklength \in \channelOutAlphWiretapper^\codebookBlocklength}
  \\
    &~~~~\codebookpmf_{\channelOutWiretapper^\codebookBlocklength | \codebookOne(\tilde{\messageAlphabetElement}_1, \cdot), \codebookTwo(\tilde{\messageAlphabetElement}_2, \cdot)}(\channelOutAlphElementWiretapper^\codebookBlocklength)
    \Probability_\messageRV(\messageRV \notin \wiretapperGuesser))
\end{aligned}
\\
&
\begin{aligned}
\hphantom{\semanticErrorTermTwo}
=
\Probability_\messageRV(\messageRV \notin \wiretapperGuesser).
\end{aligned}
\end{align*}
Observe further that, optimizing the outcome of the indicator function in $\semanticTotvarTerm$ for each term in the innermost sum individually, we get
\begin{align*}
\semanticTotvarTerm
\leq 
&\sum\limits_{(\messageAlphabetElement_1, \messageAlphabetElement_2), (\tilde{\messageAlphabetElement}_1, \tilde{\messageAlphabetElement}_2) \in \messageAlphabet}
\codebookpmf_\messageRV(\messageAlphabetElement_1, \messageAlphabetElement_2)
\codebookpmf_\messageRV(\tilde{\messageAlphabetElement}_1, \tilde{\messageAlphabetElement}_2)
\\
&\cdot \totalvariationlr{
  \codebookpmf_{\channelOutWiretapper^\codebookBlocklength | \codebookOne(\messageAlphabetElement_1, \cdot), \codebookTwo(\messageAlphabetElement_2, \cdot)}
  -
  \codebookpmf_{\channelOutWiretapper^\codebookBlocklength | \codebookOne(\tilde{\messageAlphabetElement}_1, \cdot), \codebookTwo(\tilde{\messageAlphabetElement}_2, \cdot)}
},
\end{align*}
which is clearly upper bounded by (\ref{def:distinguishing-security}). Summarizing, we have seen that $\semanticErrorTermOne-\semanticErrorTermTwo \leq 0$ and $\semanticTotvarTerm$ is upper bounded by (\ref{def:distinguishing-security}) and consequently (\ref{def:semantic-security}) is upper bounded by (\ref{def:distinguishing-security}).
\end{proof}

\subsection{Coding for the Multiple-Access Channel}
\label{sec:secrecy-mac-coding}
The capacity region of the multiple-access channel is well known and proofs can be found in almost any textbook on information theory, however, usually only the existence of a suitable codebook is proven. We need here the slightly stronger statement that the probability of randomly choosing a bad codebook is exponentially small, which requires only a slight modification of these proofs, which is, however, normally not given in the textbooks explicitly. For the sake of completeness, we state here the result that we need and present a proof that is similar to the one given in~\cite{TeInformationSpectrum}.
\begin{theorem}
\label{theorem:mac-coding}
Suppose
$\channel = (\channelInOneAlph, \channelInTwoAlph, \channelOutAlph, \channelpmf_{\channelOut | \channelInOne, \channelInTwo})$
is a channel, $\channelpmf_\channelInOne$ and $\channelpmf_\channelInTwo$ are input distributions, $\codebookRateOne < \mutualInformationConditional{\channelInOne}{\channelOut}{\channelInTwo}$, $\codebookRateTwo < \mutualInformationConditional{\channelInTwo}{\channelOut}{\channelInOne}$ and $\codebookRateOne+\codebookRateTwo < \mutualInformation{\channelInOne,\channelInTwo}{\channelOut}$.
Then there exist decoding functions
\[
\codebookDecoder_\codebookBlocklength: \channelOutAlph^\codebookBlocklength
                                       \rightarrow
                                       \{1, \dots, \exp(\codebookBlocklength\codebookRateOne)\} \times \{1, \dots, \exp(\codebookBlocklength\codebookRateTwo)\}
\]
and $\finalconstOne, \finalconstTwo > 0$ such that for large enough block length $\codebookBlocklength$, the codebook distributions of block length $\codebookBlocklength$ and rates $\codebookRateOne$ and $\codebookRateTwo$ satisfy
\begin{align}
\label{theorem:mac-coding-probability-statement}
\Probability_{\codebookOne, \codebookTwo} \left(
  \errorprob
  >
  \exp(-\finalconstOne\codebookBlocklength)
\right)
\leq
\exp\left(-\finalconstTwo\codebookBlocklength\right),
\end{align}
where we define the \emph{average decoding error probability} as
\begin{align*}
\errorprob
:=
\sum\limits_{\codewordIndex_1=1}^{\exp{\codebookBlocklength\codebookRateOne}}
\sum\limits_{\codewordIndex_2=1}^{\exp{\codebookBlocklength\codebookRateTwo}}
  &\exp(-\codebookBlocklength(\codebookRateOne+\codebookRateTwo))
  \\
  &
  \begin{aligned}
    \Probability_{\channelOut^\codebookBlocklength} \Big(
          &\codebookDecoder_\codebookBlocklength(\channelOut^\codebookBlocklength) \neq (\codewordIndex_1, \codewordIndex_2)
          ~\big|~
          \\
          &\channelInOne^\codebookBlocklength = \codebookOneWord{\codewordIndex_1},
          \channelInTwo^\codebookBlocklength = \codebookTwoWord{\codewordIndex_2}
        \Big).
  \end{aligned}
\end{align*}
\end{theorem}
\begin{proof}
We define typical sets
\begin{align*}
\typicalSetIndex{\typicalityParam}{\codebookBlocklength}{1}
&:=
\{
  (\channelInOneAlphElement^\codebookBlocklength, \channelInTwoAlphElement^\codebookBlocklength, \channelOutAlphElement^\codebookBlocklength)
  :
  \informationDensityConditional{\channelInOneAlphElement^\codebookBlocklength}{\channelOutAlphElement^\codebookBlocklength}{\channelInTwoAlphElement^\codebookBlocklength}
  \geq
  \codebookBlocklength(\mutualInformationConditional{\channelInOne}{\channelOut}{\channelInTwo} - \typicalityParam)
\}
\\
\typicalSetIndex{\typicalityParam}{\codebookBlocklength}{2}
&:=
\{
  (\channelInOneAlphElement^\codebookBlocklength, \channelInTwoAlphElement^\codebookBlocklength, \channelOutAlphElement^\codebookBlocklength)
  :
  \informationDensityConditional{\channelInTwoAlphElement^\codebookBlocklength}{\channelOutAlphElement^\codebookBlocklength}{\channelInOneAlphElement^\codebookBlocklength}
  \geq
  \codebookBlocklength(\mutualInformationConditional{\channelInTwo}{\channelOut}{\channelInOne} - \typicalityParam)
\}
\\
\typicalSetIndex{\typicalityParam}{\codebookBlocklength}{3}
&:=
\{
  (\channelInOneAlphElement^\codebookBlocklength, \channelInTwoAlphElement^\codebookBlocklength, \channelOutAlphElement^\codebookBlocklength)
  :
  \informationDensity{\channelInOneAlphElement^\codebookBlocklength, \channelInTwoAlphElement^\codebookBlocklength}{\channelOutAlphElement^\codebookBlocklength}
  \geq
  \codebookBlocklength(\mutualInformation{\channelInOne, \channelInTwo}{\channelOut} - \typicalityParam)
\}
\\
\typicalSet{\typicalityParam}{\codebookBlocklength}
&:=
\typicalSetIndex{\typicalityParam}{\codebookBlocklength}{1}
\cap
\typicalSetIndex{\typicalityParam}{\codebookBlocklength}{2}
\cap
\typicalSetIndex{\typicalityParam}{\codebookBlocklength}{3}
\end{align*}
and a corresponding joint typicality decoder
\[
\codebookDecoder_\codebookBlocklength(\channelOutAlphElement^\codebookBlocklength) :=
\begin{cases}
  (\codewordIndex_1, \codewordIndex_2), &(\codebookOneWord{\codewordIndex_1}, \codebookTwoWord{\codewordIndex_2}, \channelOutAlphElement^\codebookBlocklength)
                                         \in \typicalSet{\typicalityParam}{\codebookBlocklength}
\\
                                        &\text{for unique $(\codewordIndex_1, \codewordIndex_2)$}
\\
  (1,1), &\text{otherwise.}
\end{cases}
\]
So in order to correctly decode a received message, we need the received sequence to be jointly typical with the encoded messages, while it must not be jointly typical with any other pair of codewords. We formalize these notions by defining the probabilities of three error events corresponding to the former being false as in (\ref{theorem:mac-coding-proof-errorprob-typ}) for $\indexForTypicalSet \in \{1,2,3\}$ and the probabilities of error events corresponding to the latter being false as in (\ref{theorem:mac-coding-proof-errorprob-atyp-1}), (\ref{theorem:mac-coding-proof-errorprob-atyp-2}) and (\ref{theorem:mac-coding-proof-errorprob-atyp-3}).
\begin{figure*}
\normalsize
\begin{align}
\label{theorem:mac-coding-proof-errorprob-typ}
\errorprob_{\mathrm{typ},\indexForTypicalSet}^{\codewordIndex_1, \codewordIndex_2}
&:=
    \Probability_{\channelOut^\codebookBlocklength} \Big(
      (\channelInOne^\codebookBlocklength, \channelInTwo^\codebookBlocklength, \channelOut^\codebookBlocklength)
      \notin
      \typicalSetIndex{\typicalityParam}{\codebookBlocklength}{\indexForTypicalSet}
      ~\big|~
      \channelInOne^\codebookBlocklength = \codebookOneWord{\codewordIndex_1}, \channelInTwo^\codebookBlocklength = \codebookTwoWord{\codewordIndex_2}
    \Big)
\\
\label{theorem:mac-coding-proof-errorprob-atyp-1}
\errorprob_{\mathrm{atyp},1}^{\codewordIndex_1, \codewordIndex_2}
&:=
    \Probability_{\channelOut^\codebookBlocklength} \Big(
      \exists \tilde{\codewordIndex}_1 \neq \codewordIndex_1 :
      (\codebookOneWord{\tilde{\codewordIndex}_1}, \channelInTwo^\codebookBlocklength, \channelOut^\codebookBlocklength)
      \in
      \typicalSetIndex{\typicalityParam}{\codebookBlocklength}{1}
      ~\big|~
      \channelInOne^\codebookBlocklength = \codebookOneWord{\codewordIndex_1}, \channelInTwo^\codebookBlocklength = \codebookTwoWord{\codewordIndex_2}
    \Big)
\\
\label{theorem:mac-coding-proof-errorprob-atyp-2}
\errorprob_{\mathrm{atyp},2}^{\codewordIndex_1, \codewordIndex_2}
&:=
    \Probability_{\channelOut^\codebookBlocklength} \Big(
      \exists \tilde{\codewordIndex}_2 \neq \codewordIndex_2 :
      (\channelInOne^\codebookBlocklength, \codebookOneWord{\tilde{\codewordIndex}_2}, \channelOut^\codebookBlocklength)
      \in
      \typicalSetIndex{\typicalityParam}{\codebookBlocklength}{2}
      ~\big|~
      \channelInOne^\codebookBlocklength = \codebookOneWord{\codewordIndex_1}, \channelInTwo^\codebookBlocklength = \codebookTwoWord{\codewordIndex_2}
    \Big)
\\
\label{theorem:mac-coding-proof-errorprob-atyp-3}
\errorprob_{\mathrm{atyp},3}^{\codewordIndex_1, \codewordIndex_2}
&:=
    \Probability_{\channelOut^\codebookBlocklength} \Big(
      \exists \tilde{\codewordIndex}_1 \neq \codewordIndex_1, \tilde{\codewordIndex}_2 \neq \codewordIndex_2 :
      (\codebookOneWord{\tilde{\codewordIndex}_1}, \codebookOneWord{\tilde{\codewordIndex}_2}, \channelOut^\codebookBlocklength)
      \in
      \typicalSetIndex{\typicalityParam}{\codebookBlocklength}{3}
      ~\big|~
      \channelInOne^\codebookBlocklength = \codebookOneWord{\codewordIndex_1}, \channelInTwo^\codebookBlocklength = \codebookTwoWord{\codewordIndex_2}
    \Big)
\end{align}
\hrulefill
\end{figure*}
We observe by the union bound and the definition of $\codebookDecoder_\codebookBlocklength$ that
\begin{align*}
\errorprob
\leq
\sum\limits_{\codewordIndex_1=1}^{\exp{\codebookBlocklength\codebookRateOne}}
\sum\limits_{\codewordIndex_2=1}^{\exp{\codebookBlocklength\codebookRateTwo}}
  &\exp(-\codebookBlocklength(\codebookRateOne+\codebookRateTwo))
  \\
  &\begin{aligned}
    \Big(
      &\errorprob_{\mathrm{typ},1}^{\codewordIndex_1, \codewordIndex_2}
      +
      \errorprob_{\mathrm{typ},2}^{\codewordIndex_1, \codewordIndex_2}
      +
      \errorprob_{\mathrm{typ},3}^{\codewordIndex_1, \codewordIndex_2}
      +
      \\
      &\errorprob_{\mathrm{atyp},1}^{\codewordIndex_1, \codewordIndex_2}
      +
      \errorprob_{\mathrm{atyp},2}^{\codewordIndex_1, \codewordIndex_2}
      +
      \errorprob_{\mathrm{atyp},3}^{\codewordIndex_1, \codewordIndex_2}
    \Big)
  \end{aligned}
\end{align*}
and thus, taking into cosideration that since the codewords are identically distributed, $\errorprob_{\text{typ},\indexForTypicalSet}^{\codewordIndex_1, \codewordIndex_2}$ and $\errorprob_{\text{atyp},\indexForTypicalSet}^{\codewordIndex_1, \codewordIndex_2}$ have expectations that do not depend on $\codewordIndex_1$ or $\codewordIndex_2$, we get
\begin{alignat*}{4}
&\Expectation_{\codebookOne, \codebookTwo} \errorprob
\leq
&~&
\Expectation_{\codebookOne, \codebookTwo} \errorprob_{\mathrm{typ},1}^{1,1}
&+&
\Expectation_{\codebookOne, \codebookTwo} \errorprob_{\mathrm{typ},2}^{1,1}
&+&
\Expectation_{\codebookOne, \codebookTwo} \errorprob_{\mathrm{typ},2}^{1,1}
\\
&~
&+&
\Expectation_{\codebookOne, \codebookTwo} \errorprob_{\mathrm{atyp},1}^{1,1}
&+&
\Expectation_{\codebookOne, \codebookTwo} \errorprob_{\mathrm{atyp},2}^{1,1}
&+&
\Expectation_{\codebookOne, \codebookTwo} \errorprob_{\mathrm{atyp},3}^{1,1}
\end{alignat*}

We next bound each of these summands above by $\exp(-\codebookBlocklength\proofconstantOne)$ for sufficiently large $\codebookBlocklength$ and sufficiently small but positive $\proofconstantOne$.

To this end, we observe
\begin{align}
&\begin{aligned}
\hphantom{={}}
\Expectation_{\codebookOne, \codebookTwo} \errorprob_{\mathrm{typ},1}^{1,1}
\end{aligned}
\\
&\begin{aligned}
=
\Probability_{\channelInOne^\codebookBlocklength, \channelInTwo^\codebookBlocklength, \channelOut^\codebookBlocklength} \Big(
  (\channelInOne^\codebookBlocklength, \channelInTwo^\codebookBlocklength, \channelOut^\codebookBlocklength)
  \notin
  \typicalSetIndex{\typicalityParam}{\codebookBlocklength}{1}
\Big)
\end{aligned}
\\
\label{proof:mac-coding-typical-error-chernoffprep}
&\begin{aligned}
=
\Probability_{\channelInOne^\codebookBlocklength, \channelInTwo^\codebookBlocklength, \channelOut^\codebookBlocklength} \Bigg(
  &\left(
    \frac{\channelpmf_{\channelOut^\codebookBlocklength | \channelInOne^\codebookBlocklength, \channelInTwo^\codebookBlocklength}
          (\channelOut^\codebookBlocklength | \channelInOne^\codebookBlocklength, \channelInTwo^\codebookBlocklength)}
         {\channelpmf_{\channelOut^\codebookBlocklength | \channelInTwo^\codebookBlocklength}
          (\channelOut^\codebookBlocklength | \channelInTwo^\codebookBlocklength)}
  \right)^{\renyiParam-1}
  \\
  &>
  \exp(\codebookBlocklength(\renyiParam-1)(\mutualInformationConditional{\channelInOne}{\channelOut}{\channelInTwo}-\typicalityParam))
\Bigg)
\end{aligned}
\\
\label{proof:mac-coding-typical-error-markov}
&\begin{aligned}
\leq
&\Expectation_{\channelInOne^\codebookBlocklength, \channelInTwo^\codebookBlocklength, \channelOut^\codebookBlocklength} \left(
  \left(
    \frac{\channelpmf_{\channelOut^\codebookBlocklength | \channelInOne^\codebookBlocklength, \channelInTwo^\codebookBlocklength}
          (\channelOut^\codebookBlocklength | \channelInOne^\codebookBlocklength, \channelInTwo^\codebookBlocklength)}
         {\channelpmf_{\channelOut^\codebookBlocklength | \channelInTwo^\codebookBlocklength}
          (\channelOut^\codebookBlocklength | \channelInTwo^\codebookBlocklength)}
  \right)^{\renyiParam-1}
\right)
\\
&\cdot \exp(-\codebookBlocklength(\renyiParam-1)(\mutualInformationConditional{\channelInOne}{\channelOut}{\channelInTwo}-\typicalityParam))
\end{aligned}
\\
\label{proof:mac-coding-typical-error-renyidiv}
&\begin{aligned}
=
\exp\Big(
  &\codebookBlocklength
  (\renyiParam-1)
  \\
  \cdot
  &\left(
    \renyidiv{\renyiParam}{\Probability_{\channelInOne, \channelInTwo, \channelOut}}
                          {\Probability_{\channelInOne | \channelInTwo}\Probability_{\channelOut | \channelInTwo}\Probability_{\channelInTwo}}
    -
    \mutualInformationConditional{\channelInOne}{\channelOut}{\channelInTwo}
    +
    \typicalityParam
  \right)
\Big)
\end{aligned}
\\
\label{proof:mac-coding-typical-error-final}
&\begin{aligned}
\leq
\exp(-\codebookBlocklength\proofconstantOne)
\end{aligned}
\end{align}
where (\ref{proof:mac-coding-typical-error-chernoffprep}) follows for all $\renyiParam < 1$ by substituting the definition of $\typicalSetIndex{\typicalityParam}{\codebookBlocklength}{1}$ and reformulating, (\ref{proof:mac-coding-typical-error-markov}) follows by Markov's inequality, (\ref{proof:mac-coding-typical-error-renyidiv}) follows by reformulating and substituting the definition of Rényi divergence, and (\ref{proof:mac-coding-typical-error-final}) follows for sufficiently small but positive $\proofconstantOne$ by taking $\renyiParam < 1$ close enough to $1$ such that the exponent in (\ref{proof:mac-coding-typical-error-renyidiv}) is negative. This is possible by observing that the Rényi divergence approaches the Kullback-Leibler divergence for $\renyiParam \rightarrow 1$ and is finite and thus continuous since our alphabets are finite~\cite{RenyiDiv}.

Moreover, we have
\begin{align}
&\begin{aligned}
\hphantom{={}}
\Expectation_{\codebookOne, \codebookTwo} \errorprob_{\mathrm{atyp},1}^{1,1}
\end{aligned}
\\
\label{proof:mac-coding-atypical-without-codebook}
&\begin{aligned}
\leq
\exp(\codebookBlocklength\codebookRateOne)
\Probability_{\channelInTwo^\codebookBlocklength, \channelOut^\codebookBlocklength}\Probability_{\channelInOne^\codebookBlocklength} \left(
  (\channelInOne^\codebookBlocklength, \channelInTwo^\codebookBlocklength, \channelOut^\codebookBlocklength)
  \in
  \typicalSetIndex{\typicalityParam}{\codebookBlocklength}{1}
\right)
\end{aligned}
\\
&\begin{aligned}
=
\exp(\codebookBlocklength\codebookRateOne)
&\sum_{\channelInOneAlphElement^\codebookBlocklength \in \channelInOneAlph^\codebookBlocklength}
\sum_{\channelInTwoAlphElement^\codebookBlocklength \in \channelInTwoAlph^\codebookBlocklength}
\sum_{\channelOutAlphElement^\codebookBlocklength \in \channelOutAlph^\codebookBlocklength}
\channelpmf_{\channelInOne^\codebookBlocklength}(\channelInOneAlphElement^\codebookBlocklength)
\channelpmf_{\channelInTwo^\codebookBlocklength, \channelOut^\codebookBlocklength}(\channelInTwoAlphElement^\codebookBlocklength, \channelOutAlphElement^\codebookBlocklength)
\\
&\cdot \indicator{(\channelInOneAlphElement^\codebookBlocklength,\channelInTwoAlphElement^\codebookBlocklength,\channelOutAlphElement^\codebookBlocklength)
           \in
           \typicalSetIndex{\typicalityParam}{\codebookBlocklength}{1}
          }
\end{aligned}
\\
\label{proof:mac-coding-atypical-reformulation}
&\begin{aligned}
=
\exp(\codebookBlocklength\codebookRateOne)
&\sum_{\channelInOneAlphElement^\codebookBlocklength \in \channelInOneAlph^\codebookBlocklength}
\sum_{\channelInTwoAlphElement^\codebookBlocklength \in \channelInTwoAlph^\codebookBlocklength}
\sum_{\channelOutAlphElement^\codebookBlocklength \in \channelOutAlph^\codebookBlocklength}
\channelpmf_{\channelInOne^\codebookBlocklength, \channelInTwo^\codebookBlocklength, \channelOut^\codebookBlocklength}(\channelInOneAlphElement^\codebookBlocklength, \channelInTwoAlphElement^\codebookBlocklength, \channelOutAlphElement^\codebookBlocklength)
\\
&\cdot \exp(-\informationDensityConditional{\channelInOneAlphElement^\codebookBlocklength}{\channelOutAlphElement^\codebookBlocklength}{\channelInTwoAlphElement^\codebookBlocklength})
\indicator{(\channelInOneAlphElement^\codebookBlocklength,\channelInTwoAlphElement^\codebookBlocklength,\channelOutAlphElement^\codebookBlocklength)
           \in
           \typicalSetIndex{\typicalityParam}{\codebookBlocklength}{1}
          }
\end{aligned}
\\
\label{proof:mac-coding-atypical-mutual-information}
&\begin{aligned}
\leq
\exp(\codebookBlocklength(\codebookRateOne - \mutualInformationConditional{\channelInOne}{\channelOut}{\channelInTwo} + \typicalityParam)),
\end{aligned}
\\
\label{proof:mac-coding-atypical-final}
&\begin{aligned}
\leq
\exp(-\codebookBlocklength\proofconstantOne),
\end{aligned}
\end{align}
where (\ref{proof:mac-coding-atypical-without-codebook}) is an application of the union bound, (\ref{proof:mac-coding-atypical-reformulation}) follows by reformulating and substituting the definition of information density, (\ref{proof:mac-coding-atypical-mutual-information}) follows by the definition of $\typicalSetIndex{\typicalityParam}{\codebookBlocklength}{1}$ and (\ref{proof:mac-coding-atypical-final}) follows for all sufficiently small $\proofconstantOne > 0$ as long as $\typicalityParam$ is so small that $\codebookRateOne - \mutualInformationConditional{\channelInOne}{\channelOut}{\channelInTwo} + \typicalityParam < 0$.

We skip the calculations showing that
\[\Expectation_{\codebookOne, \codebookTwo} \errorprob_{\mathrm{typ},\indexForTypicalSet}^{1,1},~ \Expectation_{\codebookOne, \codebookTwo} \errorprob_{\mathrm{atyp},\indexForTypicalSet}^{1,1} \leq \exp(-\codebookBlocklength\proofconstantOne)\]
for $\indexForTypicalSet \in \{2,3\}$, sufficiently large $\codebookBlocklength$ and sufficiently small $\proofconstantOne > 0$, since these statements follow in a straightforward manner by calculations very similar to the ones shown above.

Putting everything together and applying Markov's inequality, we get that for sufficiently large $\codebookBlocklength$, sufficiently small $\proofconstantOne > 0$, $\finalconstOne < \proofconstantOne$, and $\finalconstTwo < \proofconstantOne - \finalconstOne$,
\begin{align*}
\Probability_{\codebookOne, \codebookTwo} \left(
  \errorprob
  >
  \exp(-\finalconstOne\codebookBlocklength)
\right)
&\leq
\Expectation_{\codebookOne, \codebookTwo} (\errorprob)
\exp(\finalconstOne\codebookBlocklength)
\\
&\leq
6 \exp(-\codebookBlocklength(\proofconstantOne - \finalconstOne))
\\
&\leq
\exp(-\finalconstTwo\codebookBlocklength),
\end{align*}
concluding the proof of the theorem.
\end{proof}

\subsection{Achievable Secrecy Rates under Distinguishing Security}
\label{sec:secrecy-achieve}
\begin{theorem}
\label{theorem:distinguishing-security}
Suppose $\channel = (\channelLegit, \channelWiretapper)$ is a wiretap channel, $\timeSharingRV$ is a random variable, $\channelpmf_{\channelInOne|\timeSharingRV}$ and $\channelpmf_{\channelInTwo|\timeSharingRV}$ are input distributions and $\codebookRateOne$ and $\codebookRateTwo$ are rates satisfying
\begin{align}
\label{theorem:distinguishing-security-rate-one}
\mutualInformationConditional{\channelInOne}{\channelOutWiretapper}{\timeSharingRV}
&<
\mutualInformationConditional{\channelInOne}{\channelOutLegit}{\channelInTwo,\timeSharingRV}
-
\codebookRateOne
\\
\label{theorem:distinguishing-security-rate-two}
\mutualInformationConditional{\channelInTwo}{\channelOutWiretapper}{\timeSharingRV}
&<
\mutualInformationConditional{\channelInTwo}{\channelOutLegit}{\channelInOne,\timeSharingRV}
-
\codebookRateTwo
\\
\label{theorem:distinguishing-security-rate-both}
\mutualInformationConditional{\channelInOne,\channelInTwo}{\channelOutWiretapper}{\timeSharingRV}
&<
\mutualInformationConditional{\channelInOne,\channelInTwo}{\channelOutLegit}{\timeSharingRV}
-
(\codebookRateOne+\codebookRateTwo).
\end{align}
Then there exist randomness rates $\codebookRandRateOne, \codebookRandRateTwo$, decoding functions
\begin{align*}
\codebookDecoder_\codebookBlocklength: \channelOutAlphLegit^\codebookBlocklength \rightarrow
  &\{1,\dots,\exp(\codebookBlocklength\codebookRateOne)\} \times \{1, \dots, \exp(\codebookBlocklength\codebookRandRateOne)\} \times \\
  &\{1,\dots,\exp(\codebookBlocklength\codebookRateTwo)\} \times \{1, \dots, \exp(\codebookBlocklength\codebookRandRateTwo)\}
\end{align*}
and $\finalconstOne, \finalconstTwo > 0$, such that for sufficiently large $\codebookBlocklength$, the probability $\Probability_\mathrm{bad}$ under the wiretap codebook distribution of drawing a bad pair of codebooks is less than $\exp(-\codebookBlocklength\finalconstTwo)$. A pair of codebooks is considered to be bad if it does not satisfy all of the following:
\begin{enumerate}
 \item \emph{(Distinguishing Security)}. For any $\codewordIndex_1, \codewordIndex_2$, we have
\[
\totalvariation{
  \codebookpmf_{\channelOutWiretapper^\codebookBlocklength | \codebookOne(\codewordIndex_1,\cdot), \codebookTwo(\codewordIndex_2, \cdot)} - \channelpmf_{\channelOut^\codebookBlocklength}
}
\leq
\exp(-\finalconstOne\codebookBlocklength),
\]
i.e. the wiretap codebooks achieve distinguishing security and the variational distance vanishes exponentially.
\item \emph{(Average decoding error)}. We have $\errorprob \leq \exp(-\codebookBlocklength\finalconstOne)$, where
\begin{align*}
\errorprob
:=
&\sum\limits_{\codewordIndex_1=1}^{\exp{\codebookBlocklength\codebookRateOne}}
\sum\limits_{\codewordIndex_2=1}^{\exp{\codebookBlocklength\codebookRateTwo}}
\sum\limits_{\randomnessIndex_1=1}^{\exp{\codebookBlocklength\codebookRandRateOne}}
\sum\limits_{\randomnessIndex_2=1}^{\exp{\codebookBlocklength\codebookRandRateTwo}}
\\
  &\exp(-\codebookBlocklength(\codebookRateOne+\codebookRateTwo+\codebookRandRateOne+\codebookRandRateTwo))
  \\
  &
  \begin{aligned}
    \Probability_{\channelOut^\codebookBlocklength} \Big(
          &\codebookDecoder_\codebookBlocklength(\channelOut^\codebookBlocklength) \neq (\codewordIndex_1, \randomnessIndex_1, \codewordIndex_2, \randomnessIndex_2)
          ~\big|~
          \\
          &\channelInOne^\codebookBlocklength = \codebookOneWord{\codewordIndex_1, \randomnessIndex_1},
          \channelInTwo^\codebookBlocklength = \codebookTwoWord{\codewordIndex_2, \randomnessIndex_2}
        \Big).
  \end{aligned}
\end{align*}
is the average decoding error; i.e. the legitimate receiver is able to (on average) reconstruct both the messages and the randomness used at the transmitters with exponentially vanishing error probability.
\end{enumerate}
\end{theorem}
\begin{proof}
Choose $\codebookRandRateOne$ such that
\begin{align}
\label{theorem:distinguishing-security-randrate-choice-one}
\max(\codebookRandRateOneLower{1}, \codebookRandRateOneLower{2})
< \codebookRandRateOne <
\min(\codebookRandRateOneUpper{1}, \codebookRandRateOneUpper{2}),
\end{align}
where
\begin{align*}
\codebookRandRateOneLower{1} &= \mutualInformationConditional{\channelInOne}{\channelOutWiretapper}{\timeSharingRV}
\\
\codebookRandRateOneLower{2}
&=
-\mutualInformationConditional{\channelInTwo}{\channelOutLegit}{\channelInOne,\timeSharingRV} + \codebookRateTwo + \mutualInformationConditional{\channelInOne,\channelInTwo}{\channelOutWiretapper}{\timeSharingRV}
\\
\codebookRandRateOneUpper{1} &= \mutualInformationConditional{\channelInOne}{\channelOutLegit}{\channelInTwo,\timeSharingRV} - \codebookRateOne
\\
\codebookRandRateOneUpper{2}
&=
-\mutualInformationConditional{\channelInTwo}{\channelOutWiretapper}{\timeSharingRV} + \mutualInformationConditional{\channelInOne,\channelInTwo}{\channelOutLegit}{\timeSharingRV} - (\codebookRateOne + \codebookRateTwo).
\end{align*}
To see that this is possible, we note that $\codebookRandRateOneLower{1} < \codebookRandRateOneUpper{1}$ is equivalent to (\ref{theorem:distinguishing-security-rate-one}), $\codebookRandRateOneLower{1} < \codebookRandRateOneUpper{2}$ and $\codebookRandRateOneLower{2} < \codebookRandRateOneUpper{1}$ both follow from (\ref{theorem:distinguishing-security-rate-both}), and $\codebookRandRateOneLower{2} < \codebookRandRateOneUpper{2}$ follows from (\ref{theorem:distinguishing-security-rate-two}) and (\ref{theorem:distinguishing-security-rate-both}).

Then, choose $\codebookRandRateTwo$ such that
\begin{align}
\label{theorem:distinguishing-security-randrate-choice-two}
\max(\codebookRandRateTwoLower{1}, \codebookRandRateTwoLower{2})
< \codebookRandRateTwo <
\min(\codebookRandRateTwoUpper{1}, \codebookRandRateTwoUpper{2}),
\end{align}
where
\begin{align*}
\codebookRandRateTwoLower{1} &= \mutualInformationConditional{\channelInOne,\channelInTwo}{\channelOutWiretapper}{\timeSharingRV} - \codebookRandRateOne
\\
\codebookRandRateTwoLower{2}
&=
\mutualInformationConditional{\channelInTwo}{\channelOutWiretapper}{\timeSharingRV}
\\
\codebookRandRateTwoUpper{1} &= \mutualInformationConditional{\channelInTwo}{\channelOutLegit}{\channelInOne,\timeSharingRV} - \codebookRateTwo
\\
\codebookRandRateTwoUpper{2}
&=
\mutualInformationConditional{\channelInOne, \channelInTwo}{\channelOutLegit}{\timeSharingRV} - \codebookRandRateOne - (\codebookRateOne + \codebookRateTwo).
\end{align*}
To see that this is possible, we note that $\codebookRandRateTwoLower{1} < \codebookRandRateTwoUpper{1}$ is equivalent to $\codebookRandRateOne > \codebookRandRateOneLower{2}$, which we have by (\ref{theorem:distinguishing-security-randrate-choice-one}), $\codebookRandRateTwoLower{1} < \codebookRandRateTwoUpper{2}$ is equivalent to (\ref{theorem:distinguishing-security-rate-both}), $\codebookRandRateTwoLower{2} < \codebookRandRateTwoUpper{1}$ is equivalent to (\ref{theorem:distinguishing-security-rate-two}) and $\codebookRandRateTwoLower{2} < \codebookRandRateTwoUpper{2}$ is equivalent to $\codebookRandRateOne < \codebookRandRateOneUpper{2}$, which we also have by (\ref{theorem:distinguishing-security-randrate-choice-one}).

Summarizing, with these choices, $\codebookRandRateOne, \codebookRandRateTwo$ have the properties
\begin{alignat*}{3}
\mutualInformationConditional{\channelInOne}{\channelOutWiretapper}{\timeSharingRV}
&~<~&
\codebookRandRateOne
&~<~&
\mutualInformationConditional{\channelInOne}{\channelOutLegit}{\channelInTwo, \timeSharingRV} &- \codebookRateOne
\\
\mutualInformationConditional{\channelInTwo}{\channelOutWiretapper}{\timeSharingRV}
&~<~&
\codebookRandRateTwo
&~<~&
\mutualInformationConditional{\channelInTwo}{\channelOutLegit}{\channelInOne,\timeSharingRV} &- \codebookRateTwo.
\\
\mutualInformationConditional{\channelInOne, \channelInTwo}{\channelOutWiretapper}{\timeSharingRV}
&~<~&
\codebookRandRateOne + \codebookRandRateTwo
&~<~&
\mutualInformationConditional{\channelInOne, \channelInTwo}{\channelOutLegit}{\timeSharingRV} &- \codebookRateTwo.
\end{alignat*}
We can look at the thus defined pair of random wiretap codebooks in two ways. On the one hand, they can be understood as a pair of codebooks of rates $\codebookRateOne + \codebookRandRateOne$ and $\codebookRateTwo + \codebookRandRateTwo$. They fulfill all the requirements of Theorem~\ref{theorem:mac-coding} (noting that by time sharing, the theorem is also valid for the convex closure specified in the theorem statement), and so we get $\finalconstOne', \finalconstTwo' > 0$ such that
\[
\Probability_{\codebookOne, \codebookTwo} (\errorprob > \exp(\codebookBlocklength\finalconstOne')) \leq \exp(\codebookBlocklength\finalconstTwo').
\]
On the other hand, we can look at them as a collection $(\codebookOne(\codewordIndex_1, \cdot))_{\codewordIndex_1=1}^{\exp(\codebookBlocklength\codebookRateOne)}$ of codebooks of rate $\codebookRandRateOne$ and $(\codebookTwo(\codewordIndex_2, \cdot))_{\codewordIndex_2=1}^{\exp(\codebookBlocklength\codebookRateTwo)}$ of codebooks of rate $\codebookRandRateTwo$, respectively. Therefore, each pair $(\codewordIndex_1, \codewordIndex_2)$ corresponds to a pair of codebooks fulfilling the requirements of Corollary~\ref{theorem:soft-covering-two-transmitters-convex} (again noting the possibility of time sharing), and so we get $\finalconstOne'', \finalconstTwo'' > 0$ such that for all $\codewordIndex_1, \codewordIndex_2$
\begin{multline*}
\Probability_{\codebookOne, \codebookTwo} \left(
  \totalvariation{
    \codebookpmf_{\channelOut^\codebookBlocklength | \codebookOne(\codewordIndex_1, \cdot), \codebookTwo(\codewordIndex_2, \cdot)} - \channelpmf_{\channelOut^\codebookBlocklength}
  }
  >
  \exp(-\finalconstOne''\codebookBlocklength)
\right)
\\
\leq
\exp\left(-\exp\left(\finalconstTwo''\codebookBlocklength\right)\right).
\end{multline*}
Choosing $\finalconstOne := \min(\finalconstOne', \finalconstOne'')$, we can apply the union bound to get
\begin{align*}
&\hphantom{{}={}}
\Probability_\mathrm{bad}
\\
&\leq
\Probability_{\codebookOne, \codebookTwo} (\errorprob > \exp(\codebookBlocklength\finalconstOne'))
+
\\
&\hphantom{{}={}}
\begin{aligned}
\sum\limits_{\codewordIndex_1=1}^{\exp(\codebookBlocklength\codebookRateOne)}
\sum\limits_{\codewordIndex_2=1}^{\exp(\codebookBlocklength\codebookRateTwo)}
\Probability_{\codebookOne, \codebookTwo} \Big(
  &\totalvariation{
    \codebookpmf_{\channelOut^\codebookBlocklength | \codebookOne(\codewordIndex_1, \cdot), \codebookTwo(\codewordIndex_2, \cdot)} - \channelpmf_{\channelOut^\codebookBlocklength}
  }
  \\
  &>
  \exp(-\finalconstOne''\codebookBlocklength)
\Big)
\end{aligned}
\\
&\leq
\exp(\finalconstTwo'\codebookBlocklength)
+
\exp(\codebookBlocklength(\codebookRateOne + \codebookRateTwo) - \exp(\finalconstTwo''\codebookBlocklength))
\\
&\leq
\exp(\finalconstTwo\codebookBlocklength)
\end{align*}
for sufficiently large $\codebookBlocklength$, as long as we choose $\finalconstTwo < \finalconstTwo'$.
\end{proof}

\bibliographystyle{plain}
\bibliography{references}

\clearpage
\onecolumn
\appendix

%

\begin{align}
\label{proof:soft-covering-two-transmitters-totalprob1}
&
\begin{aligned}
  \tilde{\Probability}
  =
  \sum\limits_{\hat{\codebookTwo}}
    \Probability_{\codebookTwo}(\codebookTwo = \hat{\codebookTwo})
    \Probability_{\codebookOne, \codebookTwo}\left(
    \vphantom{\sum\limits_{\codewordIndex_1=1}^{\exp(\codebookBlocklength\codebookRateOne)}}
    \right.
      &\sum\limits_{\codewordIndex_2=1}^{\exp(\codebookBlocklength\codebookRateTwo)}
        \exp(-\codebookBlocklength\codebookRateTwo)
        \frac{\channelpmf_{\channelOut^\codebookBlocklength | \channelInTwo^\codebookBlocklength}(\channelOutAlphElement^\codebookBlocklength | \codebookTwoWord{\codewordIndex_2})}
            {\channelpmf_{\channelOut^\codebookBlocklength}(\channelOutAlphElement^\codebookBlocklength)}
        \indicator{(\codebookTwoWord{\codewordIndex_2}, \channelOutAlphElement^\codebookBlocklength) \in \typicalSetIndex{\typicalityParam}{\codebookBlocklength}{2}}
        \\
        &\left. \vphantom{\sum\limits_{\codewordIndex_1=1}^{\exp(\codebookBlocklength\codebookRateOne)}}
        \totvarTypicalOne{\codebookTwoWord{\codewordIndex_2}}{\channelOutAlphElement^\codebookBlocklength}
        >
        1 + 3\exp(-\codebookBlocklength\proofconstantOne)
        ~|~
        \codebookOne \in \codebookSet_{\channelOutAlphElement^\codebookBlocklength}, \codebookTwo = \hat{\codebookTwo}
  \right)
\end{aligned}
\\
&
\begin{aligned}
  \phantom{\tilde{\Probability}}
  \leq
  \label{proof:soft-covering-two-transmitters-lemmaapplication1}
  \sum\limits_{\hat{\codebookTwo}}
    \Probability_{\codebookTwo}(\codebookTwo = \hat{\codebookTwo})
    \Probability_{\codebookOne, \codebookTwo}\left(
      \sum\limits_{\codewordIndex_2=1}^{\exp(\codebookBlocklength\codebookRateTwo)}
        \exp(-\codebookBlocklength\codebookRateTwo)
        \frac{\channelpmf_{\channelOut^\codebookBlocklength | \channelInTwo^\codebookBlocklength}(\channelOutAlphElement^\codebookBlocklength | \codebookTwoWord{\codewordIndex_2})}
            {\channelpmf_{\channelOut^\codebookBlocklength}(\channelOutAlphElement^\codebookBlocklength)}
        \indicator{(\codebookTwoWord{\codewordIndex_2}, \channelOutAlphElement^\codebookBlocklength) \in \typicalSetIndex{\typicalityParam}{\codebookBlocklength}{2}}
        \right. \\
      \left. \vphantom{\sum\limits_{\codewordIndex_1=1}^{\exp(\codebookBlocklength\codebookRateOne)}}
      >
      \frac{1 + 3\exp(-\codebookBlocklength\proofconstantOne)}
          {1 + \exp(-\codebookBlocklength\proofconstantOne)}
      ~|~
      \codebookOne \in \codebookSet_{\channelOutAlphElement^\codebookBlocklength}, \codebookTwo = \hat{\codebookTwo}
    \right)
\end{aligned}
\\
&\phantom{\tilde{\Probability}}\leq
\label{proof:soft-covering-two-transmitters-totalprob2-and-bound}
\Probability_{\codebookTwo}\left(
  \sum\limits_{\codewordIndex_2=1}^{\exp(\codebookBlocklength\codebookRateTwo)}
    \exp(-\codebookBlocklength\codebookRateTwo)
    \frac{\channelpmf_{\channelOut^\codebookBlocklength | \channelInTwo^\codebookBlocklength}(\channelOutAlphElement^\codebookBlocklength | \codebookTwoWord{\codewordIndex_2})}
          {\channelpmf_{\channelOut^\codebookBlocklength}(\channelOutAlphElement^\codebookBlocklength)}
    \indicator{(\codebookTwoWord{\codewordIndex_2}, \channelOutAlphElement^\codebookBlocklength) \in \typicalSetIndex{\typicalityParam}{\codebookBlocklength}{2}}
  >
  1 + \exp(-\codebookBlocklength\proofconstantOne)
\right)
\\
&\phantom{\tilde{\Probability}}\leq
\label{proof:soft-covering-two-transmitters-lemmaapplication2}
\exp\left(
  -\frac{1}{3} \exp(-\codebookBlocklength (\mutualInformation{\channelInTwo}{\channelOut} + \typicalityParam + 2\proofconstantOne - \codebookRateTwo))
\right)
\end{align}
(\ref{proof:soft-covering-two-transmitters-totalprob1}) follows from the law of total probability,~(\ref{proof:soft-covering-two-transmitters-lemmaapplication1}) is a consequence of the condition $\codebookOne \in \codebookSet_{\channelOutAlphElement^\codebookBlocklength}$,~(\ref{proof:soft-covering-two-transmitters-totalprob2-and-bound}) results from an application of the law of total probability and the assumption that $\codebookBlocklength$ is sufficiently large such that $\exp(-\codebookBlocklength\proofconstantOne) \leq 1$. Finally,~(\ref{proof:soft-covering-two-transmitters-lemmaapplication2}) follows from Lemma~\ref{lemma:soft-covering-two-transmitters-typical} with $\generalrvOne=\channelInTwo$, $\generalrvThree=\channelOut$, $\generalrvTwo$ a deterministic random variable with only one possible realization and $\lemmaconst=\exp(-\codebookBlocklength\proofconstantOne)$.
\vspace{10pt}
\hrule

\begin{align}
\label{proof:soft-covering-two-transmitters-union-bound-start}
&\phantom{=}
\Probability_{\codebookOne, \codebookTwo} \left(
  \totalvariation{ \codebookpmf_{\channelOut^\codebookBlocklength | \codebookOne, \codebookTwo} - \channelpmf_{\channelOut^\codebookBlocklength}}
  >
  7\exp(-\codebookBlocklength\proofconstantOne)
\right)
\\
\label{proof:soft-covering-two-transmitters-union-bound-application}
&
\begin{aligned}
  \leq
  &\Probability_{\codebookOne, \codebookTwo}\left(
    \totvarAtypicalOne
    >
    2\exp(-\codebookBlocklength\proofconstantOne)
  \right)
  +
  \Probability_{\codebookOne, \codebookTwo}\left(
    \totvarAtypicalTwo
    >
    2\exp(-\codebookBlocklength\proofconstantOne)
  \right)
  \\
  &+
  \sum\limits_{\channelOutAlphElement^\codebookBlocklength \in \channelOutAlph^\codebookBlocklength} \left(
    \Probability_{\codebookOne}\left(
      \codebookOne \notin \codebookSet_{\channelOutAlphElement^\codebookBlocklength}
    \right)
    +
    \Probability_{\codebookOne, \codebookTwo}\left(
      \totvarTypical{\channelOutAlphElement^\codebookBlocklength}
      >
      1 + 3\exp(-\codebookBlocklength\proofconstantOne)
      ~|~
      \codebookOne \in \codebookSet_{\channelOutAlphElement^\codebookBlocklength}
    \right)
  \right)
\end{aligned}
\\
\label{proof:soft-covering-two-transmitters-union-bound-substitutions}
&
\begin{aligned}
  \leq
  &2\exp(-2\exp(\codebookBlocklength(\min(\codebookRateOne,\codebookRateTwo)-2\proofconstantOne)))
  +
  \cardinality{\channelOutAlph}^\codebookBlocklength
  \cardinality{\channelInTwoAlph}^\codebookBlocklength
  \exp\left(
    -\frac{1}{3} \exp(-\codebookBlocklength (\mutualInformationConditional{\channelInOne}{\channelOut}{\channelInTwo} + \typicalityParam + 2\proofconstantOne - \codebookRateOne))
  \right) \\
  &+
  \cardinality{\channelOutAlph}^\codebookBlocklength
  \exp\left(
    -\frac{1}{3} \exp(-\codebookBlocklength (\mutualInformation{\channelInTwo}{\channelOut} + \typicalityParam + 2\proofconstantOne - \codebookRateTwo))
  \right)
\end{aligned}
\end{align}

(\ref{proof:soft-covering-two-transmitters-union-bound-application}) follows from~(\ref{proof:soft-covering-two-transmitters-typical-split}) and the union bound. (\ref{proof:soft-covering-two-transmitters-union-bound-substitutions}) is a substitution of~(\ref{proof:soft-covering-two-transmitters-atypical-bound-1}), (\ref{proof:soft-covering-two-transmitters-atypical-bound-2}), (\ref{proof:soft-covering-two-transmitters-typical-bound}) and~(\ref{proof:soft-covering-two-transmitters-lemmaapplication2}).
\vspace{10pt}
\hrule

\begin{align}
\label{proof:soft-covering-two-transmitters-second-order-union-bound-start}
&\phantom{=}
\Probability_{\codebookOne, \codebookTwo} \left( \totalvariation{ \codebookpmf_{\channelOut^\codebookBlocklength | \codebookOne, \codebookTwo} - \channelpmf_{\channelOut^\codebookBlocklength}}
>
(\secondOrderAtypicalProbability{2} + \secondOrderAtypicalProbability{1})
\left(1+\frac{1}{\sqrt{\codebookBlocklength}}\right)
+
\frac{3}{\sqrt{\codebookBlocklength}}
\right) \\
\label{proof:soft-covering-two-transmitters-second-order-union-bound-application}
&
\begin{aligned}
  \leq
  &\Probability_{\codebookOne, \codebookTwo}\left(
    \totvarAtypicalOne > \secondOrderAtypicalProbability{1}\left(1+\frac{1}{\sqrt{\codebookBlocklength}}\right)
  \right)
  +
  \Probability_{\codebookOne, \codebookTwo}\left(
    \totvarAtypicalTwo > \secondOrderAtypicalProbability{2}\left(1+\frac{1}{\sqrt{\codebookBlocklength}}\right)
  \right)
  \\
  &+
  \sum\limits_{\channelOutAlphElement^\codebookBlocklength \in \channelOutAlph^\codebookBlocklength} \left(
    \Probability_{\codebookOne}\left(
      \codebookOne \notin \codebookSet_{\channelOutAlphElement^\codebookBlocklength}
    \right)
    +
    \Probability_{\codebookOne, \codebookTwo}\left(
      \totvarTypicalOne{\codebookTwoWord{\codewordIndex_2}}{\channelOutAlphElement^\codebookBlocklength}
      >
      1 + \frac{3}{\sqrt{\codebookBlocklength}}
      ~|~
      \codebookOne \in \codebookSet_{\channelOutAlphElement^\codebookBlocklength}
    \right)
  \right)
\end{aligned}
\\
\label{proof:soft-covering-two-transmitters-second-order-union-bound-substitutions}
&
\begin{aligned}
  \leq
  &\exp\left(
    -\frac{2\secondOrderAtypicalProbability{1}^2}
          {\codebookBlocklength}
    \exp(\codebookBlocklength \min(\codebookRateOne,\codebookRateTwo))
  \right)
  +
  \exp\left(
    -\frac{2\secondOrderAtypicalProbability{2}^2}
          {\codebookBlocklength}
    \exp(\codebookBlocklength \min(\codebookRateOne,\codebookRateTwo))
  \right)
  \\
  &+
  \cardinality{\channelInTwoAlph}^\codebookBlocklength \cardinality{\channelOutAlph}^\codebookBlocklength
  \exp\left(
    -\frac{1}{3\codebookBlocklength} \exp(-\codebookBlocklength (\mutualInformationConditional{\channelInOne}{\channelOut}{\channelInTwo} + \typicalityParam_1 - \codebookRateOne))
  \right)
  +
  \cardinality{\channelOutAlph}^\codebookBlocklength
  \exp\left(
    -\frac{1}{3\codebookBlocklength} \exp(-\codebookBlocklength (\mutualInformation{\channelInTwo}{\channelOut} + \typicalityParam_2 - \codebookRateTwo))
  \right)
\end{aligned}  
\\
\label{proof:soft-covering-two-transmitters-second-order-union-bound-end}
&\leq
2\exp\left(
  -\frac{2\min(\secondOrderAtypicalProbability{1}^2,\secondOrderAtypicalProbability{2}^2)}
        {\codebookBlocklength}
  \exp(\codebookBlocklength \min(\codebookRateOne,\codebookRateTwo))
\right)
+
2\exp\left(
  \codebookBlocklength(\log \cardinality{\channelOutAlph} + \log \cardinality{\channelInTwoAlph})
  -\frac{1}{3}
  \codebookBlocklength^{\secondOrderParamC - \secondOrderParamD - 1}
\right)
\end{align}
(\ref{proof:soft-covering-two-transmitters-second-order-union-bound-application}) results from~(\ref{proof:soft-covering-two-transmitters-typical-split}) and the union bound, (\ref{proof:soft-covering-two-transmitters-second-order-union-bound-substitutions}) follows by substituting (\ref{proof:soft-covering-two-transmitters-second-order-atypical-bound-1}), (\ref{proof:soft-covering-two-transmitters-second-order-atypical-bound-2}) and~(\ref{proof:soft-covering-two-transmitters-second-order-typical-bound}). (\ref{proof:soft-covering-two-transmitters-second-order-union-bound-end}) follows by substituting (\ref{theorem:soft-covering-two-transmitters-second-order-rate-one}), (\ref{theorem:soft-covering-two-transmitters-second-order-rate-two}) and (\ref{proof:soft-covering-two-transmitters-second-order-typicalityparam}), as well as elementary operations.
\newpage

\begin{align}
\label{converse-mutinf-bound-start}
\codebookBlocklength
\mutualInformationConditional{\hat{\channelInOne}, \hat{\channelInTwo}}{\hat{\channelOut}}{\hat{\timeSharingRV}}
-
\mutualInformation{\tilde{\channelInOne}^\codebookBlocklength, \tilde{\channelInTwo}^\codebookBlocklength}{\tilde{\channelOut}^\codebookBlocklength}
&=
\sum\limits_{\blockIndex=1}^{\codebookBlocklength}
\sum\limits_{\channelInOneAlphElement \in \channelInOneAlph}
\sum\limits_{\channelInTwoAlphElement \in \channelInTwoAlph}
\sum\limits_{\channelOutAlphElement \in \channelOutAlph}
  \codebookpmf_{\tilde{\channelInOne}_\blockIndex, \tilde{\channelInTwo}_\blockIndex, \tilde{\channelOut}_\blockIndex}(\channelInOneAlphElement, \channelInTwoAlphElement, \channelOutAlphElement)
  \log \frac{\channelpmf_{\channelOut | \channelInOne, \channelInTwo}(\channelOutAlphElement | \channelInOneAlphElement, \channelInTwoAlphElement)}
            {\codebookpmf_{\tilde{\channelOut}_\blockIndex}(\channelOutAlphElement)}
\nonumber
\\
&\hphantom{{}={}}
-
\sum\limits_{\channelInOneAlphElement^{\codebookBlocklength} \in \channelInOneAlph^{\codebookBlocklength}}
\sum\limits_{\channelInTwoAlphElement^{\codebookBlocklength} \in \channelInTwoAlph^{\codebookBlocklength}}
\sum\limits_{\channelOutAlphElement^{\codebookBlocklength} \in \channelOutAlph^{\codebookBlocklength}}
  \codebookpmf_{\tilde{\channelInOne}^{\codebookBlocklength}, \tilde{\channelInTwo}^{\codebookBlocklength}, \tilde{\channelOut}^{\codebookBlocklength}}(\channelInOneAlphElement^{\codebookBlocklength}, \channelInTwoAlphElement^{\codebookBlocklength}, \channelOutAlphElement^{\codebookBlocklength})
\log \frac{\channelpmf_{\channelOut^{\codebookBlocklength} | \channelInOne^{\codebookBlocklength}, \channelInTwo^{\codebookBlocklength}}(\channelOutAlphElement^{\codebookBlocklength} | \channelInOneAlphElement^{\codebookBlocklength}, \channelInTwoAlphElement^{\codebookBlocklength})}
            {\codebookpmf_{\tilde{\channelOut}^{\codebookBlocklength}}(\channelOutAlphElement^{\codebookBlocklength})}
\\
&=
\sum\limits_{\blockIndex=1}^{\codebookBlocklength}
\sum\limits_{\channelInOneAlphElement \in \channelInOneAlph}
\sum\limits_{\channelInTwoAlphElement \in \channelInTwoAlph}
\sum\limits_{\channelOutAlphElement \in \channelOutAlph}
  \codebookpmf_{\tilde{\channelInOne}_\blockIndex, \tilde{\channelInTwo}_\blockIndex, \tilde{\channelOut}_\blockIndex}(\channelInOneAlphElement, \channelInTwoAlphElement, \channelOutAlphElement)
  \log \channelpmf_{\channelOut | \channelInOne, \channelInTwo}(\channelOutAlphElement | \channelInOneAlphElement, \channelInTwoAlphElement)
+
\sum\limits_{\blockIndex=1}^{\codebookBlocklength}
  \entropy{\tilde{\channelOut}_\blockIndex}
\nonumber
\\
&\hphantom{{}={}}
-
\sum\limits_{\blockIndex=1}^{\codebookBlocklength}
\sum\limits_{\channelInOneAlphElement^{\codebookBlocklength} \in \channelInOneAlph^{\codebookBlocklength}}
\sum\limits_{\channelInTwoAlphElement^{\codebookBlocklength} \in \channelInTwoAlph^{\codebookBlocklength}}
\sum\limits_{\channelOutAlphElement^{\codebookBlocklength} \in \channelOutAlph^{\codebookBlocklength}}
  \codebookpmf_{\tilde{\channelInOne}^{\codebookBlocklength}, \tilde{\channelInTwo}^{\codebookBlocklength}, \tilde{\channelOut}^{\codebookBlocklength}}(\channelInOneAlphElement^{\codebookBlocklength}, \channelInTwoAlphElement^{\codebookBlocklength}, \channelOutAlphElement^{\codebookBlocklength})
\log \channelpmf_{\channelOut | \channelInOne, \channelInTwo}(\channelOutAlphElement_\blockIndex | \channelInOneAlphElement_\blockIndex, \channelInTwoAlphElement_\blockIndex)
-
\entropy{\tilde{\channelOut}^\codebookBlocklength}
\\
&=
\sum\limits_{\blockIndex=1}^{\codebookBlocklength}
  \entropy{\tilde{\channelOut}_\blockIndex}
-
\entropy{\tilde{\channelOut}^\codebookBlocklength}
\\
&\leq
\sum_{\blockIndex=1}^\codebookBlocklength
\entropy{\channelOut_\blockIndex}
-
\entropy{\channelOut^\codebookBlocklength}
+
\codebookBlocklength
\left(
  -
  \frac{1}{2}
  \lemmaconst
  \log
  \frac{\lemmaconst}{2\cardinality{\channelOutAlph}}
\right)
+
\left(
  -
  \frac{1}{2}
  \lemmaconst
  \log
  \frac{\lemmaconst}{2\cardinality{\channelOutAlph^\codebookBlocklength}}
\right)
\label{converse-mutinf-bound-lemma}
\\
&\leq
\codebookBlocklength
\left(
  -
  \lemmaconst
  \log
  \frac{\lemmaconst}{2\cardinality{\channelOutAlph}}
\right)
\label{converse-mutinf-bound-end}
\end{align}
(\ref{converse-mutinf-bound-lemma}) follows by (\ref{converse-lemma-tvassumption}), (\ref{converse-lemma-totvar-single}) and Lemma~\ref{totvar-entropy-lemma}.
\vspace{10pt}
\hrule

\begin{align}
\label{distinguishing-implies-semantic-begin}
&\hphantom{{}={}}
\max\limits_{\partition, \Probability_\messageRV} \bigg(
  \max\limits_{\wiretapperDecoder: \channelOutAlphWiretapper^\codebookBlocklength \rightarrow \partition}
  \Probability_{\messageRV, \channelOutWiretapper^\codebookBlocklength}(\messageRV \in \wiretapperDecoder(\channelOutWiretapper^\codebookBlocklength))
  -
  \max\limits_{\wiretapperGuesser \in \partition}
  \Probability_{\messageRV}(\messageRV \in \wiretapperGuesser)
\bigg)
\\
&
\begin{aligned}
=
\sum\limits_{(\messageAlphabetElement_1, \messageAlphabetElement_2) \in \messageAlphabet}
\sum\limits_{\channelOutAlphElementWiretapper^\codebookBlocklength \in \channelOutAlphWiretapper^\codebookBlocklength}
  \codebookpmf_\messageRV(\messageAlphabetElement_1, \messageAlphabetElement_2)
  \codebookpmf_{\channelOutWiretapper^\codebookBlocklength | \codebookOne(\messageAlphabetElement_1, \cdot), \codebookTwo(\messageAlphabetElement_2, \cdot)}(\channelOutAlphElementWiretapper^\codebookBlocklength)
  \cdot
  \indicator{(\messageAlphabetElement_1,\messageAlphabetElement_2) \in \wiretapperDecoder(\channelOutAlphElementWiretapper^\codebookBlocklength)}
\\
-
\sum\limits_{(\tilde{\messageAlphabetElement}_1, \tilde{\messageAlphabetElement}_2) \in \messageAlphabet}
\sum\limits_{\channelOutAlphElementWiretapper^\codebookBlocklength \in \channelOutAlphWiretapper^\codebookBlocklength}
  \codebookpmf_\messageRV(\tilde{\messageAlphabetElement}_1, \tilde{\messageAlphabetElement}_2)
  \codebookpmf_{\channelOutWiretapper^\codebookBlocklength | \codebookOne(\tilde{\messageAlphabetElement}_1, \cdot), \codebookTwo(\tilde{\messageAlphabetElement}_2, \cdot)}(\channelOutAlphElementWiretapper^\codebookBlocklength)
  \cdot
  \indicator{(\tilde{\messageAlphabetElement}_1,\tilde{\messageAlphabetElement}_2) \in \wiretapperGuesser}
\end{aligned}
\\
&
\begin{aligned}
=
\sum\limits_{(\messageAlphabetElement_1, \messageAlphabetElement_2), (\tilde{\messageAlphabetElement}_1, \tilde{\messageAlphabetElement}_2) \in \messageAlphabet}
\codebookpmf_\messageRV(\messageAlphabetElement_1, \messageAlphabetElement_2)
\codebookpmf_\messageRV(\tilde{\messageAlphabetElement}_1, \tilde{\messageAlphabetElement}_2)
\Bigg(
\sum\limits_{\channelOutAlphElementWiretapper^\codebookBlocklength \in \channelOutAlphWiretapper^\codebookBlocklength}
\codebookpmf_{\channelOutWiretapper^\codebookBlocklength | \codebookOne(\messageAlphabetElement_1, \cdot), \codebookTwo(\messageAlphabetElement_2, \cdot)}(\channelOutAlphElementWiretapper^\codebookBlocklength)
\indicator{(\messageAlphabetElement_1,\messageAlphabetElement_2) \in \wiretapperDecoder(\channelOutAlphElementWiretapper^\codebookBlocklength)}
\\
-
\sum\limits_{\channelOutAlphElementWiretapper^\codebookBlocklength \in \channelOutAlphWiretapper^\codebookBlocklength}
\codebookpmf_{\channelOutWiretapper^\codebookBlocklength | \codebookOne(\tilde{\messageAlphabetElement}_1, \cdot), \codebookTwo(\tilde{\messageAlphabetElement}_2, \cdot)}(\channelOutAlphElementWiretapper^\codebookBlocklength)
\cdot(
  \indicator{(\messageAlphabetElement_1,\messageAlphabetElement_2) \in \wiretapperDecoder(\channelOutAlphElementWiretapper^\codebookBlocklength)}
  +
  \indicator{(\messageAlphabetElement_1,\messageAlphabetElement_2) \notin \wiretapperDecoder(\channelOutAlphElementWiretapper^\codebookBlocklength)}
  -
  \indicator{(\tilde{\messageAlphabetElement}_1,\tilde{\messageAlphabetElement}_2) \notin \wiretapperGuesser}
)
\Bigg)
\end{aligned}
\\
&=
\semanticTotvarTerm + \semanticErrorTermOne - \semanticErrorTermTwo
\label{distinguishing-implies-semantic-end}
\end{align}

with the definitions

\begin{align}
\label{distinguishing-implies-semantic-part1}
&\begin{aligned}
\semanticTotvarTerm
:=
\sum\limits_{(\messageAlphabetElement_1, \messageAlphabetElement_2), (\tilde{\messageAlphabetElement}_1, \tilde{\messageAlphabetElement}_2) \in \messageAlphabet}
\codebookpmf_\messageRV(\messageAlphabetElement_1, \messageAlphabetElement_2)
\codebookpmf_\messageRV(\tilde{\messageAlphabetElement}_1, \tilde{\messageAlphabetElement}_2)
\sum\limits_{\channelOutAlphElementWiretapper^\codebookBlocklength \in \channelOutAlphWiretapper^\codebookBlocklength}
\Big(
  \codebookpmf_{\channelOutWiretapper^\codebookBlocklength | \codebookOne(\messageAlphabetElement_1, \cdot), \codebookTwo(\messageAlphabetElement_2, \cdot)}(\channelOutAlphElementWiretapper^\codebookBlocklength)
  -
  \codebookpmf_{\channelOutWiretapper^\codebookBlocklength | \codebookOne(\tilde{\messageAlphabetElement}_1, \cdot), \codebookTwo(\tilde{\messageAlphabetElement}_2, \cdot)}(\channelOutAlphElementWiretapper^\codebookBlocklength)
\Big)
\\
\cdot
\indicator{(\messageAlphabetElement_1,\messageAlphabetElement_2) \in \wiretapperDecoder(\channelOutAlphElementWiretapper^\codebookBlocklength)}
\end{aligned}
\\
\label{distinguishing-implies-semantic-part2}
&
\semanticErrorTermOne
:=
\sum\limits_{(\tilde{\messageAlphabetElement}_1, \tilde{\messageAlphabetElement}_2) \in \messageAlphabet}
  \codebookpmf_\messageRV(\tilde{\messageAlphabetElement}_1, \tilde{\messageAlphabetElement}_2)
  \cdot
  \sum\limits_{\channelOutAlphElementWiretapper^\codebookBlocklength \in \channelOutAlphWiretapper^\codebookBlocklength}
    \codebookpmf_{\channelOutWiretapper^\codebookBlocklength | \codebookOne(\tilde{\messageAlphabetElement}_1, \cdot), \codebookTwo(\tilde{\messageAlphabetElement}_2, \cdot)}(\channelOutAlphElementWiretapper^\codebookBlocklength)
    \indicator{(\tilde{\messageAlphabetElement}_1, \tilde{\messageAlphabetElement}_2) \notin \wiretapperGuesser}
\\
\label{distinguishing-implies-semantic-part3}
&
\semanticErrorTermTwo
:=
\sum\limits_{(\messageAlphabetElement_1, \messageAlphabetElement_2), (\tilde{\messageAlphabetElement}_1, \tilde{\messageAlphabetElement}_2) \in \messageAlphabet}
\codebookpmf_\messageRV(\messageAlphabetElement_1, \messageAlphabetElement_2)
\codebookpmf_\messageRV(\tilde{\messageAlphabetElement}_1, \tilde{\messageAlphabetElement}_2)
  \sum\limits_{\channelOutAlphElementWiretapper^\codebookBlocklength \in \channelOutAlphWiretapper^\codebookBlocklength}
    \codebookpmf_{\channelOutWiretapper^\codebookBlocklength | \codebookOne(\tilde{\messageAlphabetElement}_1, \cdot), \codebookTwo(\tilde{\messageAlphabetElement}_2, \cdot)}(\channelOutAlphElementWiretapper^\codebookBlocklength)
    \indicator{(\messageAlphabetElement_1, \messageAlphabetElement_2) \notin \wiretapperDecoder(\channelOutAlphElementWiretapper^\codebookBlocklength)}.
\end{align}

\end{document}